\documentclass{article}
\usepackage{amsmath}
\usepackage{amssymb}
\usepackage{amsthm}
\usepackage{array}
\usepackage{multirow}
\usepackage{mathrsfs}
\usepackage{siunitx}

\usepackage{graphicx}
\usepackage{caption}
\usepackage{subcaption}

\usepackage{booktabs}
\usepackage{tabularx}
\usepackage{longtable}

\usepackage{bm}
\usepackage{setspace}
\usepackage[title]{appendix}

\usepackage{float}
\usepackage{natbib}

\usepackage[colorlinks,
            linkcolor=black,
            anchorcolor=black,
            citecolor=black,
            urlcolor=blue,
            ]{hyperref}

\usepackage{cleveref}

\theoremstyle{plain}
\newtheorem{lemma}{Lemma}

\newtheorem{definition}{Definition}
\newtheorem{proposition}{Proposition}

\theoremstyle{definition}
\newtheorem{remark}{Remark}[section]

\makeatletter
\newcommand{\thickhline}{%
    \noalign {\ifnum 0=`}\fi \hrule height 1pt
    \futurelet \reserved@a \@xhline
}
\newcolumntype{"}{@{\hskip\tabcolsep\vrule width 1pt\hskip\tabcolsep}}
\makeatother

\begin{document}

\onehalfspacing        
\title{Integrating Different Informations for Portfolio Selection
\thanks{This work was partially supported by
National Natural Science Foundation of China (No.72271250, No.71471180), InnoHK initiative,
The Government of the HKSAR, and Laboratory for AI-Powered Financial Technologies, and
Research Foundation of Education Bureau of Hunan Province (No.22B0522).}}

\author{ Yi Huang\thanks{College of Mathematics and Statistics, Jishou University, Jishou, China. Email: huangy@jsu.edu.cn.}
\quad Wei Zhu\thanks{Huawei Technologies Co. Ltd, Shenzhen, China. Email: zhuweifly@163.com.}
\quad Duan Li\thanks{School of Data Science, City
University of Hong Kong, Hong Kong, China. Deceased.}
\quad Shushang Zhu\thanks{Corresponding author. School of Business, Sun Yat-Sen University, Guangzhou, China.
Email: zhuss@mail.sysu.edu.cn}
\quad Shikun Wang\thanks{School of Business, Sun Yat-Sen University, Guangzhou, China. Email: wangshk9@mail2.sysu.edu.cn.}
}

\maketitle

\begin{abstract}
Following the idea of Bayesian learning via Gaussian mixture model, we organically combine the backward-looking information contained in the historical data and the forward-looking information implied by the market portfolio, which is affected by heterogeneous expectations and noisy trading behavior. The proposed combined estimation adaptively harmonizes these two types of information based on the degree of market efficiency and responds quickly at turning points of the market.
Both simulation experiments and a global empirical test confirm that the approach is a flexible and robust forecasting tool and is applicable to various capital markets with different degrees of efficiency.
\end{abstract}


\section{Introduction}\label{Sec1}
As the beginning of modern investment theory, \cite{markowitz_portfolio_1952}
develops the mean-variance model to determine the optimal investment decisions, which not only characterizes the behavioral pattern of return-risk tradeoff but also illustrates the benefits of risk diversification in a quantitative way. Obviously, the successful application of the mean-variance model depends on the accurate estimation of the means and covariance matrix of asset returns. Existing studies tend to estimate model parameters by adopting different sources of information, which can be roughly divided into two types.

The first type is the well-developed backward-looking approach, which uses statistical approaches or econometric models to estimate expected parameters from historical data. This approach is a good choice when confronted with predicting future scenarios that resemble those found in historical patterns, or lack of other effective information \citep{Huang2021Combined}.
However, in practice, the distribution of the asset returns is time-varying, hence the backward-looking approach tends to exhibit lag near the turning point of the market \citep{zhu_portfolio_2014}.
As a result, the out-of-sample performance of backward-looking approaches is usually affected by the prediction error \citep{britten1999sampling, jobson1980estimation, ledoit2004well}.
Furthermore, the portfolios using parameters estimated with backward-looking approaches are usually extremely concentrated \citep{broadie_computing_1993, green1992will}
and especially sensitive to the means \citep{best1991sensitivity, merton1980estimating}.
The mean-variance optimization based on historical data only is even criticized as an error maximizer
\citep{michaud1989markowitz},
although there are some approaches proposed to alleviate the influence of estimation errors
(see e.g., \cite{black_global_1992, jagannathan2003risk, frost_for_1988}).

The second type is the forward-looking approach, which uses financial theories to speculate on investor expectations from market variables such as the market portfolio and option prices.
For the market portfolio, a typical work is the Black-Litterman model
\citep{black_global_1992},
where CAPM
\citep{sharpe_capital_1964}
is used to specify the means of the prior distribution on expected returns.
They assume that the market portfolio is the optimal decision of the representative agent with perfect rationality, thus implying investor's expectations on future returns. With the covariance matrix estimated from historical data, the expected return is derived from the first-order conditions of the mean-variance optimization problem, which is termed the inverse optimization approach in the sense that it recovers the expected returns (parameters input for optimization) from the market portfolio (decision as the result of optimization) according to the conditions of optimality.
Likewise,
\citep{sharpe2007expected} considers a more general framework for return forecasting by using the inverse optimization approach based on the expected utility maximization.

Another category of forward-looking approaches involves extracting return information from option prices, as risk-neutral quantities and the pricing kernel summarize investors' expectations and risk preferences.
\cite{bliss2004option}
extract the implied risk-neutral probability
\citep{breeden1978prices}
of the underlying asset from option prices and convert it into physical probability via the pricing kernel.
There are also findings that option-implied information helps to improve the out-of-sample performance of portfolios \citep{demiguel_improving_2013, kempf_portfolio_2015, kostakis_market_2011}.
However, it is notable that estimating the physical return distribution of underlying asset from option prices is only applicable for a single risky asset, while the recovery of joint physical return distribution of multiple underlying assets remains an open problem \citep{Carr2018}.

Both backward-looking and forward-looking approaches have their respective advantages and disadvantages in practice.
Despite concerns regarding its out-of-sample performance and theoretical foundations, the backward-looking approach remains a popular method for parameter estimations, particularly for the terms associated with risk, such as the covariance matrix.
The forward-looking approach can theoretically possess the ability to respond quickly at turning points \citep{zhu_portfolio_2014}.
However, its predictive ability is influenced by market efficiency and the degree to which financial theory matches the real market.
Especially, the forward-looking approaches usually impose assumptions of perfect rationality and homogeneous expectations among investors, while the impact of noise traders on market equilibrium always exists \citep{1986Noise, kyle_continuous_1985, 10.2307/2937765}.
As an inverse optimization approach, another challenging issue for the forward-looking approach is the nonuniqueness of the solution of the inverse problem.

Most of the existing literature tends to make predictions by solely utilizing one of the above-mentioned approaches.
Given the limitations of relying solely on backward-looking or forward-looking information, it is prudent to incorporate information from different sources and generate forecasts that are robust across a diverse range of markets.
\citep{cheang_optimal_2020}
integrate historical data and option-implied information in estimating the covariance matrix, which plugs risk-neutral variances in the sample covariance matrix directly.
\citep{Huang2021Combined}
consider a market with heterogeneous investors and propose an attractive approach to incorporate the market-implied information and historical data in expected return forecasting, where the covariance matrix remains the one estimated by historical data.
\begin{figure}[ht]
\centering
\includegraphics[height=3in,width=6in]{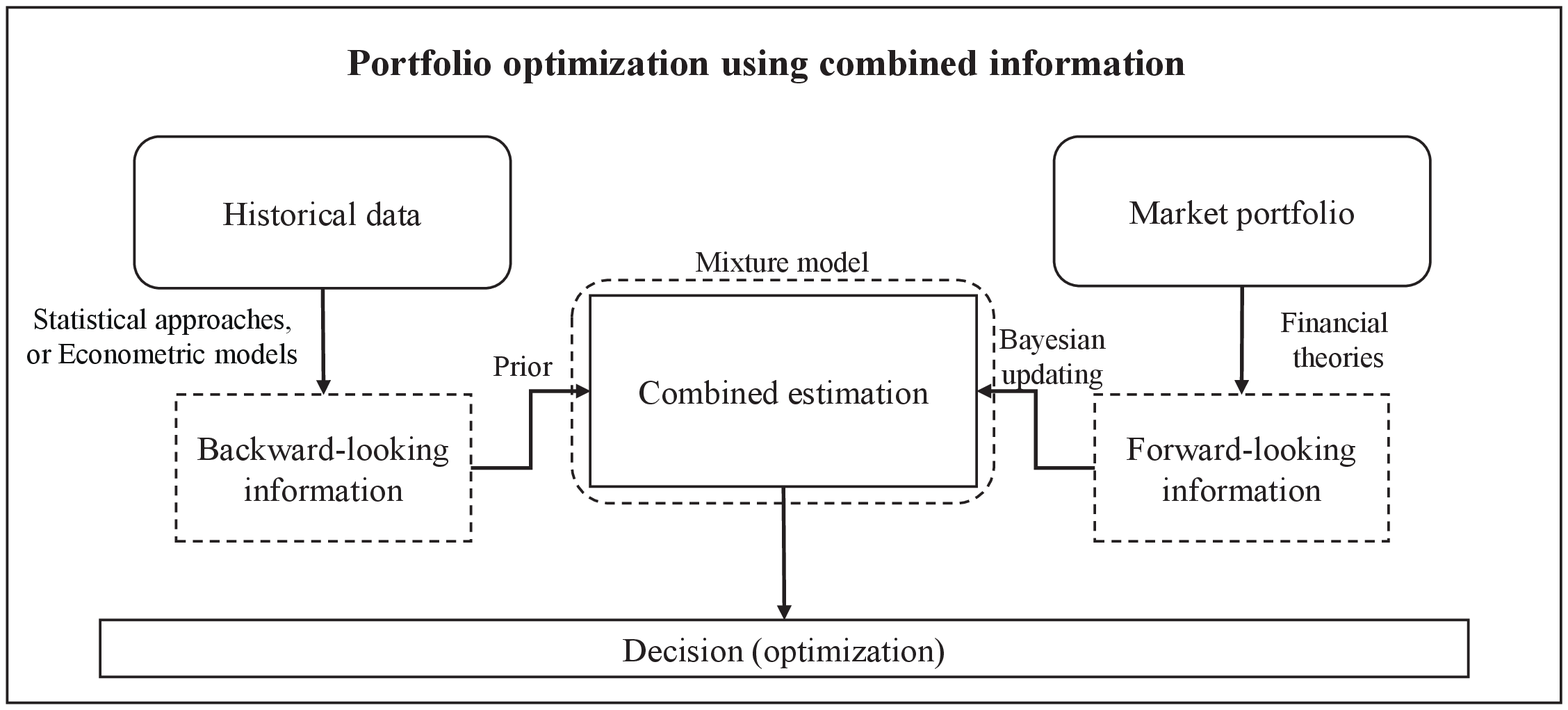}
\caption{Logic of problem formulation and research roadmap}
\label{fig:1}
\end{figure}

In this paper, we propose a general method that combines backward-looking and forward-looking information for predicting returns in portfolio management. Rather than assuming that investors are perfectly rational and have homogeneous expectations, we consider a more general market with heterogeneous investors and aim to extract the informed investors' expectations of returns from the market equilibrium influenced by the noise. Meanwhile, we provide a robust estimation of the means and covariance matrix of returns that combines the backward-looking information contained in historical data and the forward-looking information implied in the market portfolio. In addition, we intend to estimate the means and covariance matrix simultaneously, which, as an inverse problem, always suffers from the issue of multi-solution. To achieve these goals, we use the Gaussian mixture model to characterize uncertain returns and implement the estimation in the spirit of Bayesian learning. Figure \ref{fig:1} illustrates the logic and the research roadmap of our work. We theoretically show that the combined estimation can adaptively select the extent of forward-looking information to incorporate based on the degree of market efficiency, using the historical prior as a basis. As a result, it simultaneously possesses the advantages of both backward-looking and forward-looking forecasting approaches.

The remainder of the paper is organized as follows.
Section \ref{Sec2}
introduces a novel combined forecasting model and discusses some theoretical properties.
Section \ref{Sec3}
presents simulation experiments that verify the properties of the proposed combined estimation.
Section \ref{Sec4}
provides empirical analyses across 35 markets over the world and shows its out-of-sample performance.
Section \ref{Sec5}
concludes the study.

\section{Model}\label{Sec2}

In this section, we first model a general market considering investors' heterogeneous expectations and irrational behaviors. Next, we propose a novel combined approach to forecast asset returns and analyze its characteristics of self-adativity to the degree of market efficiency. Finally, we study a special case for the explicit form of the combined estimation.

\subsection{Market equilibrium with heterogeneous investors}\label{sec2.1}

In the classical CAPM model \citep{sharpe_capital_1964},
the market equilibrium is based on some ideal and strict assumptions, including homogeneous expectations and perfect rationality.
Instead, in this paper, we categorize investors into three types based on their available information and decision-making behaviors, namely a representative rational less-informed investor, a representative rational informed investor, and a representative noise trader.
The less-informed investor has only historical information, while the informed investor has access to a broad knowledge of asset returns. Both types of rational investors maximize the mean-variance utility based on their own information. The noise investor, as defined in
\cite{kyle_continuous_1985},
has no useful information and makes decisions randomly. In reality, we cannot classify all investors so clearly, as even an informed investor may behave irrationally.
However, we believe that this simple classification can effectively represent the three primary forces in the market.

Suppose there are $n$ risky assets and one risk-free asset available in the market.
Denote $\bm{r}=\left(\tilde{r}_1-r_f,\dots, \tilde{r}_n-r_f\right)'$ as the excess return of risky assets, where $\tilde{r}_i,i=1,\dots,n,$ and $r_f$ are the returns on the risky and risk-free assets, respectively. Let $\bm{x}=\left(x_1,\dots, x_n\right)'$ denote the portfolio weight on risky assets and $1-\sum_{i=1}^n x_i$ denote the weight on the risk-free asset. Both the less-informed investor and the informed investor make their investment decisions with the mean-variance model. Specifically, they determine their portfolios by solving the following optimization problem
\begin{equation}\label{eq:1}
 \mathop{\max}\limits_{\bm{x}} \bm{x}'\hat{\bm{\mu}}-\frac{\delta}{2}\bm{x}'\hat{\Sigma}\bm{x},
\end{equation}
where $\hat{\bm{\mu}}$ and $\hat{\Sigma}$ are the mean vector and covariance matrix of $\bm r$ estimated by investors and $\delta$ is the risk-aversion coefficient of investors.

Denote $\delta_U$ and $\delta_I$, $(\bm{\mu}_U,\Sigma_U)$ and $(\bm{\mu}_I,\Sigma_I)$ as the risk-aversion coefficients, mean vectors and covariance matrices specified by the less-informed investor and the informed investor, respectively. Accordingly, the optimal investment decisions of these two types of investors can be represented as
\begin{equation}\label{eq:2}
\bm{x}^*_U \triangleq \left(\delta_U\Sigma_U\right)^{-1}\bm{\mu}_U\
\mbox{and}~
\bm{x}_I^* \triangleq \left(\delta_I\Sigma_I\right)^{-1}\bm{\mu}_I.
\end{equation}
Furthermore, we assume that the informed investor has perfect knowledge on the return distribution and their predictions are equal to the true value of the parameters $(\bm{\mu}, \Sigma)$, i.e.,
\begin{equation}\label{eq:3}
\bm{\mu}_I = \bm{\mu},\ \Sigma_I = \Sigma\
\mbox{and}~
\bm{x}_I^* = \left(\delta_I\Sigma\right)^{-1}\bm{\mu}.
\end{equation}
In this paper, our main goal is to infer the informed investor's expectation  $(\bm{\mu},\Sigma)$.

By contrast, the noise investor makes decisions randomly, representing the irrational part of the market. Assume that their holding of risky assets follows a normal distribution
\begin{equation}\label{eq:4}
\bm{x}_N\sim\mathcal{N}\left(\bm{0},\Sigma_N\right),
\end{equation}
where $\bm{0}$ is a vector of zeros, $\Sigma_N = diag\left(\sigma_1,\dots,\sigma_n\right)$ is the covariance matrix, and $\{\sigma_i\}_{i = 1}^n$ are the intensities of noise trading, which implies that the random positions on risky assets are not correlated with each others.

Now we discuss the market equilibrium. Since the three investors borrow and lend equally on risk-free assets, the market portfolio is entirely allocated to risky assets.
Let $\bm{x}_M$ denote the market portfolio, and $w_U$, $w_I$, and $w_N$ denote the wealth of the less-informed investor, the informed investor, and the noise trader. The market clearing condition yields
\begin{equation}\label{eq:5}
w_U\bm{x}^*_U+w_I\bm{x}^*_I+w_N\bm{x}_N=w\bm{x}_M.
\end{equation}
Dividing both sides of
\eqref{eq:5}
by $w$ and substituting into
\eqref{eq:4},
we have
\begin{equation}\label{eq:6}
\bm{x}_M = \alpha_U\bm{x}^*_U+\alpha_I\left(\delta_I\Sigma\right)^{-1}\bm{\mu}+\bm{\varepsilon},
\end{equation}
where $\bm{\varepsilon} \triangleq \alpha_N\bm{x}_N$, and $\alpha_U
\triangleq w_U/w,~\alpha_I \triangleq w_I/w$ and $\alpha_N \triangleq w_N/w$ are the market shares of the less-informed investor, the informed investor, and the noise investor, respectively. Since $\bm{x}_N$ is normally distributed, the holding of the noise trader follows
\[
\bm{\varepsilon}\sim\mathcal{N}\left(\bm{0}_n,\Omega\right), ~
\Omega \triangleq \alpha_{N}^2\Sigma_N.
\]

Equation \eqref{eq:6}
contains significant forward-looking information about asset returns since the market portfolio $\bm{x}_M$ characterizes the market equilibrium and synthesizes investors' expectations and decision-making behaviors.
However, because of the random term $\bm{\varepsilon}$, we cannot directly derive the parameter $(\bm{\mu},\Sigma)$ from
\eqref{eq:6}.
In addition, even if the random term $\bm{\varepsilon}$ can be eliminated by adopting the expectation operator, the estimation is still difficult.
Although the less-informed investor's holding $\bm{x}^*_U$ can be determined by the estimate $(\bm{\mu}_U, \Sigma_U)$ using historical data and $\bm{x}_M$ can be directly observed, the number of unknown parameters in $(\bm{\mu},\Sigma)$ far exceeds the number of equations, which will result in several difficulties for estimation, for example, the issue of multi-solution.

To address the problem, we model the return distribution via a Gaussian mixture model (see e.g., \cite{zhu_portfolio_2014}).
More specifically, the probability density function of return $\bm{r}$ is defined by
\begin{equation}\label{eq:7}
p\left(\bm{r}\right)=\sum_{k=1}^m \lambda_k p_k\left(\bm{r}|\bm{\mu}_k,\Sigma_k\right),
\end{equation}
where $p_k\left(\bm{r}|\bm{\mu}_k,\Sigma_k\right) ~(k = 1, \cdots,m)$ is the density function of the $k^{th}$ Gaussian component with mean vector $\bm{\mu}_k$ and covariance matrix $\Sigma_k$, and the mixture coefficients $\{\lambda_k\}_{k=1}^m$ satisfy $\sum_{k=1}^m \lambda_k=1, \lambda_k\geq 0$. Accordingly, the mean and covariance of asset returns are given by
\begin{align}\label{eq:8}
\bm{\mu}=\sum\limits_{k=1}^m \lambda_k \bm{\mu}_k,\ \Sigma=\sum\limits_{k=1}^m \lambda_k \Sigma_k+\sum\limits_{k=1}^m \lambda_k\left(\bm{\mu}_k-\sum\limits_{k=1}^m \lambda_k \bm{\mu}_k\right)\left(\bm{\mu}_k-\sum\limits_{k=1}^m \lambda_k \bm{\mu}_k\right)'.
\end{align}
Since the sum of mixture coefficients is equal to one, without loss of generality, we regard $\lambda_m$ as a redundant parameter and denote the first $m-1$ mixture coefficients as $\bm{\lambda}_-=\left(\lambda_1,\ldots,\lambda_{m-1}\right)'$.

To reduce the number of parameters to be estimated, we reasonably suppose the Gaussian components $p_k\left(\bm{r}|\bm{\mu}_k,\Sigma_k\right) (k = 1, \cdots, m)$  are predetermined. Then we can transform the problem of estimating $(\bm{\mu},\Sigma)$ into the estimation of $\bm{\lambda}_-$. Rewrite \eqref{eq:6} with \eqref{eq:8}
\begin{equation}\label{eq:9}
\bm{x}_M = \alpha_U \bm{x}^*_U+\alpha_Ig\left(\bm{\lambda}_-\right)+\bm{\varepsilon},
\end{equation}
where
\begin{equation*}
\notag g\left(\bm{\lambda}_-\right) \triangleq
\frac{1}{\delta_I}\left[\sum_{k=1}^m \lambda_k\Sigma_k+
\sum_{k=1}^m \lambda_k\left(\bm{\mu}_k-\sum_{k=1}^m \lambda_k\bm{\mu}_k\right)
\left(\bm{\mu}_k-\sum_{k=1}^m \lambda_k\bm{\mu}_k\right)'\right]^{-1}
\sum_{k=1}^m \lambda_k\bm{\mu}_k.
\end{equation*}

The process of recovering $\bm{\lambda}_-$ from \eqref{eq:9} can be viewed as an inverse optimization problem, where $\bm{x}_I$ and $\bm{x}_U$ integrated into the observed market portfolio $\bm{x}_M$ are the result of optimization and the mixture weights $\bm{\lambda}_-$ is the parameter input for optimization.
Since the market portfolio contains the holding of the noise trader, we recover the forward-looking information of the informed investor from noise data.
Here, the Gaussian mixture model leads to a significant reduction in the number of unknown parameters, making the inverse optimization highly tractable.
Given mixture components, the model requires only $m-1$ mixture weights, and the number of unknown parameters is typically less than the number of assets. Furthermore, the utilization of the Gaussian mixture distribution in this model is reasonable both theoretically and empirically. As an unsupervised learning approach, the mixture distribution can approximate virtually any distribution
\citep{pang2023chance},
which makes it a suitable choice for fitting the distribution of financial data
\citep{mercuri_finite_2021}.

\subsection{Combined forecasting approach}\label{subsec2.2.2}

As shown in \eqref{eq:6}, the efficacy of forward-looking information is affected by noise trading. To achieve an estimation that is applicable across markets with varying market efficiency, we propose a novel combined approach to estimate the return distribution by integrating information from multiple channels in the spirit of Bayesian learning  (see e.g., \cite{marisu_bayesian_2023}).
Specifically, we use historical data to generate a prior distribution of mixture weights, then update with forward-looking information in \eqref{eq:9} to obtain the posterior distribution of mixture coefficients. We demonstrate the advantages of combined estimation later, specifically its ability to adaptively integrate different sources of information based on the degree of market efficiency. More specifically, if the degree of market efficiency is relatively low, the market portfolio does not contain valuable information on expectation, and the estimation is close to the prior distribution learned from the historical data. Otherwise, if the market is relatively efficient, then the estimation will reflect more of the market-implied information.

For the prior information, we adopt a Gaussian prior distribution for mixture coefficients
\[
\bm{\lambda}_- \sim \mathcal{N}\left(\hat{\bm{\lambda}}_-,\Phi\right),
\]
where $(\hat{\bm{\lambda}}_-,\Phi)$ can be estimated by existing technologies from historical data, such as the Expectation-Maximization (EM) algorithm \citep{mclachlan2007algorithm}.
The large sample theory on maximum likelihood estimation ensures the rationality of using normal distributions to model the distributions of the estimator of $\bm{\lambda}_-$ \citep{zhu_portfolio_2014}.
For the forward-looking information in \eqref{eq:9}, it is obvious that
\[
\bm{x}_M | \bm{\lambda}_- \sim \mathcal{N}\left(\alpha_U \bm{x}^*_U+\alpha_Ig\left(\bm{\lambda}_-\right),\Omega\right).
\]

Denote
\[
\bm{\theta} \triangleq (\alpha_I,\alpha_U,\alpha_N,\sigma_1,\cdots,\sigma_n,\delta_I,\delta_U)',
\]
the model parameter in the parameter space
\begin{equation}\label{para_space}
\Theta \triangleq \{ \bm{\theta} | (\alpha_I,\alpha_U,\alpha_N)'\bm{1} = 1, \bm{\theta} >0 \},
\end{equation}
and define $\Vert \bm y \Vert_P \triangleq \sqrt{\bm y'P\bm y}$. We assume that the market shares of three representative investors and noise intensities cannot reach zero to avoid the trivial case in which the forward-looking information is deterministic, and $\delta_I, \delta_U>0$ implies risk-averse investors.

Combining the prior information and forward-looking information mentioned above yields the following results on the posterior distribution of $\bm{\lambda}_-$.
\begin{proposition}\label{posterior}
Suppose that $\bm{\lambda}_- \sim\mathcal{N}\left(\hat{\bm{\lambda}}_-,\Phi\right)$, $\bm{x}_M | \bm{\lambda}_- \sim \mathcal{N}\left(\alpha_U \bm{x}^*_U+\alpha_Ig\left(\bm{\lambda}_-\right),\Omega\right)$ and both
$\Phi$ and $\Omega$ are invertible. Then the conditional distribution of $\bm{\lambda}_-$ for a given $\bm{x}_M $ satisfies
\begin{equation}\label{eq:14}
\notag p(\bm{\lambda}_- |\bm{x}_M ) \propto
exp\left\{ -\frac{1}{2}(F_1\left(\bm{\lambda}_-\right)+F_2\left(\bm{\lambda}_- | \bm{\theta} \right)) \right\},
\end{equation}
where
\begin{equation}\label{eq:11}
\begin{aligned}
F_1\left(\bm{\lambda}_-\right)
\triangleq
\left\Vert\bm{\lambda}_--\hat{\bm{\lambda}}_-\right\Vert^2_{\Phi^{-1}}\mbox{ and }
 F_2\left(\bm{\lambda}_- | \bm{\theta}\right)
\triangleq\left\Vert\frac{1}{\alpha_N}\bm{x}_M-\frac{\alpha_U}{\alpha_N}\bm{x}^*_U-
\frac{\alpha_I}{\alpha_N}g\left(\bm{\lambda}_-\right)\right\Vert^2_{\Sigma_N^{-1}}.
\end{aligned}
\end{equation}

\end{proposition}
\begin{proof}
With Bayes' rule, we can calculate the posterior probability density function as
\begin{align}
\notag p(\bm{\lambda}_- |\bm{x}_M )
\notag \propto& p(\bm{\lambda}_-)p(\bm{x}_M| \bm{\lambda}_- )\\
\notag \propto& exp\left\{-\frac{1}{2}\left(\bm{\lambda}_--\hat{\bm{\lambda}}_-\right)'\Phi^{-1}
\left(\bm{\lambda}_--\hat{\bm{\lambda}}_-\right) \right.\\
\notag &\left. -\frac{1}{2}\left(\bm{x}_M-\alpha_U \bm{x}^*_U-\alpha_Ig\left(\bm{\lambda}_-\right)\right)'\Omega^{-1}
\left(\bm{x}_M-\alpha_U \bm{x}^*_U-\alpha_Ig\left(\bm{\lambda}_-\right)\right)\right\} \\
\notag = & exp\left\{-\frac{1}{2}\left(\bm{\lambda}_--\hat{\bm{\lambda}}_-\right)'\Phi^{-1}
\left(\bm{\lambda}_--\hat{\bm{\lambda}}_-\right) \right.\\
\notag &\left. -\frac{1}{2}\left(\frac{1}{\alpha_N}\bm{x}_M-\frac{\alpha_U}{\alpha_N}\bm{x}^*_U-
\frac{\alpha_I}{\alpha_N}g\left(\bm{\lambda}_-\right)\right)'\Sigma_N^{-1}
\left(\frac{1}{\alpha_N}\bm{x}_M-\frac{\alpha_U}{\alpha_N}\bm{x}^*_U-
\frac{\alpha_I}{\alpha_N}g\left(\bm{\lambda}_-\right)\right)\right\},
\end{align}
by which we get Proposition
\ref{posterior}
with the definition of $F_1\left(\bm{\lambda}_-\right)$ and $F_2\left(\bm{\lambda}_- | \bm{\theta}\right)$ in \eqref{eq:11}.
\end{proof}

In \eqref{eq:11}, $F_1\left(\bm{\lambda}_-\right)$ is associated with backward-looking information, while $ F_2\left(\bm{\lambda}_- | \bm{\theta}\right)$ is a function of forward-looking information extracted from \eqref{eq:9}.
Here, although the variables we are interested in are mixture weights $\bm{\lambda}_-$, we add parameter $\bm{\theta}$ in $ F_2\left(\bm{\lambda}_- | \bm{\theta}\right)$ to emphasize that the forward-looking information is influenced by the three types of investors in the market, which are completely characterized by $\bm \theta$.

Although the explicit posterior distribution is given, it is unclear which distribution $\bm{\lambda}_- |\bm{x}_M$ follows due to the nonlinearity of $g\left(\bm{\lambda}_-\right)$. We propose a generalized combined forecasting model to maximize the posterior probability, which is equivalent to solving
\begin{equation}\label{eq:10}
\notag {\rm(P_C)}\
\min_{\bm{\lambda}_- \in \Gamma} F_1\left(\bm{\lambda}_-\right)+ F_2\left(\bm{\lambda}_- | \bm{\theta}\right),
\end{equation}
where
\begin{equation}\label{feasible_set}
\Gamma \triangleq\left\{\bm{\lambda}_-:\bm{1}'\bm{\lambda}_-\leq 1,\lambda_k\geq 0, k=1,\dots, m-1\right\}.\\
\end{equation}
Since $\Phi$ and $\Sigma_N$ are positive-definite matrices and $F_1\left(\bm{\lambda}_-\right)+F_2\left(\bm{\lambda}_- | \bm{\theta}\right)$ is a bounded continuous function over the feasible set $\Gamma$, there exists at least one optimal solution to problem ${\rm(P_C)}$.

\subsection{Characteristics of combined estimation}
\label{2.4}

In this section, we further analyze how the market shares, noise intensity, and the risk-aversion degree affect the combined estimation.
To do so, we define three types of estimations of mixture weights based on different information.
\begin{definition}\label{defin}
The estimation $\bm{\lambda}_-^B$ is called a backward-looking estimation if $\bm{\lambda}_-^B$ is an optimal solution to problem
\begin{equation}\label{backward}
\notag {\rm(P_B)}\ \mathop{\min}\limits_{\bm{\lambda}_-\in\Gamma} F_1\left(\bm{\lambda}_-\right);
\end{equation}
the estimation $\bm{\lambda}_-^F$ is called a forward-looking estimation if $\bm{\lambda}_-^F$ is a optimal solution to problem
\begin{equation}\label{forward}
\notag {\rm(P_F)}\ \mathop{\min}\limits_{\bm{\lambda}_-\in\Gamma} F_2\left(\bm{\lambda}_- | \bm{\theta}\right);
\end{equation}
and the estimation $\bm{\lambda}_-^C$ is called a combined estimation if $\bm{\lambda}_-^C$ is a optimal solution to problem
\begin{equation}\label{combined}
\notag {\rm(P_C)}\ \mathop{\min}\limits_{\bm{\lambda}_-\in\Gamma} F_1\left(\bm{\lambda}_-\right)+ F_2\left(\bm{\lambda}_- | \bm{\theta}\right).
\end{equation}
\end{definition}

We first provide the lemma below that is useful for further analysis. The proof of this lemma is based on the idea of
\cite{fukushima2001fundamentals}
(Chap. 3).

\begin{lemma}\label{para_opt}
Suppose that $h(\bm{\lambda}_-,\bm{\theta})$ is a continuous function on $\Gamma \times \Theta$, where $\Gamma$ is given in  \eqref{feasible_set} and $\Theta$ is given in \eqref{para_space}.
Define
\begin{equation}
\notag \psi( \bm{\theta}) \triangleq \inf \left\{ h(\bm{\lambda}_-,\bm{\theta}) | \bm{\lambda}_- \in \Gamma \right\}.
\end{equation}
Then, $\psi( \bm{\theta})$ is continuous.
\end{lemma}

\begin{proof}
Note that proving the continuity of $\psi( \bm{\theta})$ is equivalent to proving the lower and upper semi-continuity of $\psi( \bm{\theta})$. Specifically, we need to show that for any $\overline{\bm{\theta}} \in \Theta$ and sequence $\{\bm{\theta}^k\} \subseteq \Theta$, $\bm{\theta}^k \rightarrow \overline{\bm{\theta}}$, we have
\[
\psi(\overline{\bm{\theta}})
\leq
\liminf_{k \rightarrow \infty} \psi(\bm{\theta}^k)~
\text{and}~
\psi(\overline{\bm{\theta}})
\geq
\limsup_{k \rightarrow \infty} \psi(\bm{\theta}^k).
\]

First, we show that $\psi(\overline{\bm{\theta}})
\leq
\liminf\limits_{k \rightarrow \infty} \psi(\bm{\theta}^k)$. Since $h(\bm{\lambda}_-,\bm{\theta}) $ is continuous and $\Gamma$ is closed and bounded, for any $\{\bm{\theta}^k\}$, there exists sequence $\{\bm{\lambda}_-^k\}$ such that
\[
h(\bm{\lambda}_-^k,\bm{\theta}^k) = \psi(\bm{\theta}^k) .
\]
According to the definition of $\psi$, we have
\[
\psi(\overline{\bm{\theta}}) \leq
h(\bm{\lambda}_-^k,\overline{\bm{\theta}})
=
h(\bm{\lambda}_-^k,\overline{\bm{\theta}}) - h(\bm{\lambda}_-^k,\bm{\theta}^k) + \psi(\bm{\theta}^k).
\]
Furthermore, by the continuity of $h(\bm{\lambda}_-,\bm{\theta})$, we obtain
\[
\psi(\overline{\bm{\theta}}) \leq \liminf_{k \rightarrow \infty} \psi(\bm{\theta}^k),
\]
which gives the lower semi-continuity of $\psi(\bm{\theta})$ at $\overline{\bm{\theta}}$.

Now, we show that $\psi(\overline{\bm{\theta}})
\geq
\limsup\limits_{k \rightarrow \infty} \psi(\bm{\theta}^k)$. Since $h(\bm{\lambda}_-,\bm{\theta})$ is continuous on closed and bounded set $\Gamma$, for a given $\overline{\bm{\theta}}$, there exists $\overline{\bm{\lambda}}_- \in \Gamma$ such that
\[
h(\overline{\bm{\lambda}}_-,\overline{\bm{\theta}})
=
\psi(\overline{\bm{\theta}}) .
\]
Since $\Gamma$ is dense-in-itself, there exists a sequence $\{ \bm{\lambda}_-^k \} \subseteq  \Gamma$, $\bm{\lambda}_-^k \rightarrow \overline{\bm{\lambda}}_-$.  By the continuity of $h(\bm{\lambda}_-,\bm{\theta})$,
\[
\psi(\overline{\bm{\theta}}) = h(\overline{\bm{\lambda}}_-,\overline{\bm{\theta}})
=
\limsup_{k \rightarrow \infty} h(\bm{\lambda}_-^k,\bm{\theta}^k)
\geq
\limsup_{k \rightarrow \infty} \psi(\bm{\theta}^k),
\]
which gives the upper semi-continuity of $\psi(\bm{\theta})$ at $\overline{\bm{\theta}}$. Therefore, $\psi(\bm{\theta})$ is continuous at $\overline{\bm{\theta}}$.
\end{proof}

In the following analysis, we separately study the influence of three types of market parameters. For each type of parameters, we will isolate its effect by holding the other parameters constant. We first focus on the influence of market shares.

\begin{proposition}\label{prop2}
(i) In the market dominated by the noise investor, i.e., $\alpha_N\rightarrow 1^-$, if estimation $\bm{\lambda}_-^C$ solves ${\rm(P_C)}$, then $\bm{\lambda}_-^C$ solves ${\rm(P_B)}$. (ii) In the market dominated by the informed investor, i.e., $\alpha_I\rightarrow 1^-$, if estimation $\bm{\lambda}_-^C$ solves ${\rm(P_C)}$, then $\bm{\lambda}_-^C$ solves ${\rm(P_F)}$.
\end{proposition}

\begin{proof}
(i)
Note that proving the optimality of $\bm{\lambda}_-^C$ to problem ${\rm(P_B)}$ as $\alpha_N \to 1^-$ is equivalent to proving that $\lim\limits_{\alpha_N \to 1^-} F_1\left(\bm{\lambda}_-^C  \right)$ exists, and for any $\bm{\lambda}_-\in\Gamma$
\[
\lim\limits_{\alpha_N \to 1^-} F_1\left(\bm{\lambda}_-^C  \right) \leq \lim\limits_{\alpha_N \to 1^-} F_1\left(\bm{\lambda}_- \right).
\]

First, we demonstrate the existence of $\lim\limits_{\alpha_N \rightarrow 1^-} F_1\left(\bm{\lambda}_-^C  \right)$.
Since $\alpha_I,\alpha_U \rightarrow 0^+$ as $\alpha_N \rightarrow 1^-$ and $g(\bm{\lambda}_-^C )$ is bounded, we have
\[
\lim\limits_{\alpha_N \rightarrow 1^-} F_2\left(\bm{\lambda}_-^C | \bm{\theta} \right)
=
 \lim\limits_{\alpha_N \rightarrow 1^-} \left\Vert\frac{1}{\alpha_N}\bm{x}_M-\frac{\alpha_U}{\alpha_N}\bm{x}^*_U-
\frac{\alpha_I}{\alpha_N}g\left(\bm{\lambda}_-^C\right)\right\Vert^2_{\Sigma_N^{-1}}
=
\left\Vert \bm{x}_M \right\Vert^2_{\Sigma_N^{-1}},
\]
which yields the existence of $\lim\limits_{\alpha_N \rightarrow 1^-} F_2\left(\bm{\lambda}_-^C | \bm{\theta} \right)$.
Since the objective function in problem ${\rm(P_C)}$ is continuous and $\bm{\lambda}_-^C$ is the optimal solution,
Lemma \ref{para_opt}
provides the existence of $\lim\limits_{\alpha_N \rightarrow 1^-} \left\{F_1\left(\bm{\lambda}_-^C\right)+ F_2\left(\bm{\lambda}_-^C | \bm{\theta}\right)\right\}$. Taken together, we have the existence of $\lim\limits_{\alpha_N \rightarrow 1^-} F_1\left(\bm{\lambda}_-^C  \right)$ by the addition operation of limits.

Now, we demonstrate the optimality of $\bm{\lambda}_-^C$ to problem ${\rm(P_B)}$ as $\alpha_N \to 1^-$. $\forall \bm{\lambda}_-\in\Gamma$, we have
\begin{align}
\notag \lim_{\alpha_N \rightarrow 1^-} F_1\left(\bm{\lambda}_-^C\right) &=
\lim_{\alpha_N \rightarrow 1^-} F_1\left(\bm{\lambda}_-^C\right)
+ \left\Vert \bm{x}_M \right\Vert^2_{\Sigma_N^{-1}}
- \left\Vert \bm{x}_M \right\Vert^2_{\Sigma_N^{-1}}\\
\notag &=
\lim_{\alpha_N \rightarrow 1^-} F_1\left(\bm{\lambda}_-^C\right)
+ \lim_{\alpha_N \rightarrow 1^-} F_2\left(\bm{\lambda}_-^C | \bm{\theta}\right)
-  \lim_{\alpha_N \rightarrow 1^-} F_2\left(\bm{\lambda}_- | \bm{\theta}\right)\\
\notag &=
\lim_{\alpha_N \rightarrow 1^-} \left\{ F_1\left(\bm{\lambda}_-^C\right)
+  F_2\left(\bm{\lambda}_-^C | \bm{\theta}\right) \right\}
-  \lim_{\alpha_N \rightarrow 1^-} F_2\left(\bm{\lambda}_- | \bm{\theta}\right)\\
\notag &\leq
\lim_{\alpha_N \rightarrow 1^-} \left\{ F_1\left(\bm{\lambda}_-\right)
+  F_2\left(\bm{\lambda}_- | \bm{\theta}\right) \right\}
-  \lim_{\alpha_N \rightarrow 1^-} F_2\left(\bm{\lambda}_- | \bm{\theta}\right)\\
\notag &=
\lim_{\alpha_N \rightarrow 1^-} F_1\left(\bm{\lambda}_-\right) ,
\end{align}
where the second equality uses the boundedness of $g(\bm{\lambda}_-)$ and the fact that $ \lim\limits_{\alpha_N \rightarrow 1^-} F_2\left(\bm{\lambda}_- | \bm{\theta}\right) =  \left\Vert \bm{x}_M \right\Vert^2_{\Sigma_N^{-1}}$, and the inequality uses the optimality of $\bm{\lambda}_-^C$ to problem ${\rm(P_C)}$.
Therefore, $\bm{\lambda}_-^C$ is optimal to problem ${\rm(P_B)}$ as $\alpha_N \to 1^-$.

(ii)
Note that problem ${\rm(P_F)}$ is equivalent to
\[
\mathop{\min}_{\bm{\lambda}_-\in\Gamma}  \tilde{F}_2\left(\bm{\lambda}_-| \bm{\theta}\right),
\]
where $
\tilde{F}_2\left(\bm{\lambda}_-| \bm{\theta}\right) \triangleq
\left\Vert \bm{x}_M- \alpha_U\bm{x}^*_U-
\alpha_Ig\left(\bm{\lambda}_-\right)\right\Vert^2_{\Sigma_N^{-1}}$, and problem ${\rm(P_C)}$ is equivalent to
\[
\mathop{\min}_{\bm{\lambda}_-\in\Gamma} \left\{ \alpha_N^2 F_1\left(\bm{\lambda}_-\right)+ \tilde{F}_2\left(\bm{\lambda}_-| \bm{\theta}\right) \right\}.
\]

Proving the optimality of $\bm{\lambda}_-^C$ to problem ${\rm(P_F)}$ as $\alpha_I \to 1^-$ is equivalent to proving that $\lim\limits_{\alpha_I \to 1^-} \tilde{F}_2\left(\bm{\lambda}_-^C  \right)$ exists, and for any $\bm{\lambda}_-\in\Gamma$
\[
\lim\limits_{\alpha_I \to 1^-} \tilde{F}_2\left(\bm{\lambda}_-^C  \right)
\leq
\lim\limits_{\alpha_I \to 1^-} \tilde{F}_2\left(\bm{\lambda}_- \right).
\]

First, we demonstrate the existence of $\lim\limits_{\alpha_I \to 1^-} \tilde{F}_2\left(\bm{\lambda}_-^C\right)$. Since $\alpha_N,\alpha_U \rightarrow 0^+$ as $\alpha_I \rightarrow 1^-$ and $F_1\left(\bm{\lambda}^C_-\right)$ is bounded, we have that $\lim\limits_{\alpha_I \rightarrow 1^-} \alpha_N^2 F_1\left(\bm{\lambda}_-^C\right) = 0$; Since $\alpha_N^2 F_1\left(\bm{\lambda}^C_-\right)+ \tilde{F}_2\left(\bm{\lambda}^C_-| \bm{\theta}\right)$ is continuous and $\bm{\lambda}^C_-$ is the optimal solution,
Lemma \ref{para_opt}
shows that $\lim\limits_{\alpha_I \rightarrow 1}\left\{ \alpha_N^2 F_1\left(\bm{\lambda}^C_-\right)+ \tilde{F}_2\left(\bm{\lambda}^C_-| \bm{\theta}\right)\right\}$ exists.  Taken together, we have the existence of $\lim\limits_{\alpha_I \to 1^-} \tilde{F}_2\left(\bm{\lambda}_-^C\right)$ by the addition operation of limits.

Now, we demonstrate the optimality of $\bm{\lambda}_-^C$ to problem ${\rm(P_F)}$ as $\alpha_I \to 1^-$. $\forall \bm{\lambda}_-\in\Gamma$, we have
\begin{align}
\notag
\lim_{\alpha_I \rightarrow 1^-}  \tilde{F}_2\left(\bm{\lambda}_-^C| \bm{\theta}\right)
&=
\lim_{\alpha_I \rightarrow 1^-}  \tilde{F}_2\left(\bm{\lambda}_-^C| \bm{\theta}\right)
+ \lim_{\alpha_I \rightarrow 1^-}\alpha_N^2  F_1\left(\bm{\lambda}_-^C\right)
- \lim_{\alpha_I \rightarrow 1^-}\alpha_N^2  F_1\left(\bm{\lambda}_-\right) \\
\notag &= \lim_{\alpha_I \rightarrow 1^-}\left\{ \alpha_N^2  F_1\left(\bm{\lambda}_-^C\right) + \tilde{F}_2\left(\bm{\lambda}_-^C| \bm{\theta}\right)\right\}
- \lim_{\alpha_I \rightarrow 1^-}\alpha_N^2  F_1\left(\bm{\lambda}_-\right)
\\
\notag &\leq
\lim_{\alpha_I \rightarrow 1^-}\left\{ \alpha_N^2  F_1\left(\bm{\lambda}_-\right) + \tilde{F}_2\left(\bm{\lambda}_-| \bm{\theta}\right)\right\}
- \lim_{\alpha_I \rightarrow 1^-}\alpha_N^2  F_1\left(\bm{\lambda}_-\right)
\\
\notag &=
\lim_{\alpha_I \rightarrow 1^-}  \tilde{F}_2\left(\bm{\lambda}_-| \bm{\theta}\right),
\end{align}
where the first equality uses the boundedness of $F_1\left(\bm{\lambda}_-\right)$ and the fact that $\lim\limits_{\alpha_I \rightarrow 1^-}\alpha_N^2  F_1\left(\bm{\lambda}_-\right) = 0$, and the inequality uses the optimality of $\bm{\lambda}_-^C$ to problem ${\rm(P_C)}$.
Therefore, $\bm{\lambda}_-^C$ is optimal to problem ${\rm(P_F)}$ as $\alpha_I \to 1^-$.
\end{proof}

\begin{remark}
Proposition \ref{prop2}
does not provide the analysis of the market dominated by the less-informed investor.
Since $\alpha_I, \alpha_N\rightarrow 0^+$ as $\alpha_U\rightarrow 1^-$, $\forall \bm{\lambda}_- \in \Gamma$, the objective function of ${\rm(P_C)}$
\[
\lim_{\alpha_U\rightarrow 1^-} \left\{\alpha_N^2 F_1\left(\bm{\lambda}_-\right)+ \tilde{F}_2\left(\bm{\lambda}_-| \bm{\theta}\right)\right\}
= \left\Vert \bm{x}_M- \bm{x}^*_U\right\Vert^2_{\Sigma_N^{-1}},
\]
which reduces to a trivial case since the objective function is a constant independent of $\bm{\lambda}_-$ and the feasible set $\Gamma$ is the optimal solution set.

\end{remark}

\begin{remark}
If there is only the rational informed investor, as shown in case (ii) of
Proposition \ref{prop2} ,
the market equilibrium \eqref{eq:9} can be expressed as
\begin{equation*}
\lim_{\alpha_I\rightarrow 1^-} \bm{\mu}^C = \lim_{\alpha_I\rightarrow 1^-}  \bm{\mu}^F
=\notag \bm{\mu}
=\delta_I\Sigma\bm{x}_M,
\end{equation*}
which is the mean vector of the prior distribution of expected returns in \cite{black_global_1992}.
\end{remark}

We then focus on the influence of the noise trading intensity.
Denote $\sigma_{\min}=\min\left\{\sigma_1,\dots,\sigma_n\right\}$ and $\sigma_{\max}=\max\left\{\sigma_1,\dots,\sigma_n\right\}$ as the minimum and maximum noise trading intensity, respectively.

\begin{proposition}\label{prop3}
(i) In the market with infinite noise trading intensity, i.e., $\sigma_{\min}\rightarrow \infty$, if estimation $\bm{\lambda}_-^C$ solves ${\rm(P_C)}$, then $\bm{\lambda}_-^C$ solves ${\rm(P_B)}$.
(ii) In the market with zero noise trading intensity, i.e., $\sigma_{\max}\rightarrow 0^+$, if estimation $\bm{\lambda}_-^C$ solves ${\rm(P_C)}$, then $\bm{\lambda}_-^C$ solves ${\rm(P_F)}$.
\end{proposition}
\begin{proof}
(i)
Note that proving the optimality of $\bm{\lambda}_-^C$ to problem ${\rm(P_B)}$ as $\sigma_{\min}\rightarrow \infty$ is equivalent to proving that $\lim\limits_{\sigma_{\min}\rightarrow \infty} F_1\left(\bm{\lambda}_-^C  \right)$ exists, and for any $\bm{\lambda}_-\in\Gamma$
\[
\lim\limits_{\sigma_{\min}\rightarrow \infty} F_1\left(\bm{\lambda}_-^C  \right) \leq \lim\limits_{\sigma_{\min}\rightarrow \infty} F_1\left(\bm{\lambda}_- \right).
\]

First, we demonstrate the existence of $\lim\limits_{\sigma_{\min}\rightarrow \infty} F_1\left(\bm{\lambda}_-^C  \right)$.
Since $\sigma_i \rightarrow  \infty,~ i=1,\cdots,n$ as $\sigma_{\min}\rightarrow \infty$ and $\frac{1}{\alpha_N}\bm{x}_M-\frac{\alpha_U}{\alpha_N}\bm{x}^*_U-
\frac{\alpha_I}{\alpha_N}g\left(\bm{\lambda}_-\right)$ is bounded, we have
$\lim\limits_{\sigma_{\min}\rightarrow \infty} F_2\left(\bm{\lambda}_-^C| \bm{\theta}\right) = 0$;
Lemma \ref{para_opt}
demonstrates the existence of $\lim\limits_{\sigma_{\min}\rightarrow \infty} \left\{F_1\left(\bm{\lambda}_-^C\right)+ F_2\left(\bm{\lambda}_-^C | \bm{\theta}\right)\right\}$. Taken together, we have the existence of $\lim\limits_{\sigma_{\min}\rightarrow \infty} F_1\left(\bm{\lambda}_-^C  \right)$ by the addition operation of limits.

Now, we demonstrate the optimality of $\bm{\lambda}_-^C$ to problem ${\rm(P_B)}$ as $\sigma_{\min}\rightarrow \infty$. For any given $\bm{\lambda}_-\in\Gamma$, we have
\begin{align}
\notag \lim_{\sigma_{\min}\rightarrow \infty} F_1\left(\bm{\lambda}_-^C\right) &=
\lim_{\sigma_{\min}\rightarrow \infty} F_1\left(\bm{\lambda}_-^C\right)
+ \lim_{\sigma_{\min}\rightarrow \infty} F_2\left(\bm{\lambda}_-^C | \bm{\theta}\right)
-  \lim_{\sigma_{\min}\rightarrow \infty} F_2\left(\bm{\lambda}_- | \bm{\theta}\right)\\
\notag &=
\lim_{\sigma_{\min}\rightarrow \infty} \left\{ F_1\left(\bm{\lambda}_-^C\right)
+  F_2\left(\bm{\lambda}_-^C | \bm{\theta}\right) \right\}
-  \lim_{\sigma_{\min}\rightarrow \infty} F_2\left(\bm{\lambda}_- | \bm{\theta}\right)\\
\notag &\leq
\lim_{\sigma_{\min}\rightarrow \infty} \left\{ F_1\left(\bm{\lambda}_-\right)
+  F_2\left(\bm{\lambda}_- | \bm{\theta}\right) \right\}
-  \lim_{\sigma_{\min}\rightarrow \infty} F_2\left(\bm{\lambda}_- | \bm{\theta}\right)\\
\notag &=
\lim_{\sigma_{\min}\rightarrow \infty} F_1\left(\bm{\lambda}_-\right) ,
\end{align}
where the first equality uses the fact that $ \lim\limits_{\sigma_{\min}\rightarrow \infty} F_2\left(\bm{\lambda}_- | \bm{\theta}\right) =  0$, and the inequality uses the optimality of $\bm{\lambda}_-^C$ to problem ${\rm(P_C)}$.
Therefore, $\bm{\lambda}_-^C$ is optimal to problem ${\rm(P_B)}$ as $\sigma_{\min}\rightarrow \infty$.

(ii) Note that problem ${\rm(P_F)}$ is equivalent to
\[
 \mathop{\min}_{\bm{\lambda}_-\in\Gamma}  \hat{F}_2\left(\bm{\lambda}_-| \bm{\theta} \right),
\]
where
$\hat{F}_2\left(\bm{\lambda}_- | \bm{\theta}\right) \triangleq
\left\Vert \frac{1}{\alpha_N}\bm{x}_M-
\frac{\alpha_U}{\alpha_N}\bm{x}^*_U-
\frac{\alpha_I}{\alpha_N}
g\left(\bm{\lambda}_-\right)\right\Vert^2_{\tilde{\Sigma}_N^{-1}}$ with $\tilde{\Sigma}_N^{-1} \triangleq diag\left(\frac{\sigma_{\max}}{\sigma_1},\dots,\frac{\sigma_{\max}}{\sigma_n}\right)$, and problem ${\rm(P_C)}$ is equivalent to
\[
  \mathop{\min}_{\bm{\lambda}_-\in\Gamma} \left\{ \sigma_{\max} F_1\left(\bm{\lambda}_-\right)+ \hat{F}_2\left(\bm{\lambda}_-| \bm{\theta}\right) \right\}.
\]

Proving the optimality of $\bm{\lambda}_-^C$ to problem ${\rm(P_F)}$ as $\sigma_{\max} \to 0^+$ is equivalent to proving that $\lim\limits_{\sigma_{\max} \to 0^+} \hat{F}_2\left(\bm{\lambda}_-^C  \right)$ exists, and for any $\bm{\lambda}_-\in\Gamma$
\[
\lim\limits_{\sigma_{\max} \to 0^+} \hat{F}_2\left(\bm{\lambda}_-^C  \right)
\leq
\lim\limits_{\sigma_{\max} \to 0^+} \hat{F}_2\left(\bm{\lambda}_- \right).
\]

First, we demonstrate the existence of
$\lim\limits_{\sigma_{\max} \to 0^+} \hat{F}_2\left(\bm{\lambda}_-^C\right)$. Since  $F_1\left(\bm{\lambda}^C_-\right)$ is bounded, we have
$\lim\limits_{\sigma_{\max} \to 0^+} \sigma_{\max} F_1\left(\bm{\lambda}_-^C\right) = 0$. Since $\sigma_{\max} F_1\left(\bm{\lambda}^C_-\right)+ \hat{F}_2\left(\bm{\lambda}^C_-| \bm{\theta}\right)$ is continuous and $\bm{\lambda}^C_-$ is the optimal solution,
Lemma \ref{para_opt}
shows that $\lim\limits_{\sigma_{\max} \to 0^+}\left\{ \sigma_{\max} F_1\left(\bm{\lambda}^C_-\right)+ \hat{F}_2\left(\bm{\lambda}^C_-| \bm{\theta}\right)\right\}$ exists.  Taken together, we have the existence of $\lim\limits_{\sigma_{\max} \to 0^+} \hat{F}_2\left(\bm{\lambda}_-^C\right)$ by the addition operation of limits.

Now, we demonstrate the optimality of $\bm{\lambda}_-^C$ to problem ${\rm(P_F)}$ as $\sigma_{\max} \to 0^+$. $\forall \bm{\lambda}_-\in\Gamma$, we have
\begin{align}
\notag
\lim_{\sigma_{\max} \to 0^+}  \hat{F}_2\left(\bm{\lambda}_-^C| \bm{\theta}\right)
&=
\lim_{\sigma_{\max} \to 0^+}  \hat{F}_2\left(\bm{\lambda}_-^C| \bm{\theta}\right)
+ \lim_{\sigma_{\max} \to 0^+}\sigma_{\max}  F_1\left(\bm{\lambda}_-^C\right)
- \lim_{\sigma_{\max} \to 0^+}\sigma_{\max}  F_1\left(\bm{\lambda}_-\right) \\
\notag &= \lim_{\sigma_{\max} \to 0^+}\left\{ \sigma_{\max} F_1\left(\bm{\lambda}_-^C\right) + \hat{F}_2\left(\bm{\lambda}_-^C| \bm{\theta}\right)\right\}
- \lim_{\sigma_{\max} \to 0^+}\sigma_{\max} F_1\left(\bm{\lambda}_-\right)
\\
\notag &\leq
\lim_{\sigma_{\max} \to 0^+}\left\{ \sigma_{\max} F_1\left(\bm{\lambda}_-\right) + \hat{F}_2\left(\bm{\lambda}_-| \bm{\theta}\right)\right\}
- \lim_{\sigma_{\max} \to 0^+}\sigma_{\max}  F_1\left(\bm{\lambda}_-\right)
\\
\notag &=
\lim_{\sigma_{\max} \to 0^+}  \hat{F}_2\left(\bm{\lambda}_-| \bm{\theta}\right),
\end{align}
where the first equality uses the boundedness of $F_1\left(\bm{\lambda}_-\right)$ and the fact that $\lim\limits_{\sigma_{\max} \to 0^+}\sigma_{\max} F_1\left(\bm{\lambda}_-\right) = 0$, and the inequality uses the optimality of $\bm{\lambda}_-^C$ to problem ${\rm(P_C)}$.
Therefore, $\bm{\lambda}_-^C$ is optimal to problem ${\rm(P_F)}$ as $\sigma_{\max} \to 0^+$.
\end{proof}

Finally, we focus on the influence of the informed investor's risk-aversion degree.
\begin{proposition}\label{prop4}
(i) In the market where the informed investor is extremely conservative, i.e., $\delta_I\rightarrow \infty$,
if estimation $\bm{\lambda}_-^C$ solves ${\rm(P_C)}$, then $\bm{\lambda}_-^C$ solves ${\rm(P_B)}$.
(ii) In the market where the informed investor is extremely aggressive,  i.e., $\delta_I\rightarrow 0^+$, if estimation $\bm{\lambda}_-^C$ solves ${\rm(P_C)}$, then $\bm{\lambda}_-^C$ solves ${\rm(P_F)}$.
\end{proposition}
\begin{proof}
 For convenience, denote $\tilde{g}\left(\bm{\lambda}_-\right) \triangleq \delta_I g\left(\bm{\lambda}_-\right)$.

(i)
Note that proving the optimality of $\bm{\lambda}_-^C$ to problem ${\rm(P_B)}$ as $\delta_I\rightarrow \infty$ is equivalent to proving that $\lim\limits_{\delta_I \rightarrow \infty} F_1\left(\bm{\lambda}_-^C  \right)$ exists, and for any $\bm{\lambda}_-\in\Gamma$
\[
\lim\limits_{\delta_I \rightarrow \infty} F_1\left(\bm{\lambda}_-^C  \right) \leq \lim\limits_{\delta_I \rightarrow \infty} F_1\left(\bm{\lambda}_- \right).
\]

First, we demonstrate the existence of $\lim\limits_{\delta_I \rightarrow \infty} F_1\left(\bm{\lambda}_-^C  \right)$.
Since $\tilde{g}\left(\bm{\lambda}_-^C\right)$ is bounded, we have
\[
\lim\limits_{\delta_I\rightarrow \infty} F_2\left(\bm{\lambda}_-^C | \bm{\theta}\right)
=
\lim\limits_{\delta_I\rightarrow \infty} \left\Vert \frac{1}{\alpha_N}\bm{x}_M
- \frac{\alpha_U}{\alpha_N}\bm{x}^*_U
- \frac{1}{\delta_I}\frac{\alpha_I}{\alpha_N} \tilde{g}\left(\bm{\lambda}_-^C\right)
\right\Vert^2_{\Sigma_N^{-1}}
=  \left\Vert \frac{1}{\alpha_N}\bm{x}_M
- \frac{\alpha_U}{\alpha_N}\bm{x}^*_U
\right\Vert^2_{\Sigma_N^{-1}} ,
\]
Lemma\ref{para_opt}
demonstrates the existence of $\lim\limits_{\delta_I \rightarrow \infty} \left\{F_1\left(\bm{\lambda}_-^C\right)+ F_2\left(\bm{\lambda}_-^C | \bm{\theta}\right)\right\}$. Taken together, we have the existence of $\lim\limits_{\delta_I \rightarrow \infty} F_1\left(\bm{\lambda}_-^C  \right)$ by the addition operation of limits.

Now, we demonstrate the optimality of $\bm{\lambda}_-^C$ to problem ${\rm(P_B)}$ as $\delta_I \rightarrow \infty$. $\forall \bm{\lambda}_-\in\Gamma$, we have
\begin{align}
\notag \lim_{\delta_I \rightarrow \infty} F_1\left(\bm{\lambda}_-^C\right) &=
\lim_{\delta_I \rightarrow \infty} F_1\left(\bm{\lambda}_-^C\right)
+ \left\Vert \frac{1}{\alpha_N}\bm{x}_M
- \frac{\alpha_U}{\alpha_N}\bm{x}^*_U
\right\Vert^2_{\Sigma_N^{-1}}
-  \left\Vert \frac{1}{\alpha_N}\bm{x}_M
- \frac{\alpha_U}{\alpha_N}\bm{x}^*_U
\right\Vert^2_{\Sigma_N^{-1}}\\
\notag &=
\lim_{\delta_I \rightarrow \infty} F_1\left(\bm{\lambda}_-^C\right)
+ \lim_{\delta_I \rightarrow \infty} F_2\left(\bm{\lambda}_-^C | \bm{\theta}\right)
-  \lim_{\delta_I \rightarrow \infty} F_2\left(\bm{\lambda}_- | \bm{\theta}\right)\\
\notag &=
\lim_{\delta_I \rightarrow \infty} \left\{ F_1\left(\bm{\lambda}_-^C\right)
+  F_2\left(\bm{\lambda}_-^C | \bm{\theta}\right) \right\}
-  \lim_{\delta_I \rightarrow \infty} F_2\left(\bm{\lambda}_- | \bm{\theta}\right)\\
\notag &\leq
\lim_{\delta_I \rightarrow \infty} \left\{ F_1\left(\bm{\lambda}_-\right)
+  F_2\left(\bm{\lambda}_- | \bm{\theta}\right) \right\}
-  \lim_{\delta_I \rightarrow \infty} F_2\left(\bm{\lambda}_- | \bm{\theta}\right)\\
\notag &=
\lim_{\delta_I \rightarrow \infty} F_1\left(\bm{\lambda}_-\right) ,
\end{align}
where the first equality uses the boundedness of $\tilde{g}\left(\bm{\lambda}_-\right)$ and the fact that $ \lim\limits_{\sigma_{\min}\rightarrow \infty} F_2\left(\bm{\lambda}_- | \bm{\theta}\right) =  \left\Vert \frac{1}{\alpha_N}\bm{x}_M
- \frac{\alpha_U}{\alpha_N}\bm{x}^*_U
\right\Vert^2_{\Sigma_N^{-1}}$, and the inequality uses the optimality of $\bm{\lambda}_-^C$ to problem ${\rm(P_C)}$.
Therefore, $\bm{\lambda}_-^C$ is optimal to problem ${\rm(P_B)}$ as $\delta_I\rightarrow \infty$.

(ii) Note that problem ${\rm(P_F)}$ is equivalent to
\[
 \mathop{\min}_{\bm{\lambda}_-\in\Gamma}  \bar{F}_2\left(\bm{\lambda}_-| \bm{\theta}\right),
\]
where $\bar{F}_2\left(\bm{\lambda}_-| \bm{\theta}\right) \triangleq \left\Vert \delta_I \left(\frac{1}{\alpha_N}\bm{x}_M- \frac{\alpha_U}{\alpha_N}\bm{x}^*_U\right)-\frac{\alpha_I}{\alpha_N}\tilde{g}\left(\bm{\lambda}_-\right)\right\Vert^2_{\Sigma_N^{-1}}$, and problem ${\rm(P_C)}$ is equivalent to
\[
\mathop{\min}_{\bm{\lambda}_-\in\Gamma} \left\{ \delta_I^2 F_1\left(\bm{\lambda}_-\right)+ \bar{F}_2\left(\bm{\lambda}_-| \bm{\theta}\right) \right\}.
\]

Proving the optimality of $\bm{\lambda}_-^C$ to problem ${\rm(P_F)}$ as $\delta_I \to 0^+$ is equivalent to proving that $\lim\limits_{\delta_I \to 0^+} \bar{F}_2\left(\bm{\lambda}_-^C  \right)$ exists, and for any $\bm{\lambda}_-\in\Gamma$
\[
\lim\limits_{\delta_I \to 0^+} \bar{F}_2\left(\bm{\lambda}_-^C  \right)
\leq
\lim\limits_{\delta_I  \to 0^+} \bar{F}_2\left(\bm{\lambda}_- \right).
\]

First, we demonstrate the existence of $\lim\limits_{\delta_I \to 0^+} \bar{F}_2\left(\bm{\lambda}_-^C\right)$. Since  $F_1\left(\bm{\lambda}^C_-\right)$ is bounded, we have $\lim\limits_{\delta_I \to 0^+}  \delta_I^2 F_1\left(\bm{\lambda}_-^C\right) = 0$. Since $ \delta_I^2 F_1\left(\bm{\lambda}^C_-\right)+ \bar{F}_2\left(\bm{\lambda}^C_-| \bm{\theta}\right)$ is continuous and $\bm{\lambda}^C_-$ is the optimal solution,
$\lim\limits_{\delta_I \to 0^+}\left\{ \delta_I^2 F_1\left(\bm{\lambda}^C_-\right)+ \bar{F}_2\left(\bm{\lambda}^C_-| \bm{\theta}\right)\right\}$ exists by Lemma \ref{para_opt}.
Taken together, we have the existence of $\lim\limits_{\delta_I \to 0^+} \bar{F}_2\left(\bm{\lambda}_-^C\right)$ by the addition operation of limits.

Now, we demonstrate the optimality of $\bm{\lambda}_-^C$ to problem ${\rm(P_F)}$ as $\delta_I \to 0^+$. $\forall \bm{\lambda}_-\in\Gamma$, we have
\begin{align}
\notag
\lim_{\delta_I \to 0^+}  \bar{F}_2\left(\bm{\lambda}_-^C| \bm{\theta}\right)
&=
\lim_{\delta_I \to 0^+}  \bar{F}_2\left(\bm{\lambda}_-^C| \bm{\theta}\right)
+ \lim_{\delta_I \to 0^+}\delta_I^2  F_1\left(\bm{\lambda}_-^C\right)
- \lim_{\delta_I \to 0^+}\delta_I^2   F_1\left(\bm{\lambda}_-\right) \\
\notag &= \lim_{\delta_I \to 0^+}\left\{ \delta_I F_1\left(\bm{\lambda}_-^C\right) + \bar{F}_2\left(\bm{\lambda}_-^C| \bm{\theta}\right)\right\}
- \lim_{\delta_I \to 0^+}\delta_I^2  F_1\left(\bm{\lambda}_-\right)
\\
\notag &\leq
\lim_{\delta_I \to 0^+}\left\{ \delta_I^2  F_1\left(\bm{\lambda}_-\right) + \bar{F}_2\left(\bm{\lambda}_-| \bm{\theta}\right)\right\}
- \lim_{\delta_I \to 0^+}\delta_I^2  F_1\left(\bm{\lambda}_-\right)
\\
\notag &=
\lim_{\delta_I \to 0^+}  \bar{F}_2\left(\bm{\lambda}_-| \bm{\theta}\right),
\end{align}
where the first equality uses the boundedness of $F_1\left(\bm{\lambda}_-\right)$ and the fact that $\lim\limits_{\delta_I \to 0^+}\delta_I^2 F_1\left(\bm{\lambda}_-\right) = 0$, and the inequality uses the optimality of $\bm{\lambda}_-^C$ to problem ${\rm(P_C)}$.
Therefore, $\bm{\lambda}_-^C$ is optimal to problem ${\rm(P_F)}$ as $\delta_I \to 0^+$.
\end{proof}

\begin{remark}
If the risk-aversion coefficient of informed investors tends to zero, the combined estimation and forward-looking estimation of excess expected returns tend to zero, which corresponds to a risk-neutral world.
With \Cref{eq:9}, we have
\begin{equation*}
\lim_{\delta_I \rightarrow 0^+} \bm{\mu}^C = \lim_{\delta_I \rightarrow 0^+} \bm{\mu}^F
=\lim_{\delta_I \rightarrow 0^+}\frac{\delta_I}{\alpha_I}\Sigma
\left(\bm{x}_M-\alpha_U\bm{x}^*_U-\alpha_N\bm{x}_N\right)
=0.
\end{equation*}
\end{remark}

In summary,
Proposition \ref{prop2} to Proposition \ref{prop4}
indicate that the combined estimation is closer to the forward-looking estimation in markets where the market share of the informed investor is higher, the noise trading intensity is lower, or the risk aversion degree of informed traders is lower. In these markets, the informed investor plays a more dominant role, hence the market portfolio is more heavily influenced by the informed trader and efficient forward-looking information can be recovered. Conversely, in markets where the market portfolio is dominated by the less-informed investor and the noise investor, the forward-looking information is invalid hence the combined estimation tends to be closer to the backward-looking estimation.
Therefore,
Proposition \ref{prop2} to Proposition \ref{prop4}
demonstrate the self-adaptivity of the proposed combined estimation.

\subsection{Explicit form of combined estimation for a special case}\label{sec 2.2.3}
The portfolio decision is usually more sensitive to the mean vector than the covariance matrix. In this section, we mainly focus on mean estimation and assume that the covariance matrix can be exactly estimated by historical data, i.e., $\Sigma_U =\Sigma_I = \Sigma$. We also assume the homogeneous risk aversion for the informed and the less-informed investor, i.e., $\delta_U=\delta_I=\delta$.
 Denote
 \[
 \bm{q}_0 \triangleq \alpha_U\bm{x}_U^*+\alpha_I\left(\delta\Sigma\right)^{-1}\bm{\mu}_m \mbox{
 and } P \triangleq \alpha_I(\delta\Sigma)^{-1}\left(\bm{\mu}_1-\bm{\mu}_m,\dots,\bm{\mu}_{m-1}-\bm{\mu}_m\right).
 \]
Then, we can rewrite
\eqref{eq:9}
as
\begin{equation}\label{eq:12}
\notag \bm{x}_M=\bm{q}_0+P\bm{\lambda}_-+\bm{\varepsilon},
\end{equation}
and $F_1\left(\bm{\lambda}_-\right)$ and $F_2\left(\bm{\lambda}_- | \bm{\theta}\right)$ becomes
\[
F_1\left(\bm{\lambda}_-\right)
=
\left\Vert\bm{\lambda}_--\hat{\bm{\lambda}}_-\right\Vert^2_{\Phi^{-1}}\mbox{ and }
 F_2\left(\bm{\lambda}_- | \bm{\theta}\right)
=
\left\Vert \bm{x}_M - \bm{q}_0 - P\bm{\lambda}_- \right\Vert^2_{\Omega^{-1}}.
\]
For simplicity, denote
$$
G \triangleq \begin{pmatrix}
\bm{1}' \\
-\mathrm{I}_{m-1}
\end{pmatrix}_{m \times (m-1)}
\text{and} ~
h \triangleq \begin{pmatrix}
1 \\
\bm{0}
\end{pmatrix}_{m \times 1},
$$
where $\mathrm{I}_{m-1}$ is an $(m-1)\times(m-1)$ identical matrix. Problem ${\rm(P_C)}$ can be rewritten as
\begin{align}
\notag    \min_{\bm{\lambda}_-} &~
    \left\Vert\bm{\lambda}_--\hat{\bm{\lambda}}_-\right\Vert^2_{\Phi^{-1}}
    +
    \left\Vert \bm{x}_M - \bm{q}_0 - P\bm{\lambda}_- \right\Vert^2_{\Omega^{-1}} \\
\notag    \text{s.t.} &~ G\bm{\lambda}_-\leq h.
\end{align}

The following proposition gives the explicit form of the combined estimation.
\begin{proposition}\label{eq:prop1}
Suppose that $\bm{\lambda}_- \sim\mathcal{N}\left(\hat{\bm{\lambda}}_-,\Phi\right)$ and $\bm{x}_M \sim \mathcal{N}(\bm{q}_0+P\bm{\lambda}_-,\Omega)$ and both $\Phi$ and $\Omega$ are invertible. Then, there exists $\bm{\nu}^*\geq 0$, such that the combined estimation $\bm{\lambda}_-^C \in \Gamma$ takes the form
\[
\notag \bm{\lambda}_-^C = \left(\Phi^{-1}+P'\Omega^{-1}P\right)^{-1}
\left(\Phi^{-1}\hat{\bm{\lambda}}_-+P'\Omega^{-1}({\bm{x}_M}-\bm{q}_0)\right) -\frac12 \left( \left(\Phi^{-1}+P'\Omega^{-1}P\right)^{-1}G'\right)\bm{\nu}^*
\]
where $\bm{\nu}^*$ satisfies
\[
 {\bm{\nu}^*}'(G\bm{\lambda}_-^C - h) = 0.
\]
\end{proposition}

\begin{proof}
Note that
\begin{align}
\notag
F_1\left(\bm{\lambda}_-\right) + F_2\left(\bm{\lambda}_- | \bm{\theta}\right)
\propto
\left\Vert \bm{\lambda}_--\Lambda^{-1}\bm{\kappa} \right\Vert^2_{\Lambda},
\end{align}
where $\Lambda \triangleq \Phi^{-1}+P'\Omega^{-1}P$ and
$\bm{\kappa} \triangleq \Phi^{-1}\hat{\bm{\lambda}}_-+P'\Omega^{-1}\left({\bm{x}_M}-\bm{q}_0\right)$. Hence, problem ${\rm(P_C)}$ is equivalent to
\[
\min_{\bm{\lambda}_- \in \Gamma} \left\Vert \bm{\lambda}_--\Lambda^{-1}\bm{\kappa} \right\Vert^2_{\Lambda}
\]
which is a convex quadratic program, and the corresponding Lagrangian function is defined as
\[
\mathcal{L}(\bm{\lambda}_-, \bm{\nu}) \triangleq
\left\Vert \bm{\lambda}_--\Lambda^{-1}\bm{\kappa} \right\Vert^2_{\Lambda}
+
\bm{\nu}' (G\bm{\lambda}_- - h),
\]
where $\bm{\nu}$ is the Lagrangian multiplier vector for the inequality constraints. It is clear that there exists an interior point in feasible set $\Gamma$; hence, there exist $\bm{\lambda}_-^C \in \Gamma$ and $\bm{\nu}^* \geq 0$ satisfying the Karush–Kuhn–Tucker (KKT) conditions
\citep{boyd2004convex}
\begin{align}
\notag 2\Lambda(\bm{\lambda}_-^C - \Lambda^{-1}\bm{\kappa} ) + G'\bm{\nu}^* = 0,\\
\notag {\bm{\nu}^*}'(G\bm{\lambda}_-^C - h) = 0.
\end{align}
By the convexity of problem ${\rm(P_C)}$, $\bm{\lambda}_-^C$ is a optimal solution.
\end{proof}

\Cref{eq:prop1} shows that if $\bm{\lambda}_-^C = \left(\Phi^{-1}+P'\Omega^{-1}P\right)^{-1}
\left(\Phi^{-1}\hat{\bm{\lambda}}_-+P'\Omega^{-1}({\bm{x}_M}-\bm{q}_0)\right)$ is an interior point of $\Gamma$, then we have $\bm{\nu}^* = \bm{0}$ from the  complementary slackness and $\bm{\lambda}_-^C$ is optimal to problem ${\rm(P_C)}$. Therefore, we can also solve the combined estimation by Bayesian estimation. With the prior distribution and forward-looking information, the posterior distribution is still a normal distribution
\[
\bm{\lambda}_-|{\bm{x}_M} \sim\mathcal{N}
\left(\left(\Phi^{-1}+P'\Omega^{-1}P\right)^{-1}
\left(\Phi^{-1}\hat{\bm{\lambda}}_-+P'\Omega^{-1}({\bm{x}_M}-\bm{q}_0)\right),\left(\Phi^{-1}+P'\Omega^{-1}P\right)^{-1}\right).
\]
The proof is given in the Appendix.
Since the mean and mode of a normal distribution are equal, the mean of the posterior distribution is the same as the maximum a posterior estimation in problem ${\rm(P_C)}$.

Notably, the mathematical expression for the posterior mean
\[
\left(\Phi^{-1}+P'\Omega^{-1}P\right)^{-1}
\left(\Phi^{-1}\hat{\bm{\lambda}}_-+P'\Omega^{-1}({\bm{x}_M}-\bm{q}_0)\right)
\]
is the same as in \cite{black_global_1992}
and
\cite{zhu_portfolio_2014}
in form. Here, we highlight the differences between this work and theirs. In the Black-Litterman model, the prior distribution and investor views  focus directly on expected returns, while the forward-looking information in this paper is indirectly related to expected returns and the covariance matrix through mixture weights. Additionally, in extracting  market-implied information, we consider a more general equilibrium with three types of investors, while they simply impose assumptions of perfect rationality and homogeneous expectation. Compared to
\cite{zhu_portfolio_2014},
who also model the asset returns with the Gaussian mixture distribution, the present model has fundamental differences in the information for estimating the mixture weights. Their approach focuses on incorporating investors' views, while we study the market equilibrium to extract market-implied information.

\begin{remark}
Recall that we analyze how model parameters influence the combined estimation in
subsection \ref{2.4},
here, we verify
Proposition \ref{prop2} to Proposition \ref{prop4}
in an explicit form.
Take the influence of noise intensity $\sigma_{\min}$ as an example.

If $\bm{\lambda}_-^C = \left(\Phi^{-1}+P'\Omega^{-1}P\right)^{-1}\left(\Phi^{-1}\hat{\bm{\lambda}}_-+P'\Omega^{-1}({\bm{x}_M}-\bm{q}_0)\right)$
is optimal to problem ${\rm(P_C)}$, as $\sigma_{\min}$ tends to infinity, we have
\begin{align}
\notag \lim_{\sigma_{\min} \to \infty} \bm{\lambda}_-^C
&=
\lim_{\sigma_{\min} \to \infty}
 \left(\Phi^{-1}+\alpha_{N}^{-2}P'\Sigma_N^{-1}P\right)^{-1}\left(\Phi^{-1}\hat{\bm{\lambda}}_- + \alpha_{N}^{-2}P'\Sigma_N^{-1}({\bm{x}_M}-\bm{q}_0)\right) = \Phi \Phi^{-1} \hat{\bm{\lambda}} = \hat{\bm{\lambda}},
\end{align}
which is optimal to problem ${\rm(P_B)}$, consistent with
Proposition \ref{prop3}.
For this special case, the influences of other parameters can be also verified explicitly, and we can easily observe the self-adaptivity of the combined estimation.
\end{remark}

\section{Simulations}\label{Sec3}

In this section, we perform Monte Carlo simulations to  confirm the self-adaptive capabilities of the combined estimation discussed in subsection \ref{2.4} and show the significance of information integration in forecasting.
To do so, we first introduce the parameter settings and sample generation, then compare the predictive effects of the three types of estimations.

Return samples are generated based on the Gaussian mixture distribution, where different components can be interpreted as different market states and mixture weights are the probabilities of the corresponding states.
Following
\cite{zhu_portfolio_2014},
we take $m=3$ components, representing the bull, oscillator, and bear markets, and $n=10$ risky assets in the experiments.
The total length of the experiment period is 400 days, divided into four segments of length 200, 80, 50, and 70
days.
The mixture coefficients of the four segments are described as
$\bm{\lambda}_1=\left(0,0.5,0.5\right)'$, $\bm{\lambda}_2=\left(0.7,0.3,0\right)'$, $\bm{\lambda}_3=\left(0,0.5,0.5\right)'$ and $\bm{\lambda}_4=\left(0.2,0.8,0\right)'$,
respectively.\footnote{For example, $\bm{\lambda}_1$ implies a 50/50 chance of an oscillator or a bear market occurring in the first segment.} We begin our experiment on day 111, and the three representative investors rebalance their holdings every 5 days.

The Gaussian components are specified according to the single-factor model. Specifically, the excess returns $\bm{r}_k$ of the $k^{th}$ component are
\begin{equation}\label{eq:25}
\notag \bm{r}_k=\bm{\alpha}_k+\bm{\beta}_k r^M_k+\bm{\varepsilon}_k,\ k=1,2,3
\end{equation}
where $r^M_k$ is the excess return of the market index (factor), $\bm{\alpha}_k$ is the return that cannot be explained by the factor,
$\bm{\beta}_k$ is the sensitivity to the factor and
$\bm{\varepsilon}_k$ is the error term, following a normal distribution and satisfying
$\mathbb {E}\left(\bm{\varepsilon}_k\right)=\bm{0}$, $\mathbb {C}ov\left(r^M_k,\varepsilon_k^i\right)=\mathbb {C}ov\left(\varepsilon_k^i,\varepsilon_k^j\right)=0,\  i,j=1,\dots,n$.
The mean vector and the covariance matrix of the $k^{th}$ component  are given by
\begin{equation}\label{eq:26}
\notag \bm{\mu}_k=\bm{\alpha}_k+\mathbb{E}\left(r^M_k\right)\bm{\beta}_k,\
\Sigma_k=\mathbb{V}ar\left(r^M_k\right)\bm{\beta}_k\bm{\beta}_k'+
\mathbb{V}ar\left(\bm{\varepsilon}_k\right),\ k=1,2,3.
\end{equation}
We set the parameters of mean $\mathbb{E}\left(r^M_k\right)$ and variance $\mathbb{V}ar\left(r^M_k\right), k = 1,2,3$ as $(0.004,0,-0.008)$ and $(0.001,0.0002,0.002)$ for the bull, oscillator, and bear markets.
Elements of $\bm{\alpha}_k$ are the samples of $10^{-5}\tau$, and elements of $\bm{\beta}_k$ are the samples of $1.2-0.6\tau$, where $\tau\sim\mathcal{N}\left(0,1\right)$. The diagonal elements of $\mathbb{V}ar\left(\bm{\varepsilon}_k\right)$ are the samples of $0.002\tilde{\tau}$, where $\tilde{\tau}\sim\mathcal{U}\left(0,1\right)$. The above-mentioned $\tau$ and $\tilde{\tau}$ are independent of each other. With mixture coefficients and the density function \Cref{eq:7}, we can generate samples of asset returns in four investment segments. The risk-free return is set as $r_f=3.5\%$.

Then, we generate samples of the market portfolio. According to \eqref{eq:6}, the market portfolio is a linear combination of the holdings of the three types of investors.
The informed investor knows the true distribution of asset returns, while the less-informed investor estimates the mixture coefficients by using the EM algorithm
\citep{mclachlan2007algorithm}
with a $\Delta t = 30$-day rolling estimation window.
Both informed and less-informed investors determine their portfolios by solving \eqref{eq:2}.
In addition, the holdings of the noise investor are drawn from $\mathcal{N}(\bm{0},\Sigma_N)$, $\Sigma_N = diag(\sigma,\cdots,\sigma)$. The noise-trading intensity $\sigma$ is chosen according to the $1- c $ confidence interval  $\left[z_{\frac{c}{2}}\sigma,z_{1-\frac{c}{2}}\sigma\right]$, where $z_{c}$ is the $c$ quantile of the standard normal distribution.

To verify the characteristics in subsection \ref{2.4},
we conduct the following experiments to compare the three types of forecast approaches. To measure estimation errors, root mean square errors (RMSE) for turning points for the market and the whole investment period are calculated
\begin{equation}
\notag \eta_i=\frac{1}{k}\sum_{k=1}^m\sqrt{\frac{1}{N}\sum_{t=1}^{N}
\left(\hat{\lambda}_{k,t}^i-\lambda_{k,t}\right)^2},\ i=B,C,F
\end{equation}
where $\lambda_{k,t}$ is the true mixture weight of the $k^{th}$ mixture coefficient at time $t$, $\hat{\lambda}^i_{k,t}$ is the $i$ estimation, $N$ is the number of samples in the test period, B, C, and F represent backward-looking, combined, and forward-looking approaches, respectively.

{\bf Experiment 1: market shares.} In this experiment, we vary the market share of the informed investor at time 0 from 0.001 to 0.990. In addition, we set the risk-aversion coefficients of the informed investor and the less-informed investor as $\delta_I= \delta_U= 2.500$ and set the $95\%$ confidence interval of the holding of the noise investor as $\left[-1,1\right]$.

{\bf Experiment 2: noise-trading intensity.} In this experiment, we vary the noise intensity $\sigma$ such that the $95\%$ confidence interval of the holding of the noise investor range from
$\left[-0.001, 0.001 \right]$ to $\left[-100,100\right]$
In addition, we set the informed investor's market share at time 0 as $\alpha_{I,0}=0.400$ and set the risk-aversion coefficients of the informed investor and the less-informed investor as $\delta_I= \delta_U= 2.500$.

{\bf Experiment 3: risk-aversion attitude.} In this experiment, we vary the risk-aversion coefficient of the informed investor from \num{2.777e-3} to \num{2.250e3}
In addition, we set the informed investor's market share at time 0 as $\alpha_{I,0}=0.400$, set the $95\%$ confidence interval of the holding of the noise investor as $\left[-1,1\right]$, and set the risk-aversion coefficient of the less-informed investor as $\delta_U=2.500$.

\begin{table}[htbp]
  \centering
  \captionsetup{labelsep = period, labelfont = bf, font = {stretch=1}}
  \caption{Simulation tests: Comparisons of backward-looking, combined, and forward-looking estimations.
  We report the estimation errors of three types of estimations at the turning points and over the investment periods, where $\eta_B,\eta_C,\eta_F$ denote the RMSE for the backward-looking, combined, and forward-looking approaches, respectively.}\label{table:1}
    \begin{tabular}{ccccccccccc}
    \hline
\multicolumn{3}{c}{\multirow{2}[4]{*}{Influencing factors}} &       & \multicolumn{7}{c}{RMSE} \\
\cmidrule{5-11}    \multicolumn{3}{r}{}  &       & \multicolumn{3}{c}{Turning points} &       & \multicolumn{3}{c}{Investment periods} \\
\cmidrule{1-11}
\multicolumn{6}{l}{\underline{Experiment 1:  market shares}} \\
    $\alpha_{I,0}$ & $\alpha_{U,0}$ & $\alpha_{N,0}$ &       & $\eta_B$  & $\eta_C$  & $\eta_F$  &       &  $\eta_B$  & $\eta_C$  & $\eta_F$ \\
\hline
    0.001  & 0.009  & 0.990  &       & \textbf{0.421} & \textbf{0.420} & 0.540  &       & \textbf{0.168} & \textbf{0.168} & 0.570  \\
    0.100  & 0.800  & 0.100  &       & 0.421  & \textbf{0.077}  & 0.140  &       & 0.168  & \textbf{0.049}  & 0.378  \\
    0.200  & 0.700  & 0.100  &       & 0.421  & \textbf{0.037}  & 0.060  &       & 0.168  & \textbf{0.038}  & 0.233  \\
    0.300  & 0.600  & 0.100  &       & 0.421  & \textbf{0.014}  & 0.029  &       & 0.168  & \textbf{0.031}  & 0.145  \\
    0.400  & 0.500  & 0.100  &       & 0.421  & \textbf{0.004}  & 0.018  &       & 0.168  & \textbf{0.025}  & 0.055  \\
    0.500  & 0.400  & 0.100  &       & 0.421  & \textbf{0.000}  & 0.014  &       & 0.168  & \textbf{0.022}  & 0.047  \\
    0.600  & 0.300  & 0.100  &       & 0.421  & \textbf{0.004}  & 0.012  &       & 0.168  & \textbf{0.020}  & 0.040  \\
    0.700  & 0.200  & 0.100  &       & 0.421  & \textbf{0.004}  & 0.008  &       & 0.168  & \textbf{0.018}  & 0.036  \\
    0.800  & 0.100  & 0.100  &       & 0.421  & \textbf{0.004}  & 0.008  &       & 0.168  & \textbf{0.017}  & 0.032  \\
    0.990  & 0.009  & 0.001  &       & 0.421  & \textbf{0.000} & \textbf{0.000} &       & 0.168  & \textbf{0.000} & \textbf{0.000} \\
\hline
\multicolumn{6}{l}{\underline{Experiment 2: noise-trading intensity}} \\%
    \multicolumn{3}{c}{$\sigma$} &       & $\eta_B$  & $\eta_C$  & $\eta_F$  &       & $\eta_B$  & $\eta_C$  & $\eta_F$ \\
\hline
  \multicolumn{3}{c}{\num{2.603e3}} &       & \textbf{0.421} & \textbf{0.420} & 0.619  &       & \textbf{0.168} & \textbf{0.169} & 0.625  \\
    \multicolumn{3}{c}{\num{2.603e1}} &       & 0.421  & \textbf{0.417}  & 0.619  &       & \textbf{0.168}  & \textbf{0.168}  & 0.602  \\
    \multicolumn{3}{c}{\num{2.603e-1}} &       & 0.421  & \textbf{0.004}  & 0.018  &       & 0.168  & \textbf{0.025}  & 0.055  \\
    \multicolumn{3}{c}{\num{2.603e-3}} &       & 0.421  & \textbf{0.007}  & 0.017  &       & 0.168  & \textbf{0.019}  & \textbf{0.019}  \\
    \multicolumn{3}{c}{\num{2.603e-5}} &       & 0.421  & \textbf{0.000} & \textbf{0.000} &       & 0.168  & \textbf{0.001} & \textbf{0.001} \\
\hline
\multicolumn{6}{l}{\underline{Experiment 3: risk-aversion attitude}} \\ %
    \multicolumn{3}{c}{$\delta_I$} &       & $\eta_B$  & $\eta_C$  & $\eta_F$  &       & $\eta_B$  & $\eta_C$  & $\eta_F$ \\
\hline
    \multicolumn{3}{c}{\num{2.250e3}} &       & \textbf{0.421} & \textbf{0.420} & 0.540  &       & \textbf{0.168} & \textbf{0.168} & 0.570  \\
    \multicolumn{3}{c}{\num{7.500e1}} &       & 0.421  & 0.372  & \textbf{0.188}  &       & 0.168  & \textbf{0.138}  & 0.399  \\
    \multicolumn{3}{c}{\num{2.500e0}} &       & 0.421  & \textbf{0.004}  & 0.018  &       & 0.168  & \textbf{0.025}  & 0.055  \\
    \multicolumn{3}{c}{\num{8.333e-2}} &       & 0.421  & \textbf{0.000}  & \textbf{0.000}  &       & 0.168  & \textbf{0.002}  & \textbf{0.001}  \\
    \multicolumn{3}{c}{\num{2.777e-3}} &       & 0.421  & \textbf{0.000} & \textbf{0.000} &       & 0.168  & \textbf{0.000} & \textbf{0.000} \\
\hline
    \end{tabular}%
\end{table}

\begin{figure}[H]
\centering
\subcaptionbox{$\alpha_{I,0}=0.001$\label{subfig:a}}
    {
    \includegraphics[height=1.8in,width=1.8in]{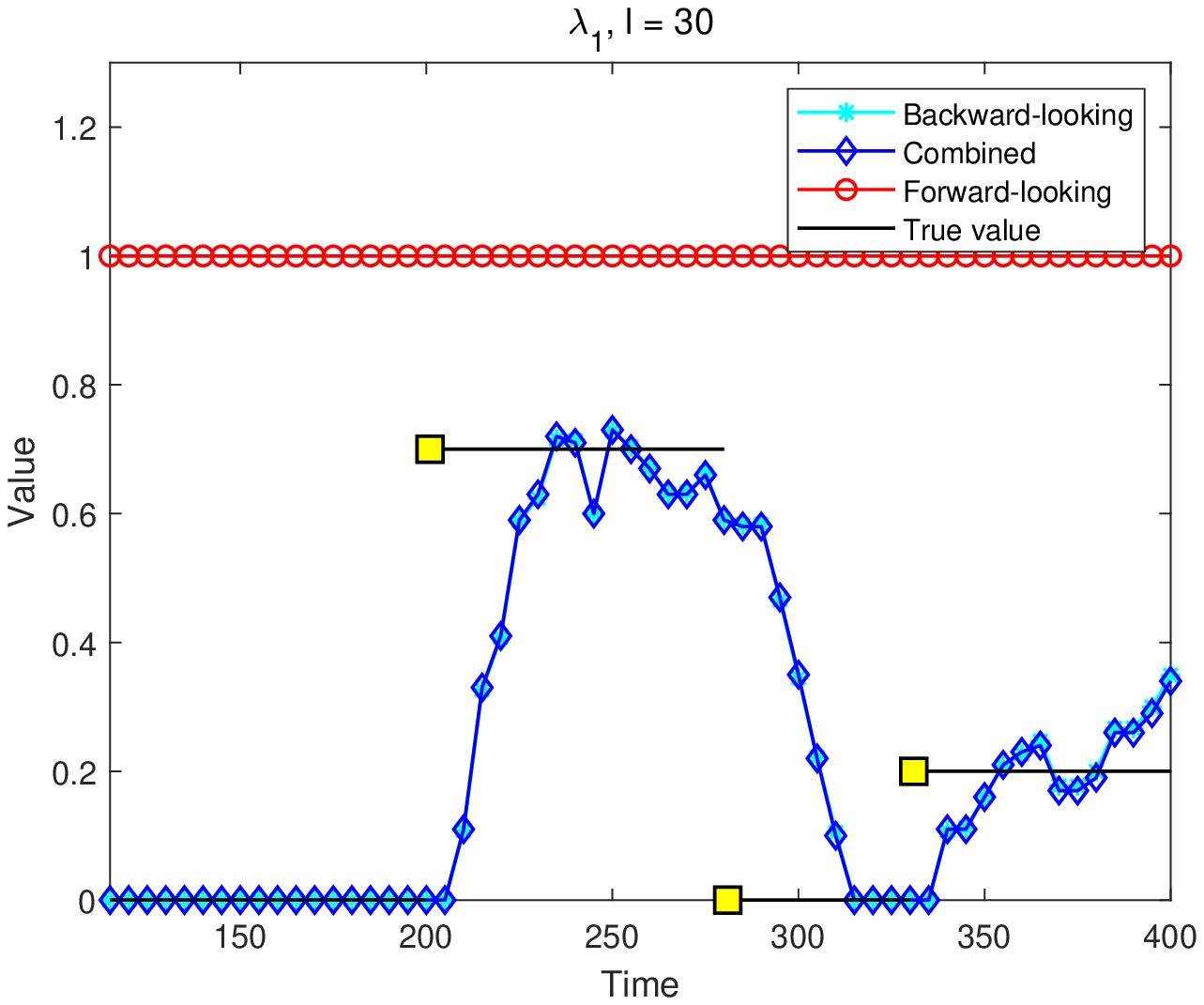}
    \includegraphics[height=1.8in,width=1.8in]{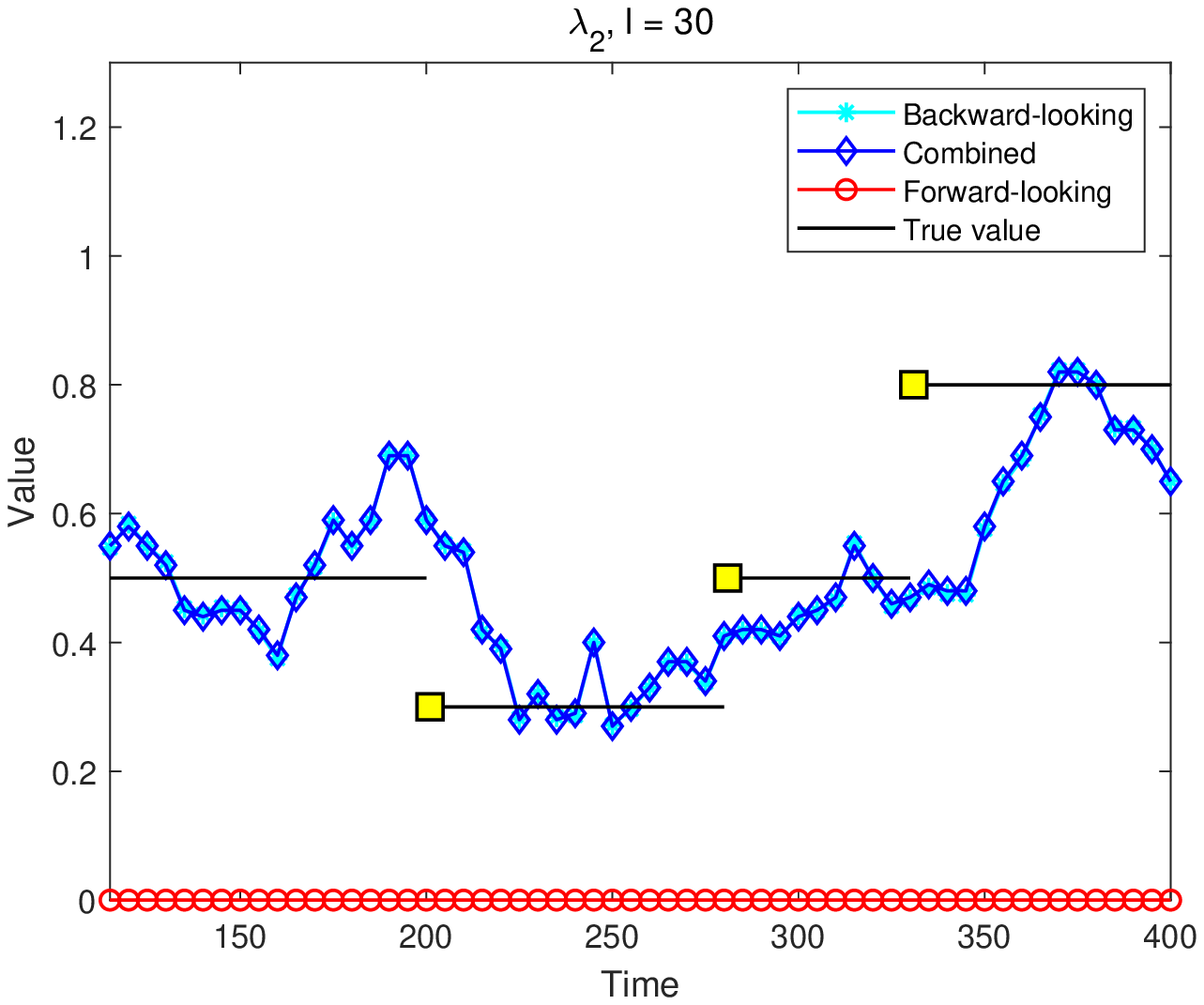}
    \includegraphics[height=1.8in,width=1.8in]{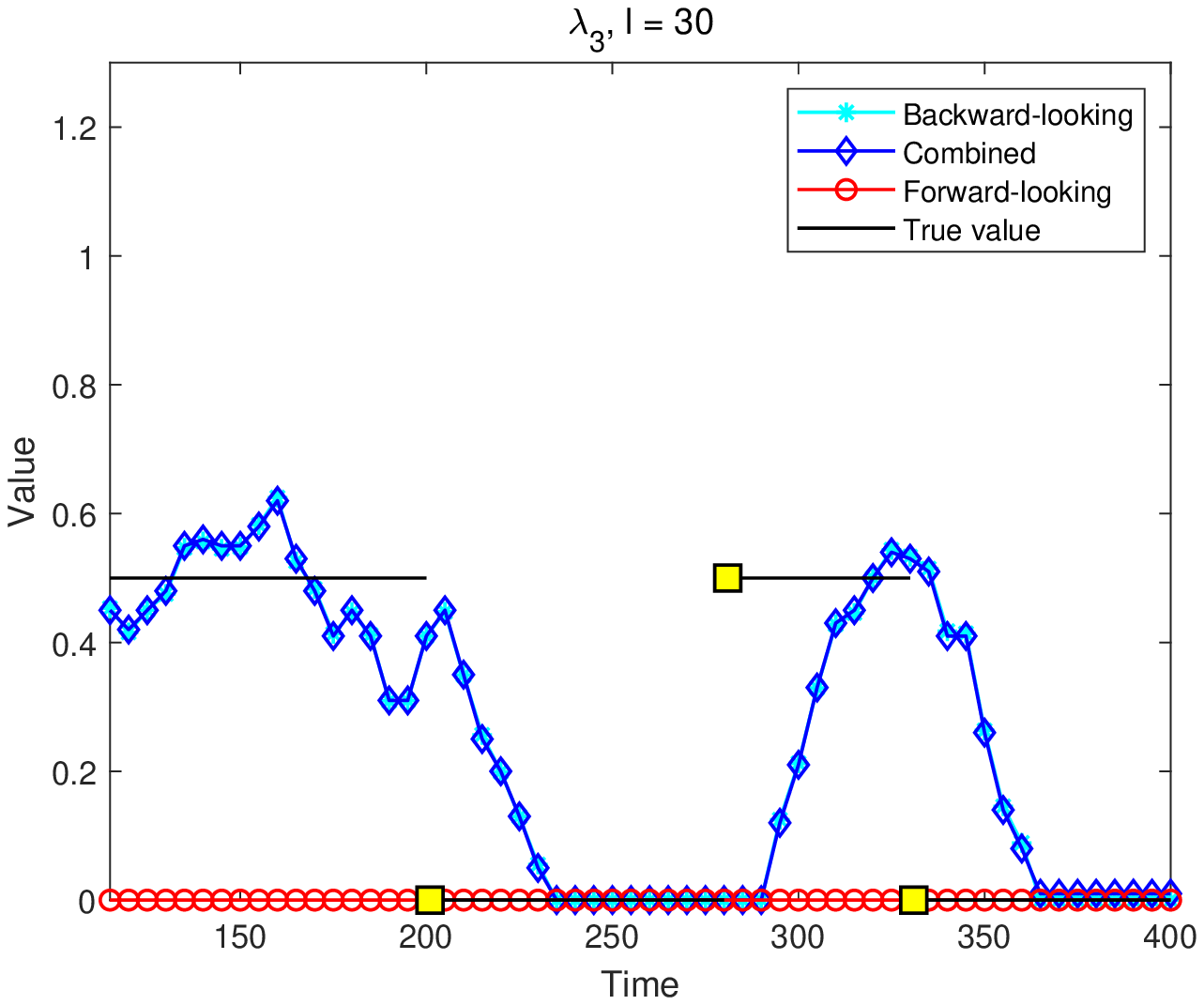}}
\subcaptionbox{$\alpha_{I,0}=0.300$\label{subfig:b}}
    {
    \includegraphics[height=1.8in,width=1.8in]{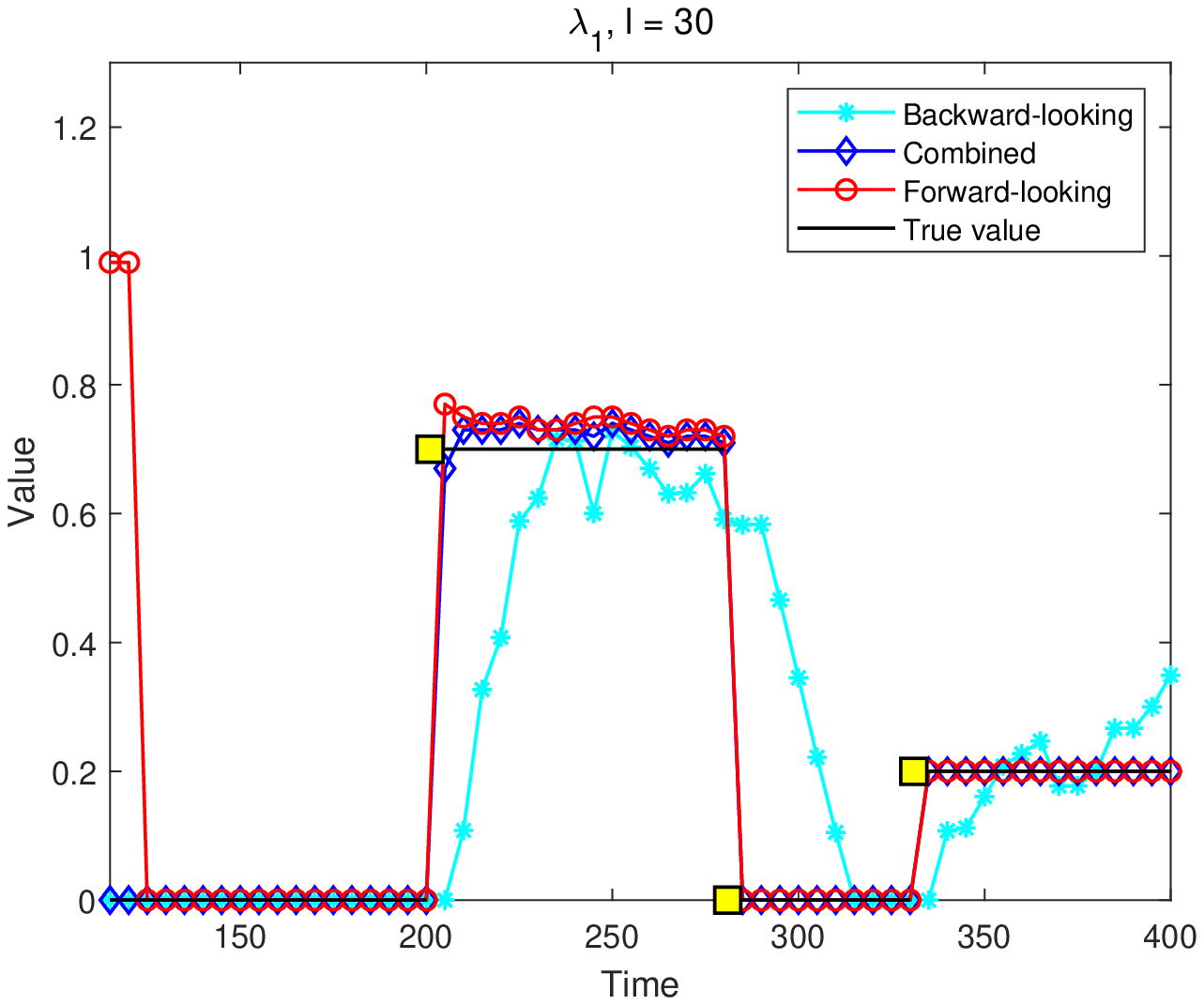}
    \includegraphics[height=1.8in,width=1.8in]{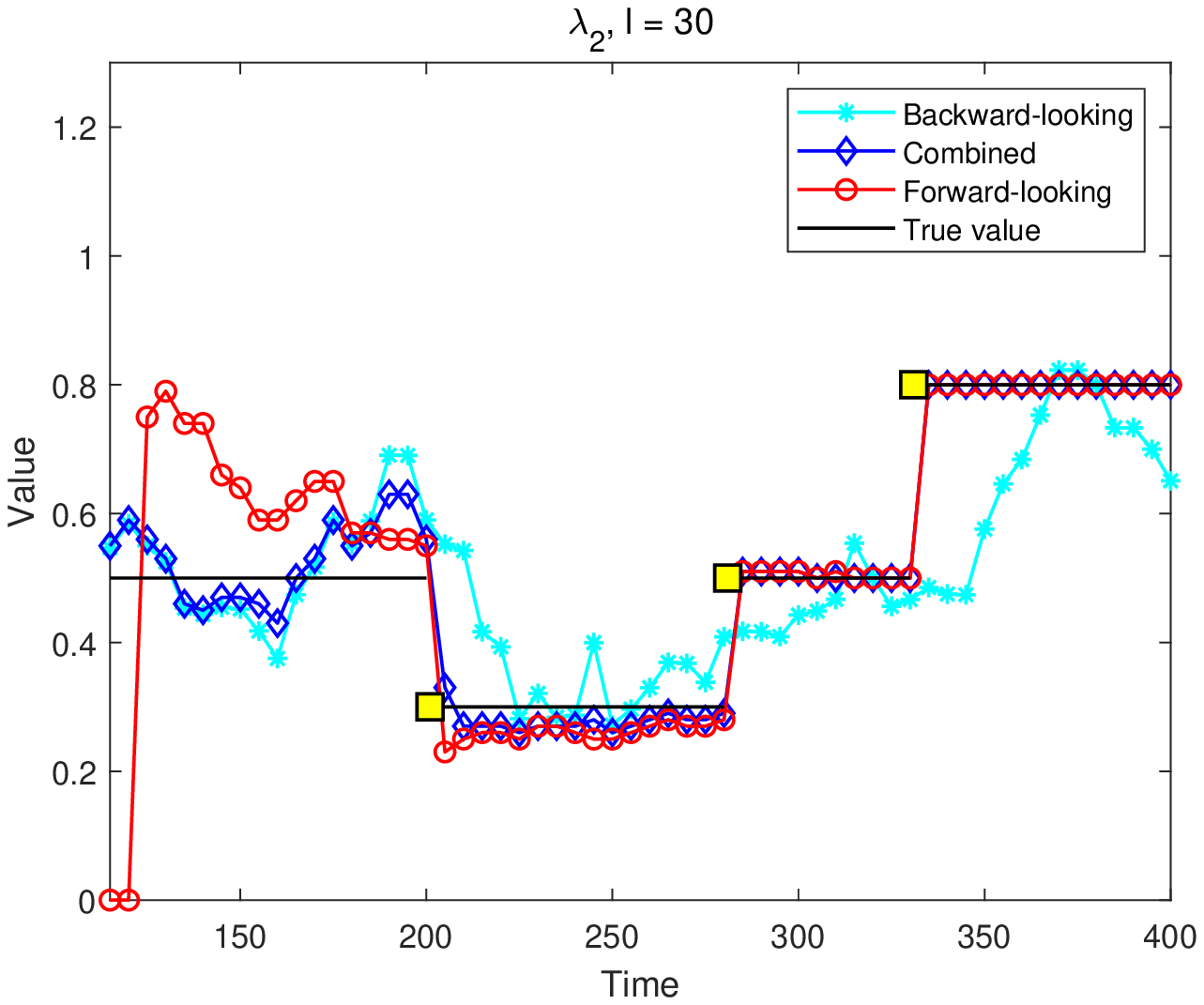}
    \includegraphics[height=1.8in,width=1.8in]{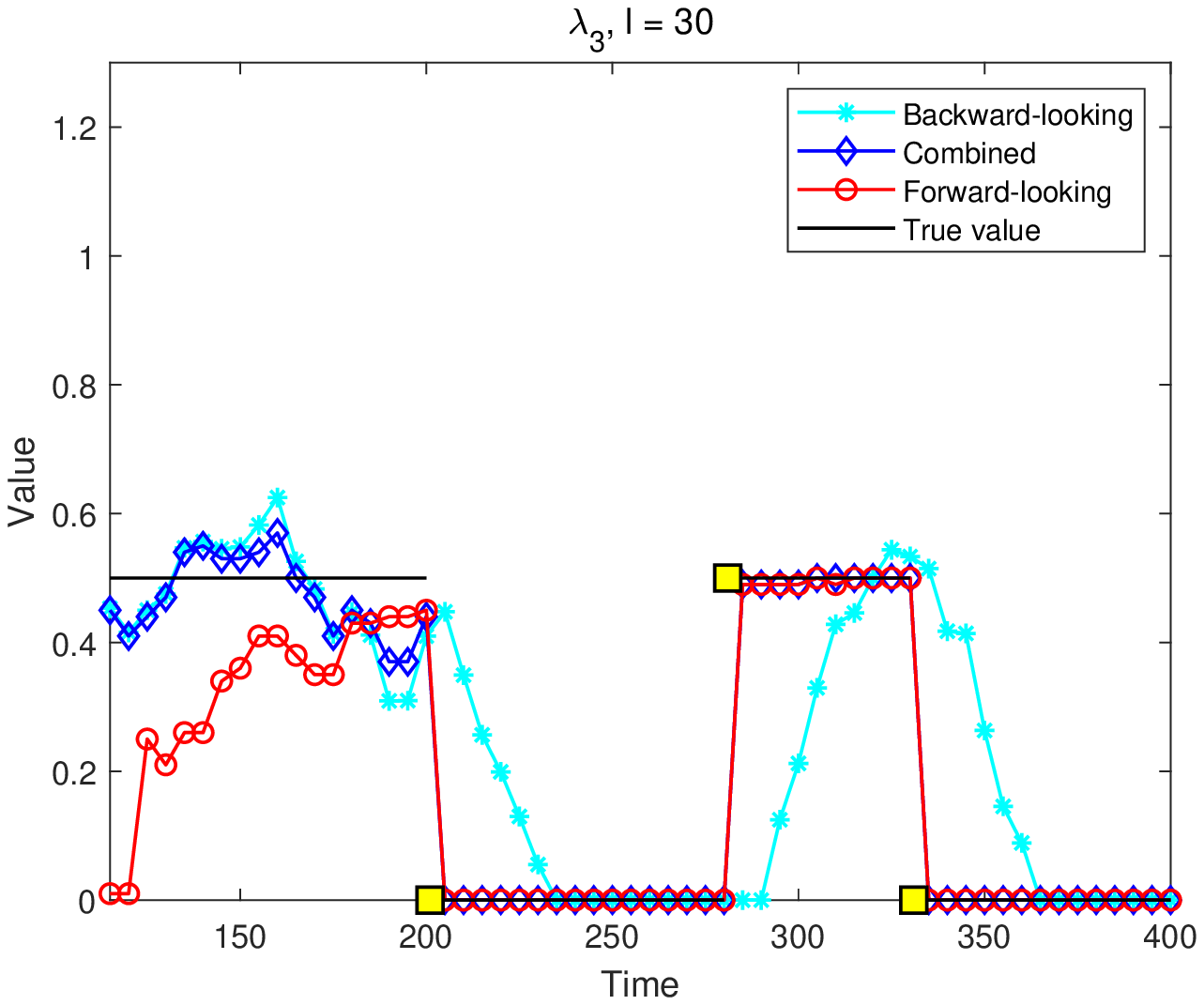}}
\subcaptionbox{$\alpha_{I,0}=0.600$\label{subfig:c}}
    {
    \includegraphics[height=1.8in,width=1.8in]{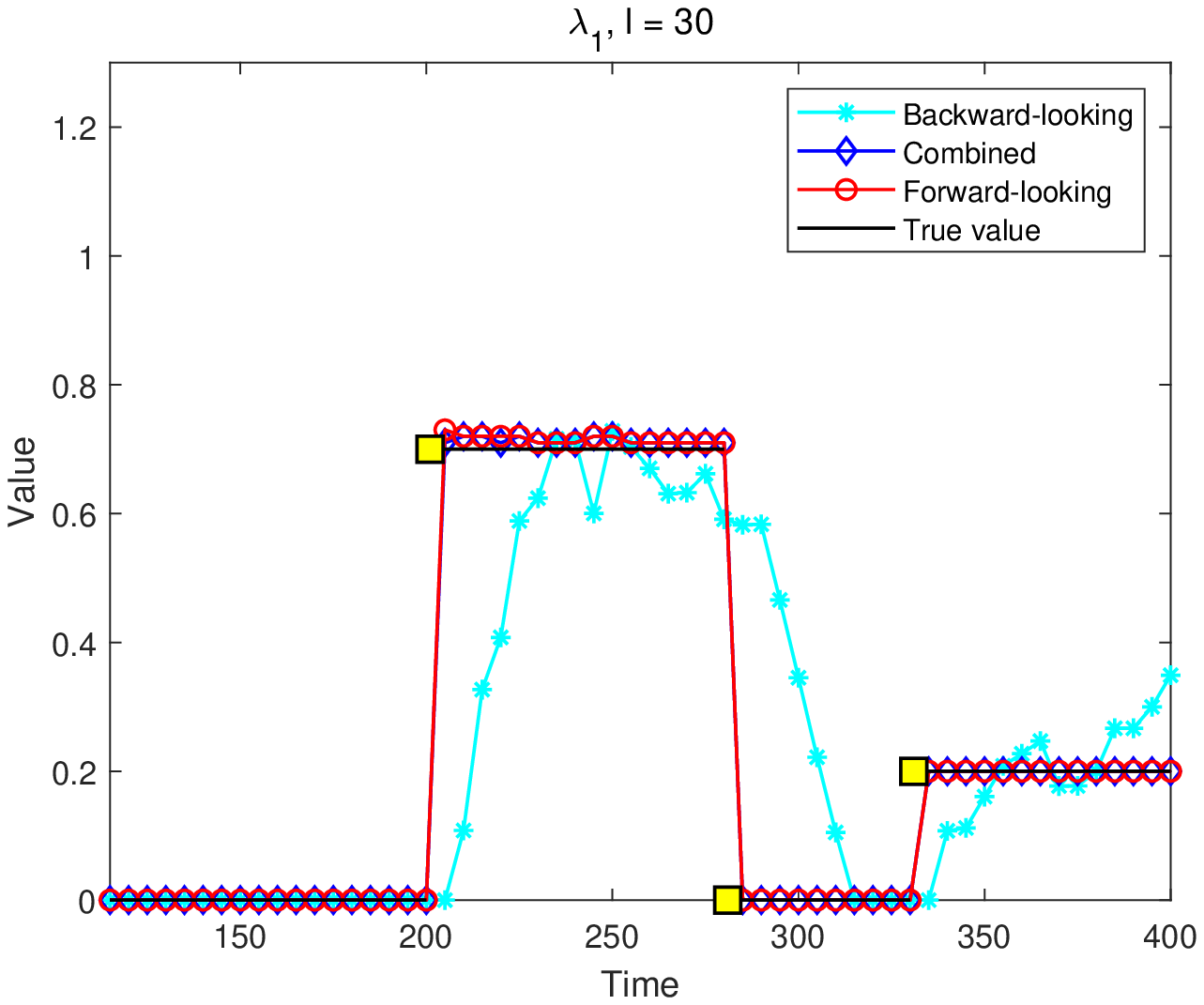}\quad
    \includegraphics[height=1.8in,width=1.8in]{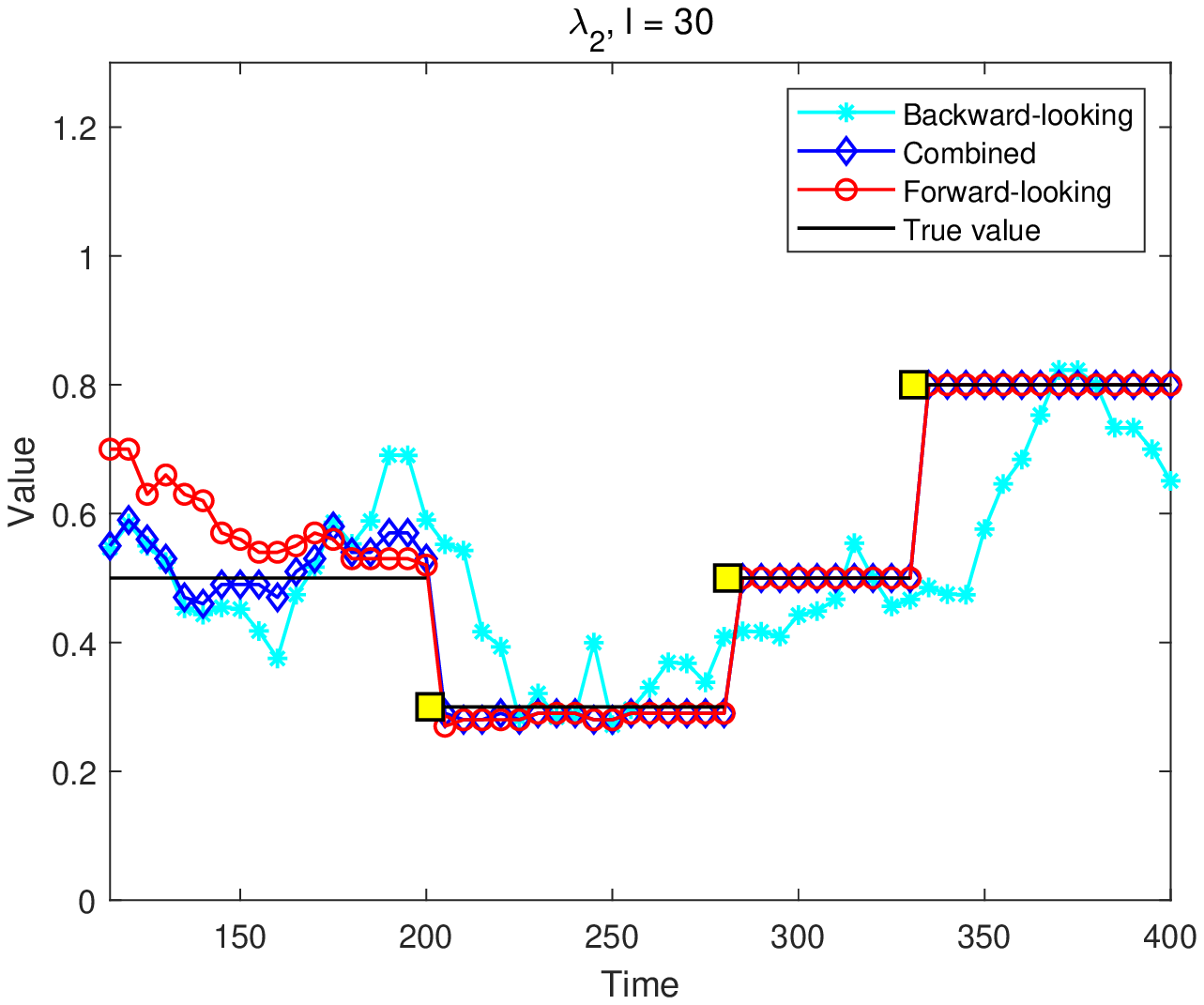}
    \includegraphics[height=1.8in,width=1.8in]{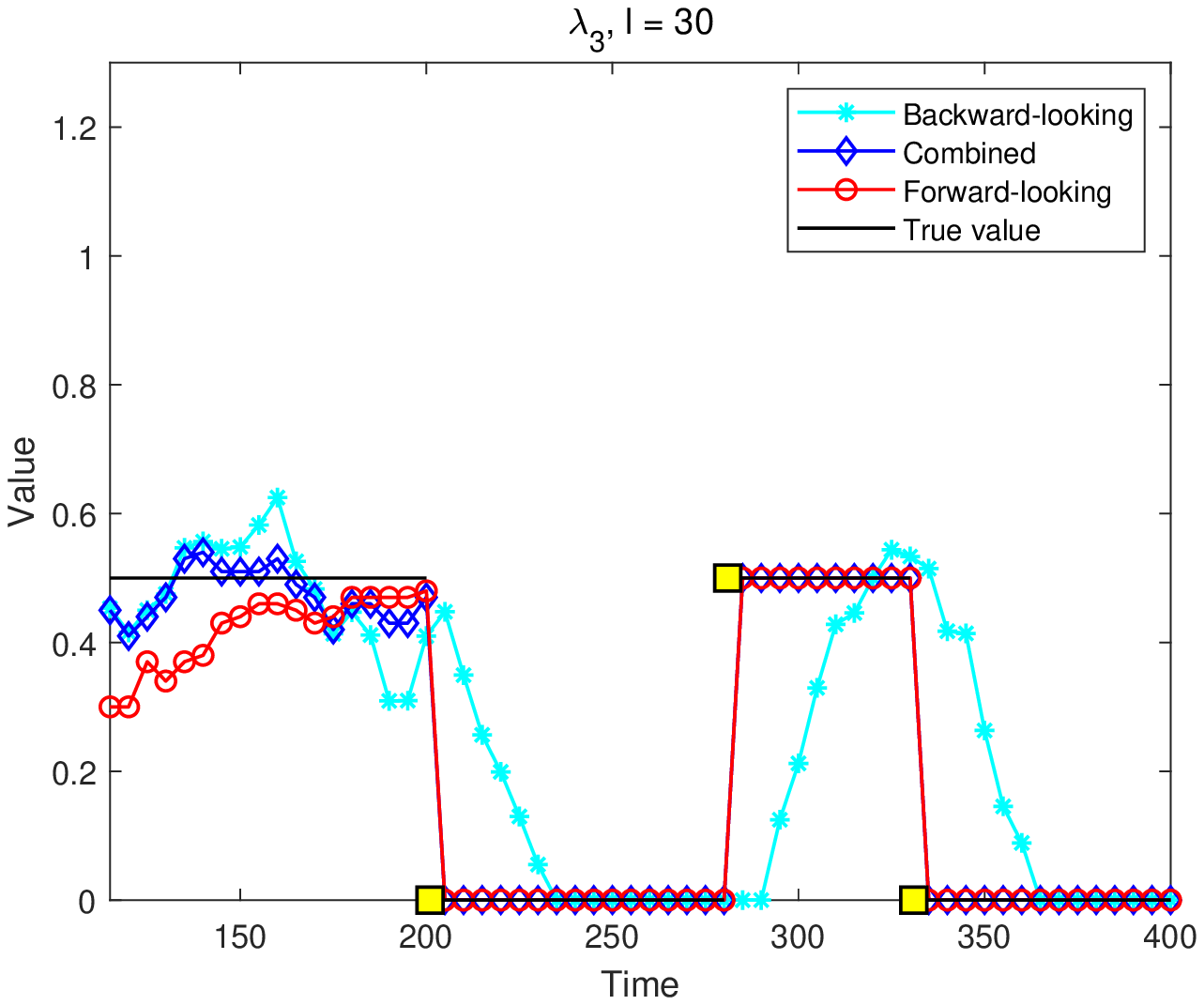}}
\subcaptionbox{$\alpha_{I,0}=0.990$\label{subfig:d}}
    {
    \includegraphics[height=1.8in,width=1.8in]{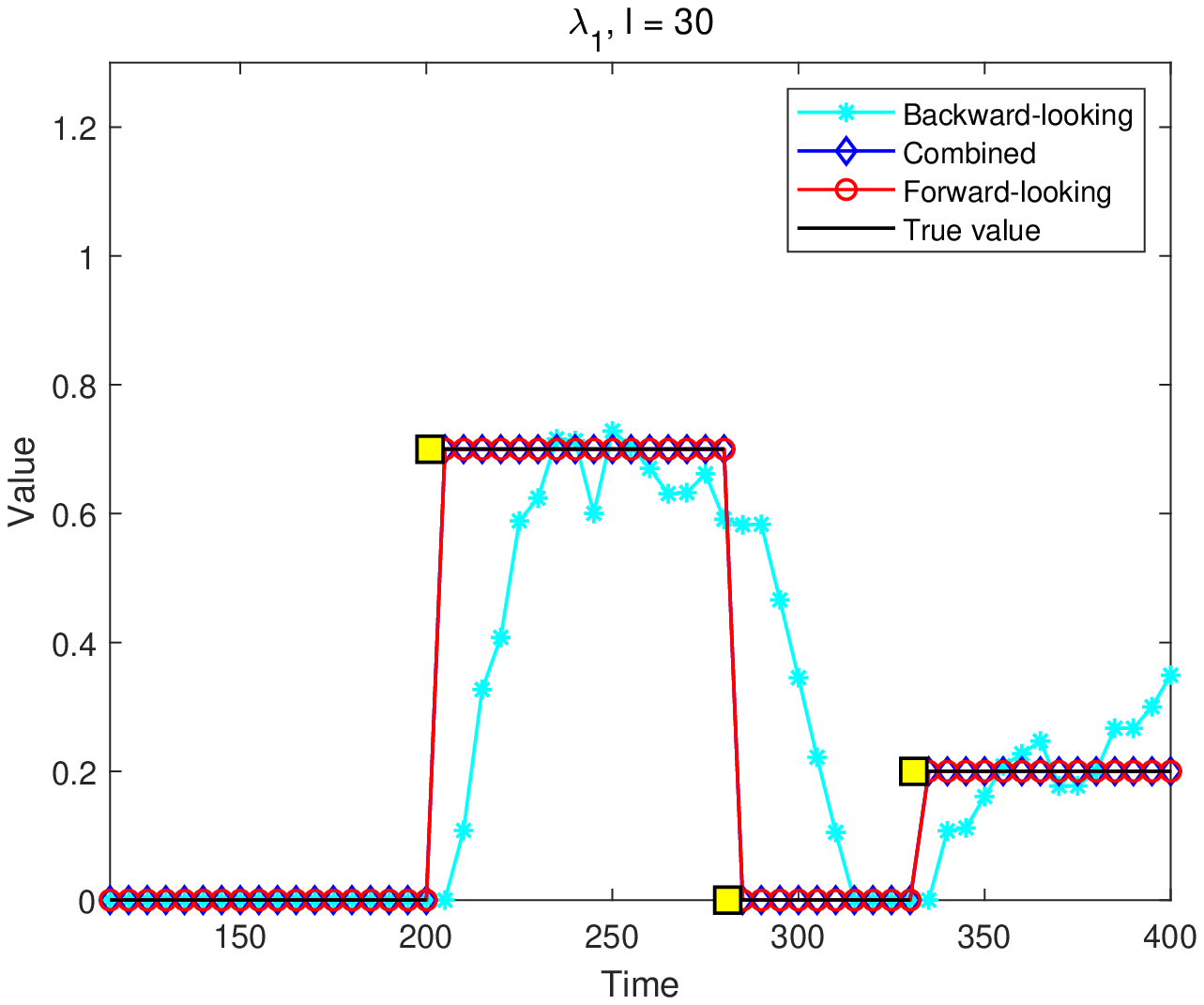}\quad
    \includegraphics[height=1.8in,width=1.8in]{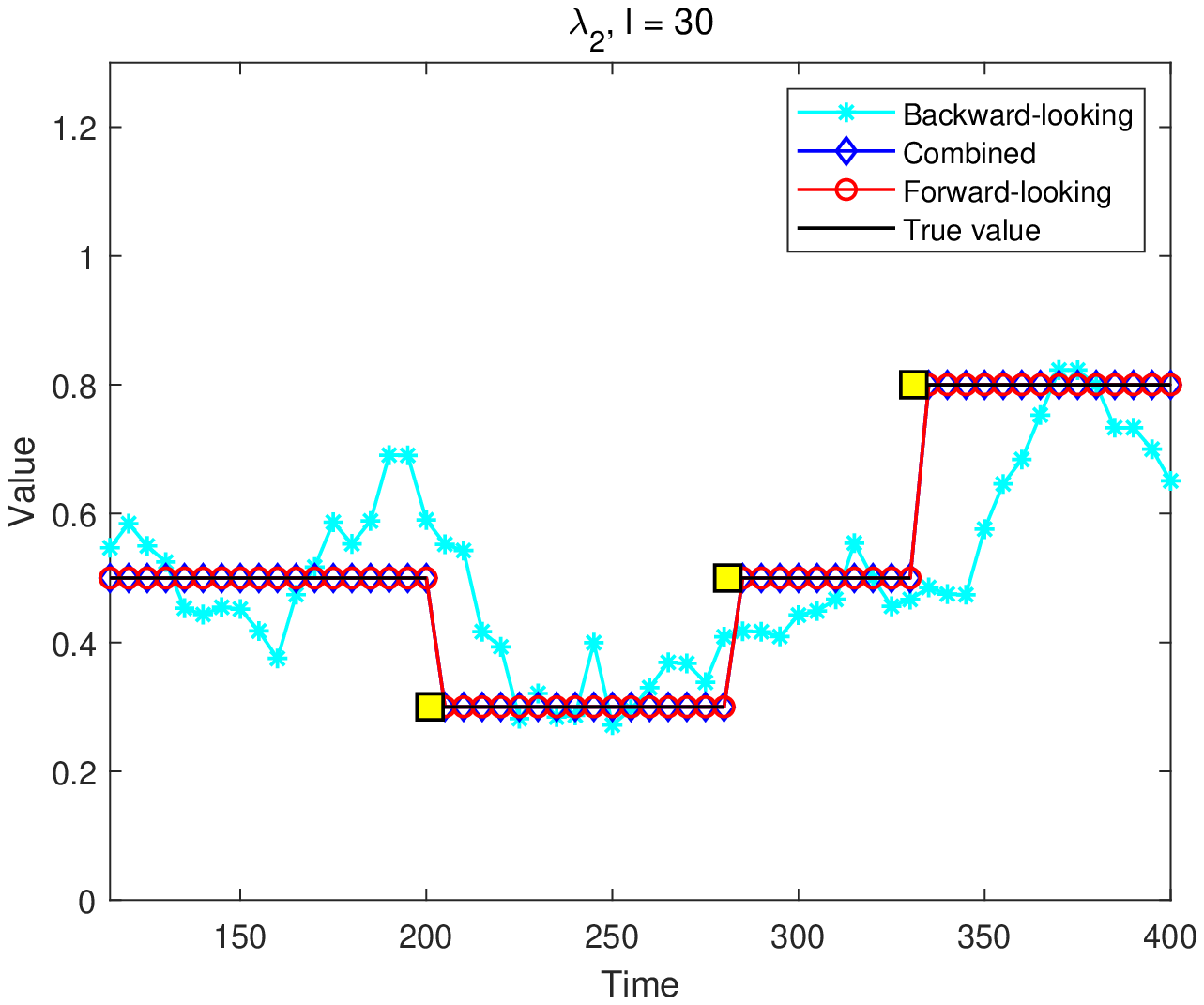}
    \includegraphics[height=1.8in,width=1.8in]{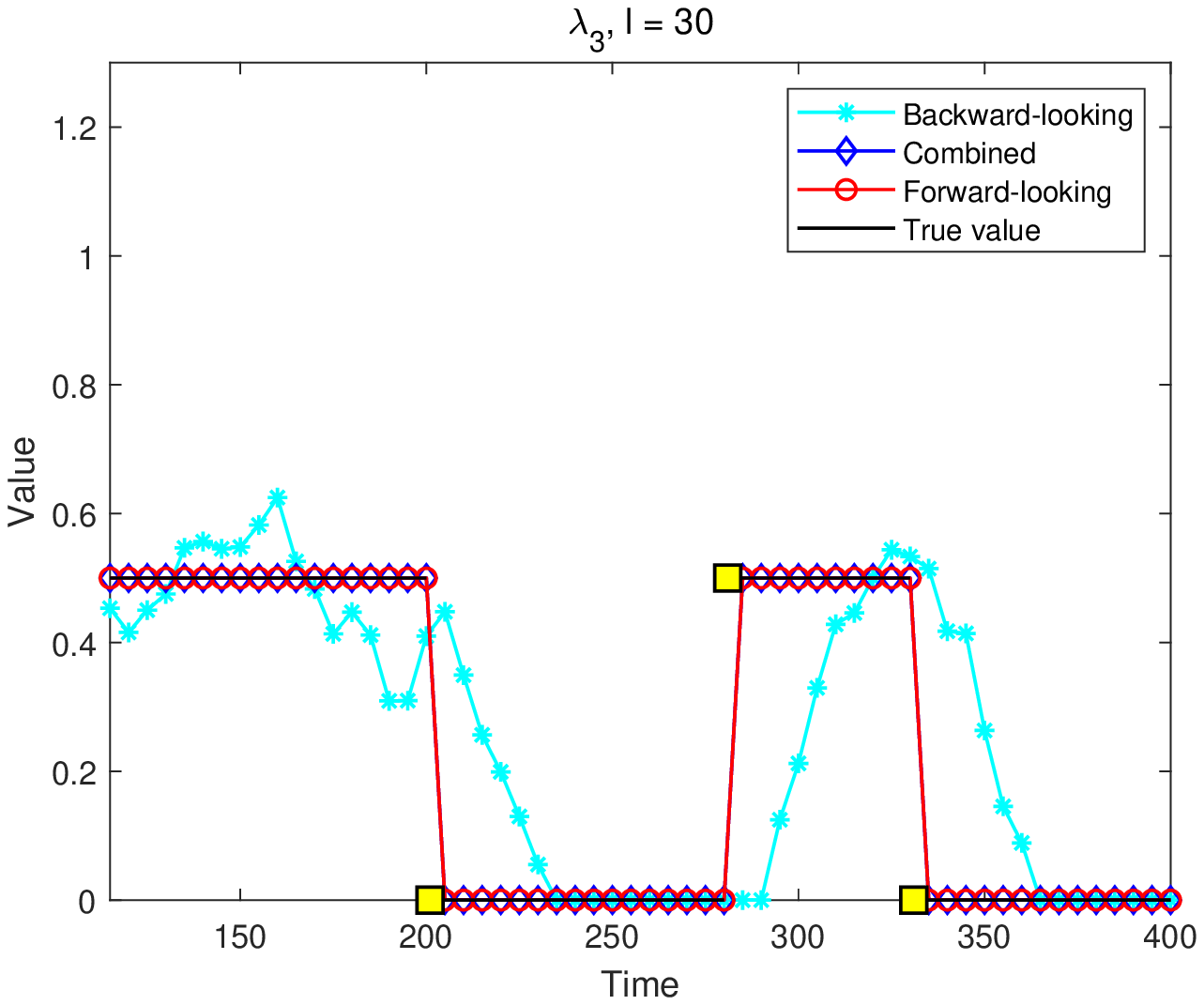}}
\caption{Comparisons of true values, backward-looking, combined, and forward-looking estimations.
We show the results in four typical markets, including a nearly completely effective market ($\alpha_{I,0}=0.990$), an almost ineffective
market ($\alpha_{I,0}=0.001$), and two relatively efficient markets ($\alpha_{I,0}=0.300,0.600$).}
\label{fig:2}
\end{figure}

Table \ref{table:1}
reports the RMSE of three types of estimations. First, we explore the impact of model parameters discussed in subsection \ref{2.4}.
It is clear that the RMSE of the backward-looking estimation $\eta_B$ in all market settings remains unchanged, while the RMSE of the combined estimation and forward-looking estimation vary with market parameters. In a market that is dominated by the informed investor ($\alpha_{I,0}=0.990$, $\sigma=\num{2.603e-5}$, or
$\delta_I=\num{2.777e-3}$),
both the combined estimation and the forward-looking estimation are close to the true value with RMSE $\eta_C=\eta_F\approx 0$. In markets where the
market
information is completely noisy  ($\alpha_{I,0}=0.001$,
$\sigma=\num{2.603e3}$, or $\delta_I=\num{2.250e3}$),
the forward-looking estimation exhibits larger errors, while the combined estimation is close to the backward-looking estimation with a smaller RMSE. For other market settings,  the combined estimation shows similar prediction pattern. It is closer to the forward-looking estimation if $\alpha_{I,0}$ is higher, $\sigma$ is lower, or $\delta_I$ is lower, and vice versa. In
addition, the combined estimation usually has a lower RMSE than the backward-looking one.

Second, we focus on the prediction power at turning points.
Figure \ref{fig:2}
shows the true mixture coefficients and the three estimations in four different markets by setting $\alpha_{I,0}$ to 0.001, 0.300, 0.600, and 0.990. Together with the RMSE at turning points shown in Table \ref{table:1},
we observe that the backward-looking approach lags in prediction when the market state changes, consistent with
\cite{zhu_portfolio_2014}.
As expected, the forward-looking and combined approaches can respond more quickly in markets with efficient forward-looking information. As the informed investor’s market share $\alpha_{I,0}$ decreases, the market becomes less efficient, and the accuracy of forward-looking estimation and combined estimation also decreases.

Overall, the forward-looking approach performs better in efficient markets, while the backward-looking approach is able to provide relatively reasonable estimation in extremely inefficient markets. However, for partially informative markets, the combined approach provides the highest accuracy in prediction, highlighting its robustness and self-adaptivity in integrating different sources of information.

\section{Empirical study}\label{Sec4}

In this section, we compare the out-of-sample performance of investment strategies determined by three types of estimations on 35 major stock markets over the world.

\subsection{Data description}

We select 35 constituent markets in the MSCI ACWI Index,\footnote{https://www.msci.com/zh/our-solutions/indexes/acwi} each of which has a complete time series of returns covering more than ten years.
Weekly returns, volatility indices, and market capitalization data of primary sector indices are obtained from the Bloomberg database.
Volatility indices are obtained to measure the risk-aversion coefficients. \footnote{For markets without option traded, we proxy the volatility index using EURO STOXX 50 volatility (VSTOXX), CBOE EFA ETF volatility index (VXEFA), or CBOE emerging markets ETF volatility index (VXEEM) according to their regions and market types.}
To proxy for the market shares of different investors, we obtain the market value ratio of institutional investors from the exchange websites.
\footnote{For markets that do not disclose this ratio, we use the indicators of markets with similar economic development levels as substitutes.}
Risk-free returns for each market are proxied by 3-month treasury yields, interbank-offered rates, or central bank rates.

Table \ref{table:2}
lists the detailed description of the data where DM represents the developed market, while EM denotes the emerging market.
Table \ref{table:3}
reports the descriptive statistics including the number of industry sectors, mean, standard deviation, skewness, and kurtosis of index excess returns.
Consistent with the existing literature, the excess returns is left-skewed and heavy-tailed, which highlights the significance of modeling asset returns with the Gaussian mixture distribution.

\begin{table}[htbp]
\centering
\captionsetup{labelsep = period, labelfont = bf, font = {stretch=1}}
  \caption{Data information of 35 markets.}\label{table:2}
  \resizebox{\linewidth}{8cm}{
    \begin{tabular}{clllllll}
    \hline
    \multicolumn{3}{c}{Market Description} &       & \multicolumn{2}{c}{Index} &  & \multicolumn{1}{c}{\multirow{2}[0]{*}{Sample period}} \\
\cmidrule{1-3}\cmidrule{5-6}    Region & Market & DM/EM &       & Market Index & Volatility Index & &  \\
\hline
    \multicolumn{1}{c}{\multirow{5}[0]{*}{Americas}} & Canada & DM    &       & SPTSX & TXLV &       & 2000.1.7-2020.11.13 \\
          & US  & DM    &       & SPX   & VIX &       & 2001.10.19-2020.11.13 \\
          & Brazil & EM    &       & MXBR  & VXEWZ &       & 2000.1.7-2020.11.13 \\
          & Chile & EM    &       & MXCL  & VXEEM &       & 2000.9.1-2020.11.13 \\
          & Mexico & EM    &       & MXMX  & VXEEM &       & 2000.1.7-2020.11.13 \\
\hline
    \multicolumn{1}{c}{\multirow{18}[0]{*}{Europe \& Middle East}} & Austria & DM    &       & MXAT  & VSTOXX &       & 2000.6.2-2020.11.13 \\
          & Belgium & DM    &       & BEL20 & VSTOXX &       & 2005.12.9-2020.11.14 \\
          & France & DM    &       & CAC40 & VSTOXX &       & 2000.12.22-2020.11.13 \\
          & Germany & DM    &       & MXDE  & VDAX &       & 2000.1.7-2020.11.13 \\
          & Ireland & DM    &       & MXIE  & VSTOXX &       & 2000.1.7-2020.11.13 \\
          & Israel & DM    &       & MXIL  & VXEFA &       & 2000.1.7-2020.11.13 \\
          & Italy & DM    &       & MXIT  & VSTOXX &       & 2000.1.7-2020.11.13 \\
          & Norway & DM    &       & OSEAX & VXEFA &       & 2000.1.7-2020.11.13 \\
          & Portugal & DM    &       & MXPT  & VSTOXX &       & 2000.1.7-2020.11.13 \\
          & Spain & DM    &       & MADX  & VSTOXX &       & 2005.1.7-2020.11.13 \\
          & Switzerland & DM    &       & MXCH  & VSMI &       & 2000.12.1-2020.11.14 \\
          & UK  & DM    &       & NMX   & VSTOXX &       & 2006.1.6-2020.11.13 \\
          & Czech Republic  & EM    &       & MXCZ  & VXEEM &       & 2000.1.7-2020.11.13 \\
          & Poland & EM    &       & MXPL  & VXEEM &       & 2000.1.7-2020.11.13 \\
          & Russia & EM    &       & MXRU  & VXEEM &       & 2005.6.3-2020.11.13 \\
          & South Africa & EM    &       & JALSH & SAVI &       & 2000.1.7-2020.11.13 \\
          & Turkey & EM    &       & MXTR  & VXEEM &       & 2000.8.4-2020.11.13 \\
          & UAE & EM    &       & ADSMI & VXEEM &       & 2006.7.13-2020.11.12 \\
\hline
    \multicolumn{1}{c}{\multirow{12}[0]{*}{Pacific}} & Australia & DM    &       & AS51  & VXEFA &       & 2000.3.31-2020.11.13 \\
          & Hong Kong & DM    &       & HSI   & VHSI &       & 2000.1.7-2020.11.13 \\
          & Japan & DM    &       & TPX   & VXEFA &       & 2000.1.7-2020.11.13 \\
          & Singapore & DM    &       & FSTAS & VXEFA &       & 2000.1.7-2020.11.13 \\
          & China & EM    &       & SH300 & VXFXI &       & 2005.1.7-2020.11.13 \\
          & India & EM    &       & BSE500 & VXEEM &       & 2005.1.9-2020.11.13 \\
          & Indonesia & EM    &       & JCI   & VXEEM &       & 2000.1.7-2020.11.13 \\
          & Korea & EM    &       & KOSPI & KSVKOSPI &       & 2000.1.2-2020.11.8 \\
          & Malaysia & EM    &       & FBMKLCI & VXEEM &       & 2000.5.19-2020.11.13 \\
          & Philippines & EM    &       & PCOMP & VXEEM &       & 2006.1.6-2020.11.13 \\
          & Taiwan & EM    &       & TWSE  & VXEEM &       & 2007.7.14-2020.11.14 \\
          & Thailand & EM    &       & SET   & VXEEM &       & 2000.1.7-2020.11.13 \\
\hline
    \end{tabular}}%
\end{table}%

\begin{table}[htbp]
\centering
\captionsetup{labelsep = period, labelfont = bf, font = {stretch=1}}
  \caption{Market parameters and descriptive statistics of excess return. This table reports the
market parameters of the 35 markets including the market share of the informed investor ($\alpha_{I,0}$) and noise-trading intensity ($\sigma$) and descriptive statistics for indices including the number of industry sectors ($n$), mean, standard deviation (Std), skewness (skew) and kurtosis (Kurt). Mean and Std are shown in percentage.}\label{table:3}
    \begin{tabular}{lccclrrrrr}
\hline
    \multicolumn{1}{l}{\multirow{2}[3]{*}{Market}} & \multicolumn{2}{c}{Parameters} &       & \multicolumn{1}{l}{\multirow{2}[3]{*}{Index}} & \multicolumn{5}{c}{Descriptive statistics} \\
\cmidrule{2-3}\cmidrule{6-10}    \multicolumn{1}{l}{} & $\alpha_{I,0}$ & $\sigma$     &       & \multicolumn{1}{c}{} &  $n$    & \multicolumn{1}{c}{Mean(\%)} & \multicolumn{1}{c}{Std(\%)} & \multicolumn{1}{c}{Skew} & \multicolumn{1}{c}{Kurt} \\
\hline
    Canada & 0.500   & 1.041  &       & SPTSX & 11    & 0.065  & 3.165  & -1.410  & 13.391  \\
    US & 0.600   & 0.260  &       & SPX & 10    & 0.118  & 2.459  & -1.019  & 12.165  \\
    Brazil & 0.300   & 4.165  &       & MXBR & 8     & 0.006  & 5.147  & -0.564  & 7.608  \\
    Chile & 0.300   & 4.165  &       & MXCL & 7     & 0.033  & 3.429  & -1.256  & 15.963  \\
    Mexico & 0.200   & 6.508  &       & MXMX & 5     & 0.050  & 3.930  & -0.736  & 10.569  \\
    Austria & 0.500   & 1.041  &       & MXAT & 5     & 0.017  & 4.082  & -1.489  & 15.005  \\
    Belgium & 0.500   & 1.041  &       & BEL20 & 16    & -0.034  & 3.443  & -1.604  & 13.843  \\
    France & 0.500   & 1.041  &       & CAC40 & 10    & -0.013  & 3.049  & -1.289  & 12.836  \\
    Germany & 0.500   & 1.041  &       & MXDE & 9     & 0.013  & 3.554  & -0.933  & 9.802  \\
    Ireland & 0.400   & 2.343  &       & MXIE & 3     & -0.049  & 3.833  & -1.869  & 18.257  \\
    Israel & 0.500   & 1.041  &       & MXIL & 4     & 0.012  & 3.155  & -0.311  & 6.532  \\
    Italy & 0.300   & 4.165  &       & MXIT & 6     & -0.057  & 3.674  & -1.203  & 11.632  \\
    Norway & 0.500   & 1.041  &       & OSEAX & 10    & 0.130  & 3.800  & -1.303  & 11.364  \\
    Portugal & 0.300   & 4.165  &       & MXPT & 5     & -0.071  & 3.288  & -1.037  & 8.745  \\
    Spain & 0.300   & 4.165  &       & MADX & 6     & -0.042  & 3.797  & -1.100  & 9.860  \\
    Switzerland & 0.500   & 1.041  &       & MXCH & 8     & 0.090  & 2.630  & -1.271  & 15.063  \\
    UK & 0.500   & 1.041  &       & NMX & 35    & -0.019  & 3.197  & -1.492  & 16.034  \\
    Czech Republic & 0.400   & 2.343  &       & MXCZ & 3     & 0.084  & 3.708  & -0.797  & 11.428  \\
    Poland & 0.400   & 2.343  &       & MXPL & 6     & -0.033  & 4.359  & -0.697  & 8.632  \\
    Russia & 0.500   & 1.041  &       & MXRU & 6     & -0.014  & 5.118  & -0.220  & 12.496  \\
    South Africa & 0.400   & 2.343  &       & JALSH & 25    & 0.058  & 3.894  & -0.334  & 7.989  \\
    Turkey & 0.300   & 4.165  &       & MXTR & 7     & 0.097  & 4.522  & -0.210  & 7.994  \\
    UAE & 0.500   & 1.041  &       & ADSMI & 8     & 0.043  & 2.680  & -1.837  & 15.099  \\
    Australia & 0.400   & 2.343  &       & AS51 & 10    & 0.083  & 3.353  & -1.746  & 18.056  \\
    Hong Kong & 0.500   & 1.041  &       & HSI & 4     & 0.042  & 3.007  & -0.233  & 5.431  \\
    Japan & 0.500   & 1.041  &       & TPX & 12    & 0.006  & 2.646  & -0.419  & 7.360  \\
    Singapore & 0.500   & 1.041  &       & FSTAS & 11    & 0.018  & 2.850  & -0.575  & 9.910  \\
    China & 0.300   & 4.165  &       & SH300 & 10    & 0.351  & 3.998  & -0.235  & 4.784  \\
    India & 0.400   & 2.343  &       & BSE500 & 11    & 0.187  & 3.023  & -0.621  & 7.761  \\
    Indonesia & 0.200   & 6.508  &       & JCI & 9     & 0.095  & 3.864  & -0.874  & 8.549  \\
    Korea & 0.400   & 2.343  &       & KOSPI & 19    & 0.067  & 4.076  & -0.795  & 9.558  \\
    Malaysia & 0.400   & 2.343  &       & PCOMP & 10    & 0.027  & 2.254  & -0.509  & 6.450  \\
    Philippines & 0.300   & 4.165  &       & FBMKLCI & 6     & 0.151  & 3.244  & -0.990  & 9.398  \\
    Taiwan & 0.400   & 2.343  &       & TWSE & 28    & 0.066  & 2.926  & -0.677  & 4.924  \\
    Thailand & 0.400   & 2.343  &       & SET   & 23    & 0.110  & 3.210  & -1.107  & 10.304  \\
\hline
\end{tabular}%
\end{table}%

\subsection{Empirical setting}
\label{4.2}

We first introduce the settings of market parameters.
The market share $\alpha_{I,0}$ of the informed investor at time 0 is set as the market share of institutional investors. In addition, we simply set the market share of the noise investor as $\alpha_{N,0}=0.100$ and the initial market share of the less-informed investor as $\alpha_{U,0}=1-\alpha_{I,0}-\alpha_{N,0}.$ The updating of the market shares of the three types of investors is in line with simulation studies.

With respect to noise intensity, we assume that markets with a higher share of the informed investor exhibit lower noise trading intensity.
We assign 95\% confidence intervals of $[-1,1],[-2,2],[-3,3],[-4,4],[-5,5]$ to the noise investor's holding in markets with informed investor shares of 0.6, 0.5, 0.4, 0.3, and 0.2.
The market shares of informed investors and noise trading intensities of the 35 markets are listed in the second and third columns of
Table \ref{table:3}.

To more effectively characterize the behavior of investors in the real world, we introduce a volatility index-based time-varying risk aversion coefficient for the informed investor, as the risk-neutral volatility has been shown to encapsulate investors' risk aversion (see e.g.,
\cite{bekaert2014vix}).
Specifically
\begin{align}
\notag
\delta_{I,t}=\rho_t\delta_0,~\rho_t=\frac{VIX_{t}}{\frac{1}{\Delta t}\sum_{i=t-\Delta t+1}^{t} VIX_i},
\end{align}
where $\delta_0=2.500$, $VIX_{t}$ is the volatility index in the corresponding market, and $\Delta t=100$ weeks.
 The risk-aversion coefficient of the less-informed investor is set as $\delta_U = 2.500$.

Now we introduce the setting of the mixture components. We set $m=4$ Gaussian components, including three determined by the historical data and one derived according to
\cite{black_global_1992}.
The three historical data related components are estimated using the return data before January 2012 and remain unchanged in the investment period, which characterize the excess return in bull, oscillating, and bear markets, respectively. \footnote{To do so, we compute the percentage rise (fall) of the index based on the highest and lowest value. A period is classfies as a bull market if the percentage rise is greater than 90\%, a bear market if the percentage fall is less than -50\%, and an oscillating market if the excess return is close to zero.}
The selection of component associated with the prior distribution of
\cite{black_global_1992}
aims to better utilize the forward-looking information. Specifically, this component is determined according to the market portfolio as  $\mathcal{N}(\delta_{I,t} \hat{\Sigma}_t \bm{x}_{M,t},\hat{\Sigma}_t )$, where $ \hat{\Sigma}_t$ is the covariance matrix dynamically estimated with the latest 100-week historical data by rolling windows.

\subsection{Empirical results}

We compute three types of mixture weights based on
Defintion\ref{defin}
and obtain corresponding estimations for the mean vector and covariance matrix with
\eqref{eq:8}.
Then, we input these estimations into mean-variance optimization
\eqref{eq:1}
and obtain the backward-looking strategy, forward-looking strategy, and combined strategy, respectively.

The out-of-sample investment period is from January 2012 to November 2020. Short-selling is not allowed and we only consider the allocation within risky assets
\citep{demiguel_improving_2013}.
For robustness, we rebalance the portfolio every month, quarter, and half a year, and a 100-week rolling window is adopted. We also set the risk-aversion coefficients as $0.3\delta_{I,t}$, $\delta_{I,t}$ and $3\delta_{I,t}$, corresponding to aggressive, moderate, and conservative investors, respectively.

We calculate indicators including average (across rebalancing window) annualized Sharpe ratio (SR), turnover rate (TRN) \citep{demiguel_improving_2013},
and maximum drawdown (MDD) \citep{chekhlov2005drawdown}
to evaluate the out-of-sample performance of the strategies
\begin{subequations}
\begin{align}
SR&=\frac{{\frac{250}{\triangle t}\frac{1}{T}\sum_{t=1}^T r_{p,t}}}{\sqrt{\frac{250}{\triangle t}\frac{1}{T-1}\sum_{t=1}^T \left(r_{p,t}-\frac{1}{T}\sum_{t=1}^T r_{p,t}\right)^2}},\nonumber \\
TRN&=\frac{1}{T-1}\sum_{t=1}^{T-1} \sum_{i=1}^n |x_{i,t+1}-x_{i,{t}^+}|,\nonumber \\
MDD&=\notag \mathop{\max}\limits_{0 \leq i< j \leq T}\{1-w_{p,j}/w_{p,i}\}.
\end{align}
\end{subequations}
where $w_{p,t}$ and $r_{p,t}$ are wealth and portfolio return at time $t$ computed by
\begin{align}
\notag w_{p,t} =  r_{p,t-1}w_{p,t-1}~\text{and}~r_{p,t} &= \sum_{i=1}^{n}\tilde{r}_{i,t}x_{i,t}\left(1-c\sum_{i=1}^{n} |x_{i,t}-x_{i,{t-1}}|\right).
\end{align}
We set the transaction cost ratio as $c=0.5\%$
\citep{demiguel_optimal_2009}.
An outperforming portfolio usually has a higher SR, a lower TRN, and a lower MDD.

Table \ref{table:4}
reports the average (across the rebalancing window) annualized SR, TRN, and MDD of the three types of strategies with moderate risk-aversion setting in the 35 markets.
\footnote{Table \ref{table:5}
in the Appendix shows results with aggressive, moderate, and conservative risk-aversion settings, results are consistent.
Table \ref{table:5}
also shows the comparisons with benchmark (1/N, index, and Minimum-variance) strategies.}
Panel A shows the results of markets where the backward-looking strategy performs the best, Panel B reports the results of markets where the combined strategy performs the best, and Panel C shows the results of markets where the forward-looking strategy performs the best. We also plot the SR of the optimal mean-variance portfolio for each market in
Figure \ref{fig:3}
for a more intuitive exhibition.
\begin{figure}[h]
\centering
\includegraphics[height=4in,width=5.5in]{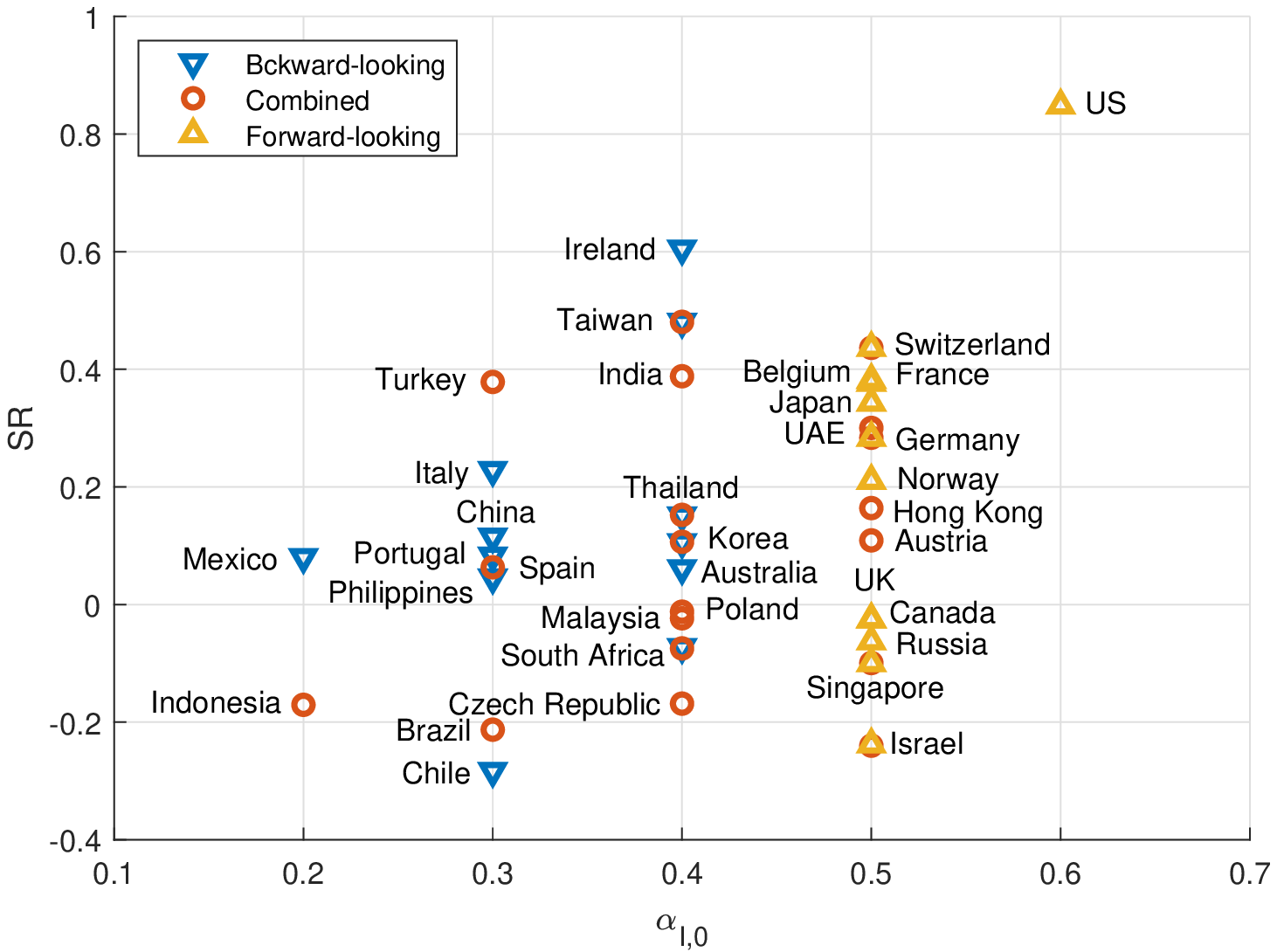}
\caption{Sharpe ratio (SR) of the optimal mean-variance portfolio for each market. The horizontal axis indicates the market share of informed investors at time 0.}
\label{fig:3}
\end{figure}

\begin{table}
  \centering
  \captionsetup{labelsep = period, labelfont = bf, font = {stretch=1}}
  \caption{Empirical tests: Comparisons of backward-looking strategy, combined strategy and forward-looking strategy. The results with a moderate risk-aversion setting are shown. This table reports out-of-sample performance indicators, including annualized SR, TRN and MDD. Each indicator is the average across the three rebalancing windows.}\label{table:4}
  \resizebox{\linewidth}{9.6cm}{
    \begin{tabular}{lcrrrrrrrrrrrr}
\hline
          &       & & \multicolumn{3}{c}{Backward-looking Strategy} & & \multicolumn{3}{c}{Combined Strategy} & & \multicolumn{3}{c}{Forward-looking Strategy} \\
\cmidrule{4-6}\cmidrule{8-10}\cmidrule{12-14}   \multicolumn{1}{c}{Market} & \multicolumn{1}{c}{$\alpha_{I,0}$} & & SR    & TRN   & MDD &   & SR    & TRN   & MDD &   & SR    & TRN   & MDD \\
\hline
    \multicolumn{14}{l}{Panel A: Markets with Better Performance for Backward-Looking Strategy} \\
\hline
    Mexico       & 0.200  &       & \textbf{0.081} & \textbf{0.279} & \textbf{0.599} &       & -0.107  & \textbf{0.278} & 0.651  &       & -0.194  & 0.483  & 0.686  \\
    Chile       & 0.300  &       & \textbf{-0.283} & \textbf{0.316} & \textbf{0.657} &       & -0.490  & 0.395  & 0.765  &       & -0.497  & 0.420  & 0.774  \\
    Italy       & 0.300  &       & \textbf{0.229} & \textbf{0.269} & \textbf{0.452} &       & 0.095  & 0.316  & 0.551  &       & 0.080  & 0.309  & 0.555  \\
    Portugal       & 0.300  &       & \textbf{0.083} & \textbf{0.058} & \textbf{0.283} &       & 0.049  & 0.086  & 0.287  &       & -0.002  & 0.149  & 0.311  \\
    Philippines       & 0.300  &       & \textbf{0.046} & \textbf{0.213} & \textbf{0.495} &       & -0.018  & 0.289  & 0.552  &       & -0.015  & 0.276  & 0.550  \\
    China       & 0.300  &       & \textbf{0.115} & 0.195  & 0.466  &       & 0.114  & \textbf{0.160} & 0.466  &       & 0.114  & 0.342  & \textbf{0.443} \\
    Australia       & 0.400  &       & \textbf{0.061} & 0.226  & \textbf{0.462} &       & -0.005  & 0.201  & 0.501  &       & -0.011  & \textbf{0.160} & 0.491  \\
    Ireland       & 0.400  &       & \textbf{0.605} & \textbf{0.097} & \textbf{0.354} &       & 0.581  & 0.122  & 0.361  &       & 0.574  & 0.130  & 0.369  \\
    Taiwan       & 0.400  &       & \textbf{0.481} & \textbf{0.024} & 0.266  &       & \textbf{0.481} & \textbf{0.024} & 0.266  &       & 0.424  & 0.397  & \textbf{0.255} \\
    Korea &    0.400  &       & \textbf{0.106} & \textbf{0.030} & \textbf{0.396} &       & \textbf{0.107} & 0.044  & \textbf{0.396} &       & -0.011  & 0.420  & 0.508  \\
    South Africa        & 0.400  &       & \textbf{-0.073} & \textbf{0.092} & \textbf{0.647} &       & \textbf{-0.075} & \textbf{0.097} & \textbf{0.647} &       & -0.231  & 0.497  & 0.779  \\
    Thailand &   0.400  &       & \textbf{0.151} & \textbf{0.073} & \textbf{0.393} &       & \textbf{0.152} & \textbf{0.071} & \textbf{0.393} &       & 0.078  & 0.821  & 0.504  \\
\hline
    \multicolumn{14}{l}{Panel B: Markets with Better Performance for Combined Strategy} \\
\hline
    Indonesia & 0.200  &       & -0.185  & 0.190  & 0.475  &       & \textbf{-0.171} & \textbf{0.144} & \textbf{0.469} &       & \textbf{-0.174} & 0.177  & \textbf{0.468} \\
    Brazil & 0.300  &       & -0.252  & 0.201  & 0.803  &       & \textbf{-0.213} & \textbf{0.093} & \textbf{0.793} &       & -0.221  & 0.183  & \textbf{0.793} \\
    Turkey & 0.300  &       & 0.347  & 0.195  & 0.307  &       & \textbf{0.379} & \textbf{0.041} & \textbf{0.301} &       & 0.350  & 0.075  & \textbf{0.301} \\
    Spain & 0.300  &       & 0.041  & 0.344  & 0.504  &       & \textbf{0.063} & 0.271  & \textbf{0.488} &       & 0.053  & \textbf{0.217} & 0.543  \\
    Korea & 0.400  &       & \textbf{0.106} & \textbf{0.030} & \textbf{0.396} &       & \textbf{0.107} & 0.044  & \textbf{0.396} &       & -0.011  & 0.420  & 0.508  \\
    Taiwan & 0.400  &       & \textbf{0.481} & \textbf{0.024} & 0.266  &       & \textbf{0.481} & \textbf{0.024} & 0.266  &       & 0.424  & 0.397  & \textbf{0.255} \\
    South Africa & 0.400  &       & \textbf{-0.073} & \textbf{0.092} & \textbf{0.647} &       & \textbf{-0.075} & 0.097  & \textbf{0.647} &       & -0.231  & 0.497  & 0.779  \\
    Thailand & 0.400  &       & \textbf{0.151} & \textbf{0.073} & \textbf{0.393} &       & \textbf{0.152} & \textbf{0.071} & \textbf{0.393} &       & 0.078  & 0.821  & 0.504  \\
    Czech Republic & 0.400  &       & -0.200  & 0.251  & 0.569  &       & \textbf{-0.168} & \textbf{0.210} & 0.552  &       & -0.199  & 0.418  & \textbf{0.532} \\
    Poland & 0.400  &       & -0.039  & \textbf{0.113} & 0.631  &       & \textbf{-0.012} & 0.211  & \textbf{0.625} &       & -0.093  & 0.280  & 0.685  \\
    India & 0.400  &       & 0.369  & 0.410  & 0.281  &       & \textbf{0.388} & 0.346  & \textbf{0.267} &       & 0.261  & \textbf{0.249} & \textbf{0.274} \\
    Malaysia & 0.400  &       & -0.028  & 0.250  & 0.280  &       & \textbf{-0.023} & 0.165  & \textbf{0.265} &       & -0.144  & \textbf{0.131} & 0.331  \\
    Austria & 0.500  &       & 0.052  & 0.164  & \textbf{0.622} &       & \textbf{0.109} & 0.128  & \textbf{0.623} &       & \textbf{0.109} & \textbf{0.095} & 0.640  \\
    HongKong & 0.500  &       & 0.100  & 0.249  & \textbf{0.269} &       & \textbf{0.164} & 0.235  & 0.288  &       & \textbf{0.157} & \textbf{0.111} & 0.298  \\
    UAE   & 0.500  &       & -0.151  & 0.376  & 0.565  &       & \textbf{0.300} & \textbf{0.118} & \textbf{0.279} &       & 0.222  & 0.239  & 0.330  \\
    Germany & 0.500  &       & 0.082  & 0.397  & 0.469  &       & \textbf{0.284} & \textbf{0.044} & \textbf{0.416} &       & \textbf{0.284} & \textbf{0.044} & \textbf{0.416} \\
    Israel & 0.500  &       & -0.382  & 0.264  & 0.657  &       & \textbf{-0.240} & \textbf{0.111} & \textbf{0.533} &       & \textbf{-0.238} & \textbf{0.109} & \textbf{0.530} \\
    Switzerland & 0.500  &       & 0.329  & 0.281  & 0.274  &       & \textbf{0.437} & \textbf{0.016} & \textbf{0.230} &       & \textbf{0.437} & \textbf{0.016} & \textbf{0.230} \\
    Singapore & 0.500  &       & -0.178  & 0.177  & 0.495  &       & \textbf{-0.099} & \textbf{0.058} & \textbf{0.458} &       & \textbf{-0.099} & \textbf{0.058} & \textbf{0.458} \\
\hline
    \multicolumn{14}{l}{Panel C: Markets with Better Performance for Forward-Looking Strategy} \\
\hline
    Germany & 0.500  &       & 0.082  & 0.397  & 0.469  &       & \textbf{0.284} & \textbf{0.044} & \textbf{0.416} &       & \textbf{0.284} & \textbf{0.044} & \textbf{0.416} \\
    Israel & 0.500  &       & -0.382  & 0.264  & 0.657  &       & \textbf{-0.240} & \textbf{0.111} & \textbf{0.533} &       & \textbf{-0.238} & \textbf{0.109} & \textbf{0.530} \\
    Switzerland & 0.500  &       & 0.329  & 0.281  & 0.274  &       & \textbf{0.437} & \textbf{0.016} & \textbf{0.230} &       & \textbf{0.437} & \textbf{0.016} & \textbf{0.230} \\
    Singapore & 0.500  &       & -0.178  & 0.177  & 0.495  &       & \textbf{-0.099} & \textbf{0.058} & \textbf{0.458} &       & \textbf{-0.099} & \textbf{0.058} & \textbf{0.458} \\
    Norway & 0.500  &       & 0.185  & \textbf{0.139} & 0.545  &       & 0.184  & 0.172  & 0.541  &       & \textbf{0.211} & 0.340  & \textbf{0.529} \\
    Canada & 0.500  &       & -0.036  & 0.236  & \textbf{0.475} &       & -0.037  & 0.187  & 0.487  &       & \textbf{-0.025} & \textbf{0.090} & \textbf{0.473} \\
    Belgium & 0.500  &       & 0.262  & 0.276  & 0.461  &       & 0.264  & \textbf{0.268 } & 0.460  &       & \textbf{0.382} & 0.298  & \textbf{0.355} \\
    France & 0.500  &       & 0.333  & 0.115  & \textbf{0.331} &       & 0.359  & 0.050  & \textbf{0.331} &       & \textbf{0.377} & \textbf{0.021} & \textbf{0.331} \\
    UK    & 0.500  &       & -0.027  & 0.073  & 0.528  &       & -0.030  & \textbf{0.038} & 0.528  &       & \textbf{-0.025} & 0.131  & \textbf{0.511} \\
    Russia & 0.500  &       & -0.123  & 0.167  & 0.660  &       & -0.090  & 0.106  & 0.639  &       & \textbf{-0.062} & \textbf{0.073} & \textbf{0.620} \\
    Japan & 0.500  &       & 0.056  & 0.686  & 0.356  &       & 0.247  & 0.397  & 0.252  &       & \textbf{0.343} & \textbf{0.272} & \textbf{0.228} \\
    US    & 0.600  &       & 0.684  & 0.282  & 0.165  &       & 0.828  & 0.045  & \textbf{0.143} &       & \textbf{0.849} & \textbf{0.020} & \textbf{0.143} \\
\hline
    \end{tabular}}%
\end{table}%

Most markets reported in Panel A are the emerging markets with lower $\alpha_{I,0}$ (0.200-0.400) such as Mexico and Thailand, while most markets reported in Panel C are developed markets with higher $\alpha_{I,0}$ (0.500-0.600) such as the US and Germany, which shows that backward (forward)-looking estimation predicts better in emerging (developed) markets. An intuitive explanation is that in an efficient market, it is difficult to predict future returns with historical data, and the market portfolio is close to the optimal portfolio. In contrast, in emerging markets, the forward-looking information in the market portfolio is less valuable due to the noise trader, and backward-looking information plays a more important role in investment.

The combined strategy achieves the best out-of-sample performance in a wide range of markets, as shown in Panel B. An interesting finding is that in developed markets such as Germany and Switzerland, the combined strategy is similar to the forward-looking strategies, while the combined strategy is similar to the backward-looking strategy in emerging markets such as South Africa and Thailand.
In addition, in markets shown in Panels A and C, the performance of the combined strategy is closer to the best strategy. These pieces of evidence show that our combined strategy can adaptively integrate different sources of information based on market characteristics.

Overall, although the backward (forward)-looking strategy has better performance in some emerging (developed) markets, the combined strategy outperforms in a wider range of markets, as shown in Panel B. Even if in markets reported in Panels A and C, the combined strategy is close to the best strategy and does not behave the worst, which confirms its flexibility and adaptivity. Therefore, our combined strategy shows the potential of improving portfolio out-of-sample performance.

\section{Conclusion}\label{Sec5}

In this paper, we propose a novel combined approach to estimate the future returns of assets and test the out-of-sample performance of investment decisions. The forecasting work is based on a heterogeneous market model with an informed investor, a less-informed investor, and a noise trader, and we derive the market equilibrium according to the decision behavior of investors.
By the Gaussian mixture model, we  extract forward-looking information from the market portfolio in the spirit of inverse optimization. With the historical prior, we maximize the posterior probability and obtain the combined estimation.

The combined estimation is influenced by three market characteristics including the market share of different investors, the intensity of noise trading, and the risk-aversion attitude of the informed investor.
We show through both theoretical analysis and simulation experiments that the combined approach exhibits self-adaptability in integrating different sources of information according to the degree of market efficiency.
Empirical analyses across 35 constituent markets of the MSCI ACWI Index are performed. While the backward-looking (forward-looking) strategy has better performance in emerging (developed) markets, we show that the combined strategy can be well applied to a wider range of stock markets, not limited to emerging or developed markets.

In summary, we attempt to provide a new research paradigm for return forecast in portfolio management.
Note that neither the market portfolio nor the Gaussian mixture model in this paper is unalterable; in fact, our research paradigm can also be applied to other market variables containing expectation information, other statistical method used to model asset returns, and even the framework other than mean-variance analysis. Uncovering more effective forecasting information from financial markets and generating more robust and practical optimal portfolios remains a longstanding issue in both practical portfolio management and academic research on portfolio theory.

\section{Appendix}
\subsection{Posterior distribution in the special case}

The following proposition gives an explicit posterior distribution of $\bm{\lambda}_-$.

\begin{proposition}
Suppose $\bm{\lambda}_- \sim\mathcal{N}\left(\hat{\bm{\lambda}}_-,\Phi\right)$, $\bm{x}_M=\bm{q}_0+P\bm{\lambda}_-+\bm{\varepsilon}$ where
$\bm{\varepsilon}\sim\mathcal{N}\left(\bm{0}_n,\Omega\right)$ and $\mathbb{C}ov\left(\bm{\lambda}_-,\bm{\varepsilon}\right)=\bm{0}$. If $\Phi$ and $\Omega$ are invertible,
then the conditional distribution of $\bm{\lambda}_-$ for given $\bm{x}_M$ becomes
\[
\bm{\lambda}_-|{\bm{x}_M} \sim\mathcal{N}
\left(\left(\Phi^{-1}+P'\Omega^{-1}P\right)^{-1}
\left(\Phi^{-1}\hat{\bm{\lambda}}_-+P'\Omega^{-1}({\bm{x}_M}-\bm{q}_0)\right),\left(\Phi^{-1}+P'\Omega^{-1}P\right)^{-1}\right).
\]
\end{proposition}

\begin{proof}

With Bayes' rule, we can rewrite the posterior probability as
\begin{eqnarray*}
\bm{\lambda}_-|{\bm{x}_M}
&\propto& p\left({\bm{x}_M}|\bm{\lambda}_-\right)p\left(\bm{\lambda}_-\right) \nonumber\\
&\propto& exp\left\{-\frac{1}{2}\left(\bm{\lambda}_--\hat{\bm{\lambda}}_-\right)'\Phi^{-1}
\left(\bm{\lambda}_--\hat{\bm{\lambda}}_-\right)
-\frac{1}{2}\left({\bm{x}_M}-\bm{q}_0-P\bm{\lambda}_-\right)'\Omega^{-1}
\left({\bm{x}_M}-\bm{q}_0-P\bm{\lambda}_-\right)\right\} \nonumber\\
\notag &\propto& exp\left\{-\frac{1}{2}\left(\bm{\lambda}_--\Lambda^{-1}\bm{\upsilon}\right)'
\Lambda\left(\bm{\lambda}_--\Lambda^{-1}\bm{\upsilon}\right)\right\}.
\end{eqnarray*}
Therefore, the posterior distribution of $\bm{\lambda}_-$ is
\[
\bm{\lambda}_-|{\bm{x}_M} \sim\mathcal{N}
\left(\Lambda^{-1}\bm{\upsilon},\Lambda^{-1}\right).
\]
\end{proof}

\subsection{Empirical results}

Table \ref{table:5} reports the detailed out-of-sample indicators of six types of investment strategies on 35 major stock markets over the world.
\linespread{1.06}
\footnotesize
\begin{longtable}{lccccccccc}
\captionsetup{labelsep = period, labelfont = bf, font = {small,stretch=1}}
 \caption{Empirical tests: Comparisons of the mean-variance strategies and benchmark strategies (1/N, Index, and Minimum-variance). Results with three types of risk-aversion settings are shown. This table reports out-of-sample performance indicators, including annualized SR, TRN and MDD. Each indicator is the average across the three rebalancing windows.}\label{table:5}\\
\hline
    \multicolumn{1}{c}{\multirow{2}[3]{*}{Strategy}}      & \multicolumn{3}{c}{\underline{Aggressive Investor}}       & \multicolumn{3}{c}{\underline{Moderate Investor}}       & \multicolumn{3}{c}{\underline{Conservative Investor}} \\
& SR    & TRN   & MDD          & SR    & TRN   & MDD          & SR    & TRN   & MDD \\
\hline
\endfirsthead

\specialrule{0em}{0pt}{11.06pt}

\multicolumn{10}{l}{ {\bf Table \ref{table:5}.} Continued}\\

\specialrule{0em}{0pt}{6.03pt}

\hline
    \multicolumn{1}{c}{\multirow{2}[3]{*}{Strategy}}      & \multicolumn{3}{c}{\underline{Aggressive Investor}}       & \multicolumn{3}{c}{\underline{Moderate Investor}}       & \multicolumn{3}{c}{\underline{Conservative Investor}} \\
& SR    & TRN   & MDD          & SR    & TRN   & MDD          & SR    & TRN   & MDD \\
\hline
\endhead

\hline
\endfoot

\multicolumn{4}{l}{\underline{Panel A: Americas Markets}} \\
    \multicolumn{10}{c}{Canada: $\alpha_{I,0}=0.500, \sigma=1.041$} \\
\hline
    Backward-looking & -0.337  & 0.237  & 0.798  & -0.036  & 0.236  & \textbf{0.475 } & \textbf{0.181} & \textbf{0.316} & \textbf{0.282} \\
    Combined & \textbf{-0.323} & 0.238  & 0.790  & -0.037  & 0.187  & 0.487  & 0.171  & \textbf{0.321} & \textbf{0.280} \\
    Forward-looking & \textbf{-0.320} & \textbf{0.206} & \textbf{0.783} & \textbf{-0.025} & \textbf{0.090} & \textbf{0.473} & \textbf{0.180} & \textbf{0.319} & \textbf{0.277} \\
    \hline
    1/N   & 0.112  & \textbf{0.060} & 0.407  & 0.112  & \textbf{0.060} & 0.407  & 0.112  & \textbf{0.060} & 0.407  \\
    Index & -0.016  & -     & 0.466  & -0.016  & -     & 0.466  & -0.016  & -     & 0.466  \\
    Minimum-variance & \textbf{0.214} & 0.312  & \textbf{0.252} & \textbf{0.214} & 0.312  & \textbf{0.252} & \textbf{0.214} & 0.312  & \textbf{0.252} \\
    \hline
    \multicolumn{10}{c}{US: $\alpha_{I,0}=0.600, \sigma=0.260$} \\
    \hline
    Backward-looking & 0.484  & 0.506  & 0.226  & 0.684  & 0.282  & 0.165  & 0.790  & 0.336  & 0.115  \\
    Combined & 0.585  & \textbf{0.352} & 0.222  & 0.828  & 0.045  & \textbf{0.143} & \textbf{0.797} & \textbf{0.301} & \textbf{0.113} \\
    Forward-looking & \textbf{0.615} & \textbf{0.350} & \textbf{0.204} & \textbf{0.849} & \textbf{0.020} & \textbf{0.143} & \textbf{0.799} & \textbf{0.301} & \textbf{0.111} \\
    \hline
    1/N   & 0.705  & \textbf{0.036} & 0.145  & 0.705  & \textbf{0.036} & 0.145  & 0.705  & \textbf{0.036 } & 0.145  \\
    Index & \textbf{0.849} & -     & 0.143  & \textbf{0.849} & -     & 0.143  & \textbf{0.849} & -     & 0.143  \\
    Minimum-variance & 0.835  & 0.297  & \textbf{0.103} & 0.835  & 0.297  & \textbf{0.103} & 0.835  & 0.297  & \textbf{0.103} \\
    \hline
    \multicolumn{10}{c}{Brazil: $\alpha_{I,0}=0.300, \sigma=4.165$} \\
    \hline
    Backward-looking & -0.199  & 0.251  & \textbf{0.861} & -0.252  & 0.201  & 0.803  & -0.309  & \textbf{0.215} & 0.716  \\
    Combined & \textbf{-0.156} & \textbf{0.195} & \textbf{0.856} & \textbf{-0.213} & \textbf{0.093} & \textbf{0.793} & \textbf{-0.304} & \textbf{0.212} & \textbf{0.713} \\
    Forward-looking & -0.181  & 0.246  & \textbf{0.860} & -0.221  & 0.183  & \textbf{0.793} & -0.310  & 0.247  & 0.720  \\
    \hline
    1/N   & \textbf{-0.186} & \textbf{0.073} & 0.775  & \textbf{-0.186} & \textbf{0.073} & 0.775  & \textbf{-0.186} & \textbf{0.073} & 0.775  \\
    Index & -0.204  & -     & 0.794  & -0.204  & -     & 0.794  & -0.204  & -     & 0.794  \\
    Minimum-variance & -0.336  & 0.227  & \textbf{0.718} & -0.336  & 0.227  & \textbf{0.718} & -0.336  & 0.227  & \textbf{0.718 } \\
    \hline
    \multicolumn{10}{c}{Chile: $\alpha_{I,0}=0.300, \sigma=4.165$} \\
    \hline
    Backward-looking & \textbf{-0.269} & \textbf{0.635} & \textbf{0.752} & \textbf{-0.283} & \textbf{0.316} & \textbf{0.657} & \textbf{-0.302} & \textbf{0.293} & \textbf{0.620} \\
    Combined & -0.485  & 0.651  & 0.808  & -0.490  & 0.395  & 0.765  & -0.449  & 0.394  & 0.724  \\
    Forward-looking & -0.442  & 0.650  & 0.786  & -0.497  & 0.420  & 0.774  & -0.442  & 0.423  & 0.735  \\
    \hline
    1/N   & -0.326  & \textbf{0.072} & 0.714  & -0.326  & \textbf{0.072} & 0.714  & -0.326  & \textbf{0.072} & 0.714  \\
    Index & \textbf{-0.294} & -     & 0.654  & \textbf{-0.294} & -     & 0.654  & \textbf{-0.294} & -     & 0.654  \\
    Minimum-variance & -0.343  & 0.311  & \textbf{0.648} & -0.343  & 0.311  & \textbf{0.648} & -0.343  & 0.311  & \textbf{0.648} \\
    \hline
    \multicolumn{10}{c}{Mexico: $\alpha_{I,0}=0.200, \sigma=6.508$} \\
    \hline
    Backward-looking & \textbf{0.120} & \textbf{0.417} & \textbf{0.642} & \textbf{0.081} & \textbf{0.279} & \textbf{0.599} & \textbf{-0.081 } & \textbf{0.226} & \textbf{0.618} \\
    Combined & -0.014  & 0.456  & 0.676  & -0.107  & \textbf{0.278} & 0.651  & -0.167  & 0.286  & 0.652  \\
    Forward-looking & -0.172  & 0.535  & 0.720  & -0.194  & 0.483  & 0.686  & -0.200  & 0.447  & 0.673  \\
    \hline
    1/N   & -0.109  & \textbf{0.065} & 0.728  & -0.109  & \textbf{0.065} & 0.728  & -0.109  & \textbf{0.065} & 0.728  \\
    Index & \textbf{-0.066} & -     & \textbf{0.630} & \textbf{-0.066} & -     & \textbf{0.630} & \textbf{-0.066} & -     & \textbf{0.630} \\
    Minimum-variance & -0.154  & 0.204  & 0.650  & -0.154  & 0.204  & 0.650  & -0.154  & 0.204  & 0.650  \\
    \hline
\multicolumn{4}{l}{\underline{Panel B: Europe \& Middle East Markets}} \\
    \multicolumn{10}{c}{Austria: $\alpha_{I,0}=0.500, \sigma=1.041$} \\
    \hline
    Backward-looking & 0.069  & 0.346  & 0.642  & 0.052  & 0.164  & \textbf{0.622} & 0.001  & 0.186  & \textbf{0.621} \\
    Combined & \textbf{0.124} & 0.331  & \textbf{0.635} & \textbf{0.109} & 0.128  & \textbf{0.623} & \textbf{0.014} & 0.179  & \textbf{0.618} \\
    Forward-looking & 0.106  & \textbf{0.265} & 0.645  & \textbf{0.109} & \textbf{0.095} & 0.640  & \textbf{0.006} & \textbf{0.169} & 0.639  \\
    1/N   & \textbf{0.136} & \textbf{0.078} & \textbf{0.592} & \textbf{0.136} & \textbf{0.078} & \textbf{0.592} & \textbf{0.136} & \textbf{0.078} & \textbf{0.592} \\
    Index & 0.107  & -     & 0.635  & 0.107  & -     & 0.635  & 0.107  & -     & 0.635  \\
    Minimum-variance & -0.049  & 0.203  & 0.649  & -0.049  & 0.203  & 0.649  & -0.049  & 0.203  & 0.649  \\
    \hline
    \multicolumn{10}{c}{Belgium: $\alpha_{I,0}=0.500, \sigma=1.041$} \\
    \hline
    Backward-looking & 0.099  & \textbf{0.344} & 0.469  & 0.262  & 0.276  & 0.461  & 0.270  & 0.313  & 0.339  \\
    Combined & 0.089  & \textbf{0.345} & 0.479  & 0.264  & \textbf{0.268} & 0.460  & 0.266  & 0.312  & 0.341  \\
    Forward-looking & \textbf{0.150} & \textbf{0.340} & \textbf{0.428} & \textbf{0.382} & 0.298  & \textbf{0.355} & \textbf{0.377} & \textbf{0.304} & \textbf{0.274} \\
    \hline
    1/N   & \textbf{0.412} & \textbf{0.061} & \textbf{0.257} & \textbf{0.412} & \textbf{0.061} & \textbf{0.257} & \textbf{0.412} & \textbf{0.061} & \textbf{0.257} \\
    Index & 0.381  & -     & 0.381  & 0.381  & -     & 0.381  & 0.381  & -     & 0.381  \\
    Minimum-variance & 0.331  & 0.317  & 0.268  & 0.331  & 0.317  & 0.268  & 0.331  & 0.317  & 0.268  \\
    \hline
    \multicolumn{10}{c}{France: $\alpha_{I,0}=0.500, \sigma=1.041$} \\
    \hline
    Backward-looking & \textbf{0.436} & 0.344  & \textbf{0.356} & 0.333  & 0.115  & \textbf{0.331} & \textbf{0.381} & 0.265  & \textbf{0.305} \\
    Combined & \textbf{0.438 } & \textbf{0.318} & \textbf{0.356} & 0.359  & 0.050  & \textbf{0.331} & \textbf{0.381} & \textbf{0.260} & \textbf{0.305} \\
    Forward-looking & \textbf{0.438} & \textbf{0.318} & \textbf{0.356} & \textbf{0.377} & \textbf{0.021} & \textbf{0.331} & \textbf{0.380} & \textbf{0.258} & \textbf{0.305} \\
    \hline
    1/N   & 0.317  & \textbf{0.046} & 0.335  & 0.317  & \textbf{0.046} & 0.335  & 0.317  & \textbf{0.046} & 0.335  \\
    Index & 0.377  & -     & 0.331  & 0.377  & -     & 0.331  & 0.377  & -     & 0.331  \\
    Minimum-variance & \textbf{0.393} & 0.324  & \textbf{0.298} & \textbf{0.393} & 0.324  & \textbf{0.298} & \textbf{0.393} & 0.324  & \textbf{0.298} \\
    \hline
    \multicolumn{10}{c}{Germany: $\alpha_{I,0}=0.500, \sigma=1.041$} \\
    \hline
    Backward-looking & 0.147  & 0.507  & 0.599  & 0.082  & 0.397  & 0.469  & 0.349  & 0.337  & 0.350  \\
    Combined & \textbf{0.198} & \textbf{0.317} & \textbf{0.571} & \textbf{0.284} & \textbf{0.044} & \textbf{0.416} & \textbf{0.407} & \textbf{0.252} & \textbf{0.339} \\
    Forward-looking & \textbf{0.198} & \textbf{0.317} & \textbf{0.571} & \textbf{0.284} & \textbf{0.044} & \textbf{0.416} & \textbf{0.407} & \textbf{0.252} & \textbf{0.339} \\
    \hline
    1/N   & 0.307  & \textbf{0.050} & 0.369  & 0.307  & \textbf{0.050} & 0.369  & 0.307  & \textbf{0.050} & 0.369  \\
    Index & 0.288  & -     & 0.413  & 0.288  & -     & 0.413  & 0.288  & -     & 0.413  \\
    Minimum-variance & \textbf{0.455} & 0.292  & \textbf{0.312} & \textbf{0.455} & 0.292  & \textbf{0.312} & \textbf{0.455} & 0.292  & \textbf{0.312} \\
    \hline
    \multicolumn{10}{c}{Ireland: $\alpha_{I,0}=0.400, \sigma=2.343$} \\
    \hline
    Backward-looking & \textbf{0.386} & \textbf{0.238} & 0.506  & \textbf{0.605} & \textbf{0.097} & \textbf{0.354} & \textbf{0.721} & 0.098  & \textbf{0.297} \\
    Combined & 0.354  & 0.291  & \textbf{0.500} & 0.581  & 0.122  & 0.361  & 0.700  & \textbf{0.092} & \textbf{0.302 } \\
    Forward-looking & 0.348  & 0.283  & \textbf{0.497} & 0.574  & 0.130  & 0.369  & 0.700  & \textbf{0.088} & 0.309  \\
    \hline
    1/N   & \textbf{0.737} & \textbf{0.060} & 0.318  & \textbf{0.737} & \textbf{0.060} & 0.318  & \textbf{0.737} & \textbf{0.060} & 0.318  \\
    Index & 0.596  & -     & 0.374  & 0.596  & -     & 0.374  & 0.596  & -     & 0.374  \\
    Minimum-variance & 0.715  & 0.111  & \textbf{0.287} & 0.715  & 0.111  & \textbf{0.287} & 0.715  & 0.111  & \textbf{0.287} \\
    \hline
    \multicolumn{10}{c}{Israel: $\alpha_{I,0}=0.500, \sigma=1.041$} \\
    \hline
    Backward-looking & -0.546  & 0.324  & 0.864  & -0.382  & 0.264  & 0.657  & -0.135  & 0.192  & 0.462  \\
    Combined & -0.455  & \textbf{0.270} & 0.860  & \textbf{-0.240} & \textbf{0.111} & \textbf{0.533} & \textbf{-0.068} & 0.136  & \textbf{0.392} \\
    Forward-looking & \textbf{-0.448} & 0.277  & \textbf{0.853} & \textbf{-0.238} & \textbf{0.109} & \textbf{0.530} & -0.082  & \textbf{0.129} & 0.400  \\
    \hline
    1/N   & -0.154  & \textbf{0.094} & 0.474  & -0.154  & \textbf{0.094} & 0.474  & -0.154  & \textbf{0.094} & 0.474  \\
    Index & -0.184  & -     & 0.526  & -0.184  & -     & 0.526  & -0.184  & -     & 0.526  \\
    Minimum-variance & \textbf{0.050} & 0.133  & \textbf{0.386} & \textbf{0.050} & 0.133  & \textbf{0.386 } & \textbf{0.050} & 0.133  & \textbf{0.386} \\
    \hline
    \multicolumn{10}{c}{Italy: $\alpha_{I,0}=0.300, \sigma=4.165$} \\
    \hline
    Backward-looking & \textbf{0.251} & \textbf{0.358} & \textbf{0.548} & \textbf{0.229} & \textbf{0.269} & \textbf{0.452} & \textbf{0.211} & \textbf{0.228} & \textbf{0.408} \\
    Combined & 0.037  & 0.457  & 0.732  & 0.095  & 0.316  & 0.551  & 0.159  & 0.331  & 0.415  \\
    Forward-looking & 0.041  & 0.456  & 0.729  & 0.080  & 0.309  & 0.555  & 0.157  & 0.321  & \textbf{0.415} \\
    \hline
    1/N   & 0.147  & \textbf{0.064} & 0.472  & 0.147  & \textbf{0.064} & 0.472  & 0.147  & \textbf{0.064 } & 0.472  \\
    Index & 0.102  & -     & 0.527  & 0.102  & -     & 0.527  & 0.102  & -     & 0.527  \\
    Minimum-variance & \textbf{0.173} & 0.284  & \textbf{0.407} & \textbf{0.173} & 0.284  & \textbf{0.407} & \textbf{0.173} & 0.284  & \textbf{0.407} \\
    \hline
    \multicolumn{10}{c}{Norway: $\alpha_{I,0}=0.500, \sigma=1.041$} \\
    \hline
    Backward-looking & 0.030  & \textbf{0.211} & 0.727  & 0.185  & \textbf{0.139} & 0.545  & 0.271  & \textbf{0.255} & 0.392  \\
    Combined & 0.054  & 0.249  & 0.726  & 0.184  & 0.172  & 0.541  & 0.253  & 0.303  & 0.392  \\
    Forward-looking & \textbf{0.081} & 0.426  & \textbf{0.634} & \textbf{0.211} & 0.340  & \textbf{0.529} & \textbf{0.342} & 0.301  & \textbf{0.382} \\
    \hline
    1/N   & \textbf{0.411} & \textbf{0.065} & 0.441  & \textbf{0.411} & \textbf{0.065} & 0.441  & \textbf{0.411} & \textbf{0.065} & 0.441  \\
    Index & 0.143  & -     & 0.548  & 0.143  & -     & 0.548  & 0.143  & -     & 0.548  \\
    Minimum-variance & 0.264  & 0.239  & \textbf{0.343} & 0.264  & 0.239  & \textbf{0.343} & 0.264  & 0.239  & \textbf{0.343} \\
    \hline
    \multicolumn{10}{c}{Portugal: $\alpha_{I,0}=0.300, \sigma=4.165$} \\
    \hline
    Backward-looking & \textbf{0.202} & \textbf{0.294} & \textbf{0.278} & \textbf{0.083} & \textbf{0.058} & \textbf{0.283} & \textbf{-0.012} & \textbf{0.100} & \textbf{0.303} \\
    Combined & 0.168  & 0.338  & 0.295  & 0.049  & 0.086  & 0.287  & -0.030  & 0.115  & \textbf{0.303} \\
    Forward-looking & 0.002  & 0.354  & 0.324  & -0.002  & 0.149  & 0.311  & -0.047  & 0.125  & 0.313  \\
    \hline
    1/N   & -0.069  & \textbf{0.074} & \textbf{0.244} & -0.069  & \textbf{0.074} & \textbf{0.244} & -0.069  & \textbf{0.074 } & \textbf{0.244} \\
    Index & \textbf{0.015} & -     & 0.286  & \textbf{0.015} & -     & 0.286  & \textbf{0.015} & -     & 0.286  \\
    Minimum-variance & -0.054  & 0.148  & 0.319  & -0.054  & 0.148  & 0.319  & -0.054  & 0.148  & 0.319  \\
    \hline
    \multicolumn{10}{c}{Spain: $\alpha_{I,0}=0.300, \sigma=4.165$} \\
    \hline
    Backward-looking & -0.123  & 0.538  & 0.709  & 0.041  & 0.344  & 0.504  & 0.232  & 0.264  & \textbf{0.453 } \\
    Combined & -0.097  & 0.498  & \textbf{0.685} & \textbf{0.063} & 0.271  & \textbf{0.488} & \textbf{0.241} & \textbf{0.245} & \textbf{0.446} \\
    Forward-looking & \textbf{-0.083} & \textbf{0.137} & 0.752  & 0.053  & \textbf{0.217} & 0.543  & \textbf{0.239} & 0.271  & \textbf{0.453} \\
    \hline
    1/N   & 0.139  & \textbf{0.047} & 0.496  & 0.139  & \textbf{0.047} & 0.496  & 0.139  & \textbf{0.047} & 0.496  \\
    Index & 0.084  & -     & 0.532  & 0.084  & -     & 0.532  & 0.084  & -     & 0.532  \\
    Minimum-variance & \textbf{0.294} & 0.257  & \textbf{0.445} & \textbf{0.294} & 0.257  & \textbf{0.445} & \textbf{0.294} & 0.257  & \textbf{0.445} \\
    \hline
    \multicolumn{10}{c}{Switzerland: $\alpha_{I,0}=0.500, \sigma=1.041$} \\
    \hline
    Backward-looking & 0.213  & 0.507  & 0.337  & 0.329  & 0.281  & 0.274  & 0.426  & 0.280  & 0.247  \\
    Combined & \textbf{0.327} & \textbf{0.313} & \textbf{0.324} & \textbf{0.437} & \textbf{0.016} & \textbf{0.230} & \textbf{0.487} & \textbf{0.203} & \textbf{0.200} \\
    Forward-looking & \textbf{0.327} & \textbf{0.313} & \textbf{0.324} & \textbf{0.437} & \textbf{0.016} & \textbf{0.230} & \textbf{0.487} & \textbf{0.203} & \textbf{0.200} \\
    \hline
    1/N   & \textbf{0.537} & \textbf{0.065} & 0.236  & \textbf{0.537} & \textbf{0.065} & 0.236  & \textbf{0.537} & \textbf{0.065} & 0.236  \\
    Index & 0.436  & -     & 0.230  & 0.436  & -     & 0.230  & 0.436  & -     & 0.230  \\
    Minimum-variance & 0.496  & 0.259  & \textbf{0.200} & 0.496  & 0.259  & \textbf{0.200} & 0.496  & 0.259  & \textbf{0.200} \\
    \hline
    \multicolumn{10}{c}{UK: $\alpha_{I,0}=0.500, \sigma=1.041$} \\
    \hline
    Backward-looking & \textbf{0.056} & \textbf{0.386} & 0.598  & -0.027  & 0.073  & 0.528  & 0.187  & \textbf{0.303} & 0.347  \\
    Combined & \textbf{0.055} & \textbf{0.386} & 0.598  & -0.030  & \textbf{0.038} & 0.528  & 0.186  & \textbf{0.303 } & 0.347  \\
    Forward-looking & \textbf{0.064} & 0.432  & \textbf{0.575} & \textbf{-0.025} & 0.131  & \textbf{0.511} & \textbf{0.207} & 0.328  & \textbf{0.330} \\
    \hline
    1/N   & 0.131  & \textbf{0.057} & 0.487  & 0.131  & \textbf{0.057} & 0.487  & 0.131  & \textbf{0.057} & 0.487  \\
    Index & -0.030  & -     & 0.527  & -0.030  & -     & 0.527  & -0.030  & -     & 0.527  \\
    Minimum-variance & \textbf{0.195} & 0.355  & \textbf{0.320} & \textbf{0.195} & 0.355  & \textbf{0.320} & \textbf{0.195} & 0.355  & \textbf{0.320} \\
    \hline
    \multicolumn{10}{c}{Czech Republic: $\alpha_{I,0}=0.400, \sigma=2.343$} \\
    \hline
    Backward-looking & -0.107  & \textbf{0.326} & 0.548  & -0.200  & 0.251  & 0.569  & \textbf{-0.253} & 0.225  & 0.608  \\
    Combined & \textbf{-0.056} & 0.362  & \textbf{0.508} & \textbf{-0.168} & \textbf{0.210} & 0.552  & \textbf{-0.251} & \textbf{0.188} & 0.606  \\
    Forward-looking & -0.155  & 0.427  & 0.566  & -0.199  & 0.418  & \textbf{0.532} & -0.286  & 0.283  & \textbf{0.577} \\
    \hline
    1/N   & \textbf{-0.094} & \textbf{0.064} & \textbf{0.507} & \textbf{-0.094} & \textbf{0.064} & \textbf{0.507} & \textbf{-0.094} & \textbf{0.064} & \textbf{0.507} \\
    Index & -0.227  & -     & 0.580  & -0.227  & -     & 0.580  & -0.227  & -     & 0.580  \\
    Minimum-variance & -0.318  & 0.185  & 0.654  & -0.318  & 0.185  & 0.654  & -0.318  & 0.185  & 0.654  \\
    \hline
    \multicolumn{10}{c}{Poland: $\alpha_{I,0}=0.400, \sigma=2.343$} \\
    \hline
    Backward-looking & -0.064  & \textbf{0.310} & \textbf{0.754} & -0.039  & \textbf{0.113} & 0.631  & -0.154  & \textbf{0.153} & \textbf{0.613} \\
    Combined & \textbf{-0.060} & 0.389  & \textbf{0.746} & \textbf{-0.012} & 0.211  & \textbf{0.625} & \textbf{-0.132} & 0.192  & \textbf{0.613} \\
    Forward-looking & -0.139  & 0.412  & 0.800  & -0.093  & 0.280  & 0.685  & -0.156  & 0.233  & 0.635  \\
    \hline
    1/N   & \textbf{-0.079} & \textbf{0.088} & \textbf{0.604} & \textbf{-0.079} & \textbf{0.088} & \textbf{0.604} & \textbf{-0.079} & \textbf{0.088} & \textbf{0.604} \\
    Index & \textbf{-0.077} & -     & 0.631  & \textbf{-0.077} & -     & 0.631  & \textbf{-0.077} & -     & 0.631  \\
    Minimum-variance & -0.211  & 0.213  & 0.650  & -0.211  & 0.213  & 0.650  & -0.211  & 0.213  & 0.650  \\
    \hline
    \multicolumn{10}{c}{Russia: $\alpha_{I,0}=0.500, \sigma=1.041$} \\
    \hline
    Backward-looking & -0.148  & 0.379  & 0.758  & -0.123  & 0.167  & 0.660  & -0.037  & 0.223  & 0.601  \\
    Combined & -0.120  & 0.370  & 0.740  & -0.090  & 0.106  & 0.639  & \textbf{-0.029} & \textbf{0.207} & 0.597  \\
    Forward-looking & \textbf{-0.081} & \textbf{0.329} & \textbf{0.724} & \textbf{-0.062} & \textbf{0.073} & \textbf{0.620} & -0.039  & \textbf{0.208} & \textbf{0.592} \\
    \hline
    1/N   & -0.130  & \textbf{0.080} & 0.668  & -0.130  & \textbf{0.080} & 0.668  & -0.130  & \textbf{0.080} & 0.668  \\
    Index & -0.071  & -     & 0.619  & -0.071  & -     & 0.619  & -0.071  & -     & 0.619  \\
    Minimum-variance & \textbf{0.008} & 0.248  & \textbf{0.601} & \textbf{0.008} & 0.248  & \textbf{0.601} & \textbf{0.008} & 0.248  & \textbf{0.601} \\
    \hline
    \multicolumn{10}{c}{South Africa: $\alpha_{I,0}=0.4, \sigma=2.343$} \\
    \hline
    Backward-looking & \textbf{-0.065} & \textbf{0.449} & 0.744  & \textbf{-0.073} & \textbf{0.092} & \textbf{0.647} & \textbf{-0.214} & \textbf{0.357} & \textbf{0.726} \\
    Combined & \textbf{-0.065} & \textbf{0.451 } & 0.744  & \textbf{-0.075} & 0.097  & \textbf{0.647} & \textbf{-0.215} & \textbf{0.357} & \textbf{0.727} \\
    Forward-looking & -0.330  & 0.385  & \textbf{0.725} & -0.231  & 0.497  & 0.779  & -0.354  & 0.658  & 0.847  \\
    \hline
    1/N   & -0.215  & \textbf{0.090} & 0.755  & -0.215  & \textbf{0.090} & 0.755  & -0.215  & \textbf{0.090} & 0.755  \\
    Index & \textbf{-0.052} & -     & \textbf{0.646} & \textbf{-0.052} & -     & \textbf{0.646} & \textbf{-0.052} & -     & \textbf{0.646} \\
    Minimum-variance & -0.255  & 0.387  & 0.741  & -0.255  & 0.387  & 0.741  & -0.255  & 0.387  & 0.741  \\
    \hline
    \multicolumn{10}{c}{Turkey: $\alpha_{I,0}=0.300, \sigma=4.165$} \\
    \hline
    Backward-looking & \textbf{0.205} & \textbf{0.061} & \textbf{0.472} & 0.347  & 0.195  & 0.307  & 0.462  & 0.223  & 0.212  \\
    Combined & \textbf{0.211} & \textbf{0.059} & \textbf{0.473} & \textbf{0.379} & \textbf{0.041} & \textbf{0.301} & \textbf{0.482} & \textbf{0.194} & \textbf{0.204} \\
    Forward-looking & \textbf{0.211} & \textbf{0.059} & \textbf{0.473} & 0.350  & 0.075  & \textbf{0.301} & \textbf{0.478} & \textbf{0.189} & \textbf{0.204} \\
    \hline
    1/N   & \textbf{0.568} & \textbf{0.058} & 0.230  & \textbf{0.568} & \textbf{0.058} & 0.230  & \textbf{0.568} & \textbf{0.058} & 0.230  \\
    Index & 0.379  & -     & 0.301  & 0.379  & -     & 0.301  & 0.379  & -     & 0.301  \\
    Minimum-variance & 0.501  & 0.195  & \textbf{0.191} & 0.501  & 0.195  & \textbf{0.191} & 0.501  & 0.195  & \textbf{0.191} \\
    \hline
    \multicolumn{10}{c}{UAE: $\alpha_{I,0}=0.500, \sigma=1.041$} \\
    \hline
    Backward-looking & -0.139  & 0.426  & 0.676  & -0.151  & 0.376  & 0.565  & -0.084  & 0.249  & 0.370  \\
    Combined & \textbf{0.248} & \textbf{0.282} & \textbf{0.432} & \textbf{0.300} & \textbf{0.118} & \textbf{0.279} & \textbf{0.237} & \textbf{0.168} & \textbf{0.189} \\
    Forward-looking & 0.195  & 0.404  & 0.485  & 0.222  & 0.239  & 0.330  & 0.200  & 0.214  & 0.197  \\
    \hline
    1/N   & \textbf{0.334} & \textbf{0.068} & 0.287  & \textbf{0.334} & \textbf{0.068} & 0.287  & \textbf{0.334} & \textbf{0.068} & 0.287  \\
    Index & 0.265  & -     & 0.283  & 0.265  & -     & 0.283  & 0.265  & -     & 0.283  \\
    Minimum-variance & 0.027  & 0.206  & \textbf{0.233} & 0.027  & 0.206  & \textbf{0.233 } & 0.027  & 0.206  & \textbf{0.233} \\
    \hline
\multicolumn{4}{l}{\underline{Panel C: Pacific Markets}} \\
    \multicolumn{10}{c}{Australia: $\alpha_{I,0}=0.400, \sigma=2.343$} \\
    \hline
    Backward-looking & \textbf{-0.114} & \textbf{0.325} & \textbf{0.711} & \textbf{0.061} & 0.226  & \textbf{0.462} & \textbf{0.124} & \textbf{0.320} & \textbf{0.359} \\
    Combined & -0.172  & 0.353  & \textbf{0.713} & -0.005  & 0.201  & 0.501  & 0.096  & \textbf{0.317} & 0.387  \\
    Forward-looking & -0.165  & 0.374  & \textbf{0.711} & -0.011  & \textbf{0.160} & 0.491  & 0.099  & 0.326  & 0.381  \\
    \hline
    1/N   & 0.093  & \textbf{0.047} & 0.435  & 0.093  & \textbf{0.047} & 0.435  & 0.093  & \textbf{0.047} & 0.435  \\
    Index & 0.000  & -     & 0.487  & 0.000  & -     & 0.487  & 0.000  & -     & 0.487  \\
    Minimum-variance & \textbf{0.114} & 0.333  & \textbf{0.354} & \textbf{0.114} & 0.333  & \textbf{0.354 } & \textbf{0.114} & 0.333  & \textbf{0.354} \\
    \hline
    \multicolumn{10}{c}{HongKong: $\alpha_{I,0}=0.500, \sigma=1.041$} \\
    \hline
    Backward-looking & 0.162  & 0.380  & 0.310  & 0.100  & 0.249  & \textbf{0.269} & 0.021  & 0.161  & \textbf{0.253} \\
    Combined & \textbf{0.238} & \textbf{0.362 } & \textbf{0.289} & \textbf{0.164} & 0.235  & 0.288  & \textbf{0.060} & 0.149  & \textbf{0.252} \\
    Forward-looking & 0.233  & 0.400  & 0.310  & \textbf{0.157} & \textbf{0.111} & 0.298  & \textbf{0.058} & \textbf{0.134} & 0.255  \\
    \hline
    1/N   & 0.161  & \textbf{0.030} & \textbf{0.264} & 0.161  & \textbf{0.030} & \textbf{0.264} & 0.161  & \textbf{0.030} & \textbf{0.264} \\
    Index & \textbf{0.168} & -     & 0.291  & \textbf{0.168} & -     & 0.291  & \textbf{0.168} & -     & 0.291  \\
    Minimum-variance & -0.003  & 0.115  & \textbf{0.263} & -0.003  & 0.115  & \textbf{0.263} & -0.003  & 0.115  & \textbf{0.263} \\
    \hline
    \multicolumn{10}{c}{Japan: $\alpha_{I,0}=0.500, \sigma=1.041$} \\
    \hline
    Backward-looking & 0.107  & 0.481  & 0.351  & 0.056  & 0.686  & 0.356  & -0.039  & 0.829  & 0.363  \\
    Combined & 0.183  & \textbf{0.465} & \textbf{0.333} & 0.247  & 0.397  & 0.252  & 0.189  & 0.534  & 0.258  \\
    Forward-looking & \textbf{0.244} & \textbf{0.465} & \textbf{0.335} & \textbf{0.343} & \textbf{0.272} & \textbf{0.228} & \textbf{0.257} & \textbf{0.443} & \textbf{0.230} \\
    \hline
    1/N   & \textbf{0.359} & \textbf{0.036} & \textbf{0.218} & \textbf{0.359} & \textbf{0.036} & \textbf{0.218} & \textbf{0.359} & \textbf{0.036} & \textbf{0.218} \\
    Index & \textbf{0.360} & -     & \textbf{0.225} & \textbf{0.360} & -     & \textbf{0.225} & \textbf{0.360} & -     & \textbf{0.225} \\
    Minimum-variance & 0.196  & 0.420  & 0.241  & 0.196  & 0.420  & 0.241  & 0.196  & 0.420  & 0.241  \\
    \hline
    \multicolumn{10}{c}{Singapore: $\alpha_{I,0}=0.500, \sigma=1.041$} \\
    \hline
    Backward-looking & \textbf{-0.284} & 0.450  & \textbf{0.598} & -0.178  & 0.177  & 0.495  & -0.250  & \textbf{0.357} & 0.520  \\
    Combined & \textbf{-0.280} & \textbf{0.430} & \textbf{0.597} & \textbf{-0.099} & \textbf{0.058} & \textbf{0.458} & \textbf{-0.231} & 0.366  & \textbf{0.510} \\
    Forward-looking & \textbf{-0.280} & \textbf{0.430} & \textbf{0.597} & \textbf{-0.099} & \textbf{0.058} & \textbf{0.458} & \textbf{-0.231} & 0.366  & \textbf{0.510} \\
    \hline
    1/N   & -0.399  & \textbf{0.052} & 0.614  & -0.399  & \textbf{0.052} & 0.614  & -0.399  & \textbf{0.052} & 0.614  \\
    Index & \textbf{-0.091} & -     & \textbf{0.454} & \textbf{-0.091} & -     & \textbf{0.454} & \textbf{-0.091} & -     & \textbf{0.454} \\
    Minimum-variance & -0.257  & 0.410  & 0.523  & -0.257  & 0.410  & 0.523  & -0.257  & 0.410  & 0.523  \\
    \hline
    \multicolumn{10}{c}{China: $\alpha_{I,0}=0.300, \sigma=4.165$} \\
    \hline
    Backward-looking & \textbf{0.014} & \textbf{0.395} & 0.625  & \textbf{0.115} & 0.195  & 0.466  & \textbf{0.197} & \textbf{0.241} & \textbf{0.380} \\
    Combined & 0.002  & 0.420  & \textbf{0.621} & 0.114  & \textbf{0.160} & 0.466  & 0.184  & 0.260  & \textbf{0.381} \\
    Forward-looking & -0.012  & 0.465  & 0.633  & 0.114  & 0.342  & \textbf{0.443} & \textbf{0.200} & 0.339  & 0.390  \\
    \hline
    1/N   & 0.173  & \textbf{0.054} & 0.466  & 0.173  & \textbf{0.054} & 0.466  & 0.173  & \textbf{0.054} & 0.466  \\
    Index & 0.131  & -     & 0.464  & 0.131  & -     & 0.464  & 0.131  & -     & 0.464  \\
    Minimum-variance & \textbf{0.226} & 0.241  & \textbf{0.366} & \textbf{0.226} & 0.241  & \textbf{0.366} & \textbf{0.226} & 0.241  & \textbf{0.366} \\
    \hline
    \multicolumn{10}{c}{India: $\alpha_{I,0}=0.400, \sigma=2.343$} \\
    \hline
    Backward-looking & 0.139  & 0.601  & 0.503  & 0.369  & 0.410  & 0.281  & 0.646  & 0.324  & 0.167  \\
    Combined & \textbf{0.190} & \textbf{0.560} & 0.466  & \textbf{0.388} & 0.346  & \textbf{0.267} & \textbf{0.659} & \textbf{0.307} & \textbf{0.164} \\
    Forward-looking & 0.081  & 0.588  & \textbf{0.435} & 0.261  & \textbf{0.249} & \textbf{0.274} & 0.630  & 0.353  & \textbf{0.159} \\
    \hline
    1/N   & 0.424  & \textbf{0.057} & 0.267  & 0.424  & \textbf{0.057} & 0.267  & 0.424  & \textbf{0.057} & 0.267  \\
    Index & 0.309  & -     & 0.266  & 0.309  & -     & 0.266  & 0.309  & -     & 0.266  \\
    Minimum-variance & \textbf{0.820} & 0.279  & \textbf{0.129} & \textbf{0.820} & 0.279  & \textbf{0.129} & \textbf{0.820} & 0.279  & \textbf{0.129} \\
    \hline
    \multicolumn{10}{c}{Indonesia: $\alpha_{I,0}=0.200, \sigma=6.508$} \\
    \hline
    Backward-looking & -0.151  & 0.564  & \textbf{0.539} & -0.185  & 0.190  & 0.475  & -0.329  & 0.301  & \textbf{0.613} \\
    Combined & \textbf{-0.144} & \textbf{0.529} & \textbf{0.536} & \textbf{-0.171} & \textbf{0.144} & \textbf{0.469} & \textbf{-0.323} & \textbf{0.247} & \textbf{0.609} \\
    Forward-looking & -0.151  & 0.538  & \textbf{0.539} & \textbf{-0.174} & 0.177  & \textbf{0.468} & -0.325  & 0.255  & \textbf{0.609} \\
    \hline
    1/N   & -0.265  & \textbf{0.053} & 0.546  & -0.265  & \textbf{0.053} & 0.546  & -0.265  & \textbf{0.053} & 0.546  \\
    Index & \textbf{-0.159 } & -     & \textbf{0.465} & \textbf{-0.159} & -     & \textbf{0.465} & \textbf{-0.159} & -     & \textbf{0.465} \\
    Minimum-variance & -0.384  & 0.288  & 0.658  & -0.384  & 0.288  & 0.658  & -0.384  & 0.288  & 0.658  \\
    \hline
    \multicolumn{10}{c}{Korea: $\alpha_{I,0}=0.400, \sigma=2.343$} \\
    \hline
    Backward-looking & -0.004  & \textbf{0.441} & \textbf{0.474} & \textbf{0.106} & \textbf{0.030} & \textbf{0.396} & \textbf{0.148} & \textbf{0.342} & \textbf{0.420} \\
    Combined & \textbf{0.012} & 0.451  & \textbf{0.474} & \textbf{0.106} & 0.044  & \textbf{0.396} & \textbf{0.150} & 0.347  & \textbf{0.420} \\
    Forward-looking & -0.066  & 0.597  & 0.533  & -0.011  & 0.420  & 0.508  & 0.132  & 0.430  & 0.442  \\
    \hline
    1/N   & 0.100  & \textbf{0.063} & 0.430  & 0.100  & \textbf{0.063} & 0.430  & 0.100  & \textbf{0.063 } & 0.430  \\
    Index & 0.105  & -     & \textbf{0.396} & 0.105  & -     & \textbf{0.396} & 0.105  & -     & \textbf{0.396} \\
    Minimum-variance & \textbf{0.154} & 0.389  & 0.425  & \textbf{0.154} & 0.389  & 0.425  & \textbf{0.154} & 0.389  & 0.425  \\
    \hline
    \multicolumn{10}{c}{Malaysia: $\alpha_{I,0}=0.400, \sigma=2.343$} \\
    \hline
    Backward-looking & 0.129  & 0.607  & 0.199  & -0.028  & 0.250  & 0.280  & -0.060  & 0.300  & 0.280  \\
    Combined & \textbf{0.142} & \textbf{0.497} & \textbf{0.183} & \textbf{-0.023} & 0.165  & \textbf{0.265} & \textbf{-0.055} & \textbf{0.281} & \textbf{0.266} \\
    Forward-looking & 0.001  & 0.525  & 0.271  & -0.144  & \textbf{0.131} & 0.331  & -0.127  & 0.289  & 0.306  \\
    \hline
    1/N   & \textbf{0.014} & \textbf{0.050} & 0.292  & \textbf{0.014} & \textbf{0.050} & 0.292  & \textbf{0.014} & \textbf{0.050} & 0.292  \\
    Index & -0.024  & -     & \textbf{0.271} & -0.024  & -     & \textbf{0.271} & -0.024  & -     & \textbf{0.271} \\
    Minimum-variance & -0.053  & 0.278  & \textbf{0.270} & -0.053  & 0.278  & \textbf{0.270} & -0.053  & 0.278  & \textbf{0.270} \\
    \hline
    \multicolumn{10}{c}{Philippines: $\alpha_{I,0}=0.300, \sigma=4.165$} \\
    \hline
    Backward-looking & \textbf{0.329} & \textbf{0.459} & \textbf{0.389} & \textbf{0.046} & \textbf{0.213} & \textbf{0.495} & \textbf{-0.127} & \textbf{0.230} & \textbf{0.576} \\
    Combined & 0.219  & 0.532  & 0.468  & -0.018  & 0.289  & 0.552  & -0.173  & 0.269  & 0.592  \\
    Forward-looking & 0.212  & 0.514  & 0.474  & -0.015  & 0.276  & 0.550  & -0.177  & 0.265  & 0.595  \\
    \hline
    1/N   & -0.037  & \textbf{0.050} & 0.588  & -0.037  & \textbf{0.050} & 0.588  & -0.037  & \textbf{0.050} & 0.588  \\
    Index & \textbf{0.074} & -     & \textbf{0.493} & \textbf{0.074} & -     & \textbf{0.493} & \textbf{0.074} & -     & \textbf{0.493} \\
    Minimum-variance & -0.162  & 0.243  & 0.598  & -0.162  & 0.243  & 0.598  & -0.162  & 0.243  & 0.598  \\
    \hline
    \multicolumn{10}{c}{Taiwan: $\alpha_{I,0}=0.400, \sigma=2.343$} \\
    \hline
    Backward-looking & \textbf{0.364} & \textbf{0.372} & \textbf{0.336} & \textbf{0.481} & \textbf{0.024} & 0.266  & 0.153  & \textbf{0.352} & 0.281  \\
    Combined & \textbf{0.364} & \textbf{0.372} & \textbf{0.336} & \textbf{0.481} & \textbf{0.024} & 0.266  & 0.153  & \textbf{0.352} & 0.281  \\
    Forward-looking & 0.346  & 0.682  & 0.320  & 0.424  & 0.397  & \textbf{0.255} & \textbf{0.164} & 0.429  & \textbf{0.247} \\
    \hline
    1/N   & 0.313  & \textbf{0.046} & 0.304  & 0.313  & \textbf{0.046} & 0.304  & 0.313  & \textbf{0.046} & 0.304  \\
    Index & \textbf{0.482} & -     & \textbf{0.265} & \textbf{0.482} & -     & \textbf{0.265} & \textbf{0.482} & -     & \textbf{0.265} \\
    Minimum-variance & 0.054  & 0.344  & 0.295  & 0.054  & 0.344  & 0.295  & 0.054  & 0.344  & 0.295  \\
    \hline
    \multicolumn{10}{c}{Thailand: $\alpha_{I,0}=0.400, \sigma=2.343$} \\
    \hline
    Backward-looking & \textbf{0.177} & 0.431  & \textbf{0.444} & \textbf{0.151} & \textbf{0.073} & \textbf{0.393} & \textbf{0.040} & \textbf{0.323} & \textbf{0.404} \\
    Combined & \textbf{0.179} & \textbf{0.420} & \textbf{0.444} & \textbf{0.154} & \textbf{0.067} & \textbf{0.393} & \textbf{0.038} & \textbf{0.318 } & \textbf{0.404} \\
    Forward-looking & 0.020  & 0.967  & 0.543  & 0.078  & 0.821  & 0.504  & 0.010  & 0.696  & 0.485  \\
    \hline
    1/N   & \textbf{0.232} & \textbf{0.062} & 0.396  & \textbf{0.232} & \textbf{0.062} & 0.396  & \textbf{0.232} & \textbf{0.062} & 0.396  \\
    Index & 0.163  & -     & \textbf{0.393} & 0.163  & -     & \textbf{0.393} & 0.163  & -     & \textbf{0.393} \\
    Minimum-variance & -0.016  & 0.284  & 0.419  & -0.016  & 0.284  & 0.419  & -0.016  & 0.284  & 0.419  \\
    \hline
    Times Periods & \multicolumn{9}{c}{2012.1.6-2020.11.30} \\
\hline
\end{longtable}

\bibliographystyle{ormsv080}
\bibliography{Forward_looking_paper}

\end{document}